%% file: main.tex
\documentclass{article}

\usepackage[
	persons={Nima, June, Anqi, Alireza, Fred, CJ, Carlo},
	authors=footnotes,
	bibliosources={refs.bib},
	tags={
	}
]{pomegranate}

\usepackage{makecell}
\DeclareTheorem<definition>{assumption}

\DeclareDelimiter{\Otilde}[\mathnormal{\widetilde{O}}]{\lparen}{\rparen}
\DeclareDelimiter{\Omegatilde}[\mathnormal{\widetilde{\Omega}}]{\lparen}{\rparen}

\DeclareDelimiter{\Normal}[\mathcal{N}]{\lparen}{\rparen}
\DeclareDelimiter{\Uniform}[\operatorname{Uniform}]{\lparen}{\rparen}
\DeclareOperator{\Leb}

\DeclareDelimiter{\opnorm}{\lVert}{\rVert_{\operatorname{op}}}
\DeclareDelimiter{\tr}[\operatorname{tr}]{\lparen}{\rparen}
\DeclareDelimiter{\cov}[\operatorname{cov}]{\lparen}{\rparen}
\DeclareDelimiter{\mean}[\operatorname{mean}]{\lparen}{\rparen}
\DeclareDelimiter{\law}[\operatorname{law}]{\lparen}{\rparen}
\DeclareDelimiter{\dTV}[\mathnormal{d}_{\operatorname{TV}}]{\lparen}{\rparen}
\DeclareDelimiter{\DKL}[\mathcal{D}_{\operatorname{KL}}]{\lparen}{\rparen}

\DeclareDelimiter{\I}[\mathcal{I}]{\lparen}{\rparen}
\DeclareDelimiter{\H}[\mathcal{H}]{\lparen}{\rparen}
\DeclareDelimiter{\C}[\mathcal{C}]{\lparen}{\rparen}
\DeclareDelimiter{\Ent}[\operatorname{Ent}]{\lbrack}{\rbrack}

\allowdisplaybreaks
\DeclareDocumentMathCommand{\Dtilde}{}{\widetilde{D}}
\DeclareDocumentMathCommand{\epsilonTV}{}{\epsilon_{\operatorname{TV}}}

\DeclareDelimiter{\RS}[{\operatorname{{RS}^2}}]{\lparen}{\rparen}
\DeclareDelimiter{\fallback}[\operatorname{FallbackSampler}]{\lparen}{\rparen}

\DeclareDocumentMathCommand{\epsilonTV}{}{\epsilon_{\operatorname{TV}}}

\SetKwProg{Function}{function}{}{}
\SetKwFor{ParallelFor}{for}{in parallel}{}

\usepgfplotslibrary{fillbetween}

\tikzset{samples=40}

\Tag<draft>{
	\tikzset{samples=4}
}

\pgfplotsset{
    colormap={OrangeWhite}{
        color=(White)
        color=(Orange)
    }
}

\ExplSyntaxOn
\tl_new:N \l_pom_tower_tl
\DeclareDocumentCommand{\TowerPlot}{m}{
	\tl_clear:N \l_pom_tower_tl
	\clist_set:Nn \l_tmpa_clist {#1}
	\int_step_variable:nNn {\clist_count:N \l_tmpa_clist} \l_tmpa_tl {
		\clist_set:Nx \l_tmpb_clist {\clist_item:Nn \l_tmpa_clist \l_tmpa_tl}
		\int_step_variable:nNn {\clist_count:N \l_tmpb_clist} \l_tmpb_tl {
			\tl_put_right:Nx \l_pom_tower_tl {
				(\int_eval:n {\l_tmpa_tl - 1}, \int_eval:n {\l_tmpb_tl - 1}, \clist_item:Nn \l_tmpb_clist \l_tmpb_tl)
			}
			\tl_put_right:Nx \l_pom_tower_tl {
				(\int_eval:n {\l_tmpa_tl - 1}, \int_eval:n {\l_tmpb_tl}, \clist_item:Nn \l_tmpb_clist \l_tmpb_tl)
			}
		}
		\int_step_variable:nNn {\clist_count:N \l_tmpb_clist} \l_tmpb_tl {
			\tl_put_right:Nx \l_pom_tower_tl {
				(\int_eval:n {\l_tmpa_tl}, \int_eval:n {\l_tmpb_tl - 1}, \clist_item:Nn \l_tmpb_clist \l_tmpb_tl)
			}
			\tl_put_right:Nx \l_pom_tower_tl {
				(\int_eval:n {\l_tmpa_tl}, \int_eval:n {\l_tmpb_tl}, \clist_item:Nn \l_tmpb_clist \l_tmpb_tl)
			}
		}
	}
	\use:x {\exp_not:N \addplot3[surf, mark=none, mesh/rows={\int_eval:n {2*\clist_count:N \l_tmpa_clist}}, mesh/cols={\int_eval:n {2*\clist_count:N \l_tmpb_clist}}] coordinates {\l_pom_tower_tl};}
}
\ExplSyntaxOff

\tikzset{
    declare function={
        normal(\x,\y,\s)=exp(-0.5*(pow((\x),2)+pow((\y),2))/(\s));	
    }
}

\title{Parallel Sampling via Autospeculation}
\AddTag<noanon>
\Tag<anon>{\RemoveTag<noanon>}
\Tag<noanon>{
    \date{}
    \author[1]{Nima Anari}
    \author[1]{Carlo Baronio}
    \author[2]{CJ Chen}
    \author[1]{Alireza Haqi}
    \author[3]{Frederic Koehler}
    \author[1]{Anqi Li}
    \author[4]{Thuy-Duong Vuong}
    \affil[1]{Stanford University, \url{{anari,cbaronio,ahaqi,aqli}@stanford.edu}}
    \affil[2]{University of Arizona, \url{schen9@arizona.edu}}
    \affil[3]{University of Chicago, \url{fkoehler@uchicago.edu}}
    \affil[4]{UC Berkeley, \url{tdvuong@berkeley.edu}}
}
\Tag<anon>{
    \date{}
    \author{}
}

\newcommand{\branchset}[1]{\mathcal{M}^{#1}}

\begin{document}
    \maketitle
    \begin{abstract}
        \input{abstract}	
    \end{abstract}

    \clearpage
   	
    \input{intro}
    \input{prelims}
    \input{algorithm}
    \input{improved-analysis-polylog}
    \input{new-autoregression}
    \input{corrected-improved-diffusion-polylog}
    \input{missing-proofs}
    \PrintBibliography
\end{document}

%% file: abstract.tex
We present parallel algorithms to accelerate sampling via counting in two settings: any-order autoregressive models and denoising diffusion models. An any-order autoregressive model accesses a target distribution $\mu$ on $[q]^n$ through an oracle that provides conditional marginals, while a denoising diffusion model accesses a target distribution $\mu$  on $\R^n$ through an oracle that provides conditional means under Gaussian noise. Standard sequential sampling algorithms require $\widetilde{O}(n)$ time to produce a sample from $\mu$ in either setting. We show that, by issuing oracle calls in parallel, the expected sampling time can be reduced to $\widetilde{O}(n^{1/2})$. This improves the previous $\widetilde{O}(n^{2/3})$ bound for any-order autoregressive models and yields the first parallel speedup for diffusion models in the high-accuracy regime, under the relatively mild assumption that the support of $\mu$ is bounded. 

We introduce a novel technique to obtain our results: speculative rejection sampling. This technique leverages an auxiliary ``speculative'' distribution~$\nu$ that approximates~$\mu$ to accelerate sampling. Our technique is inspired by the well-studied ``speculative decoding'' techniques popular in large language models, but differs in key ways. Firstly, we use ``autospeculation,'' namely we build the speculation $\nu$ out of the same oracle that defines~$\mu$. In contrast, speculative decoding typically requires a separate, faster, but potentially less accurate ``draft'' model $\nu$. Secondly, the key differentiating factor in our technique is that we make and accept speculations at a ``sequence'' level rather than at the level of single (or a few) steps. This last fact is key to unlocking our parallel runtime of $\Otilde{n^{1/2}}$.

%% file: intro.tex
\section{Introduction}\label{sec:intro}

Efficient sampling from complex distributions is a cornerstone of modern machine learning and theoretical computer science. Its intimate connection to the task of counting or computing partition functions \cite[see, e.g.,][]{JVV86} has been very fruitful in theory, for example, yielding efficient algorithms for approximating the volume of convex bodies \cite{DFK91} and approximating the matrix permanent \cite{JSV04}. The link between sampling and counting is also at the heart of modern techniques in machine learning, where state-of-the-art generative methods, including autoregressive models (e.g., most large language models) and diffusion models (powering, e.g., image generation), rely on reductions from sampling to counting or Bayesian inference. We revisit the sampling-via-counting paradigm, asking whether parallel computation can lead to faster algorithms compared to the standard sequential methods. We focus on two settings: discrete distributions supported on $[q]^n$, accessed via ``any-order autoregressive'' oracles, and continuous distributions supported on $\R^n$, accessed via ``denoising diffusion'' oracles, which we briefly outline below.

\paragraph{Autoregressive generation.} \Textcite{JVV86} established the algorithmic equivalence of (approximate) sampling and (approximate) counting for ``self-reducible'' problems. The most widely used form of self-reducibility concerns distributions $\mu$ on a space $[q]^n$, and their \emph{pinnings}---conditional distributions obtained by fixing values for some subset of the $n$ coordinates. See \cref{fig:pinnings-and-gaussians} for a visual depiction. Counting is the task of computing the partition functions of these pinnings, that is $\P*_{X\sim \mu}{X_S=\omega_S}$ for some subset $S\subseteq [n]$ and $\omega_S\in[q]^S$, which is computationally equivalent to computing \emph{conditional marginals}, that is $\P*_{X\sim \mu}{X_i=\omega_i\given X_S=\omega_S}$ for some $i\notin S$. This theoretical framework mirrors practical autoregressive models in machine learning \cite[see, e.g.,][]{LM11,VKK16,VSPUJLAKP17,DCLT18,YDYCSL19,BOpenAI20,SSE22}, where often a neural network oracle is trained to provide exactly these conditional marginals:
\begin{equation}\label{eq:conditional marginal def}
    (S, \omega_S, i) \mapsto \parens*{\P*_{X\sim \mu}{X_i=\omega_i\given X_S=\omega_S}\given \omega_i\in [q]}.
\end{equation}
This setting is known as the \textbf{any-order autoregressive} models \cite{SSE22}, where the algorithm is given access to the above conditional marginals for any $(S,\omega_S,i ,\omega_i).$\footnote{
Certain autoregression models require $S$ to be the prefix $[i-1]$; no parallel speedup is possible in this setting (see \cite{AGR24}).}

Looking through the lens of Bayesian inference, one can view this oracle as giving information about a posterior. Imagine first sampling $\omega\sim \mu$ (as the prior) and then revealing the noisy information $\omega_S$ (as the observation). This oracle provides the marginals of the posterior distribution. For this reason, we call this oracle the \textbf{coordinate denoiser}.

Given a coordinate denoiser, sampling from $\mu$ is straightforward via \cref{alg:sequential-autoregressive}.

\begin{Algorithm}[H]
    $X\gets \emptyset$\;
    \For{$i=1,\dots,n$}{
        \For{$\omega \in [q]$}{
            $p_\omega \gets \P*_{Y\sim \mu}{Y_i=\omega \given Y_{[i-1]}=X_{[i-1]}}$\;
        }
        sample $X_i\sim (p_1,\dots,p_q)$\;
    }
    \Return{$X$}\;
    \caption{Autoregressive sampling \label{alg:sequential-autoregressive}}
\end{Algorithm}

Note that the ordering of $1,2,\dots,n$ on the coordinates is arbitrary, and any other ordering would work too. For example, if we choose a uniformly random permutation, we get a process that is basically what \textcite{CE22} named \emph{coordinate localization}. Despite its simplicity, this algorithm inherently requires $n$ sequential steps, which motivates the search for parallelizable algorithms in this work.

\paragraph{Denoising diffusion.} For continuous distributions $\mu$ on $\R^n$, one can define an alternative to pinnings, by considering distributions $\nu$ obtained from multiplying a Gaussian density with $\mu$:
\begin{equation}\label{eq:conditional mean def}
   \dd{\nu}{\mu}(x)\propto \exp\parens*{\frac{-t\norm{x-\omega}^2}{2}}  
\end{equation}

for some $t\in \R_{\geq 0}$ and $\omega\in \R^n$. See \cref{fig:pinnings-and-gaussians} for a visual depiction. The task of computing conditional marginals is now replaced by the task of \emph{computing the mean} of $\nu$. The corresponding oracle, which can be called a ``denoising diffusion oracle'',  provides the following mapping:
\[ (t, \omega) \mapsto \mean{\nu}=\E_{X\sim \nu}{X}=\E*_{X\sim \mu, g\sim \Normal{0, I/t}}{X\given X+g=\omega}. \]
Again, looking through the lens of Bayesian inference, this oracle is providing information about a posterior. If we sample $X\sim \mu$ (as the prior) and reveal the noisy information $X+g$ (as the observation) where $g\sim \Normal{0, I/t}$ is independent Gaussian noise, then this oracle is giving us the mean of the posterior distribution. For this reason, we call this oracle the \textbf{Gaussian denoiser}.

While less obvious than in the autoregressive case, a Gaussian denoiser enables approximate sampling from $\mu$ via a process known as \emph{stochastic localization} \cite{Eld13}, equivalent \cite{Mon23} to \emph{denoising diffusion} in machine learning \cite[see, e.g.,][]{SWMG15,HJA20,SM19,SSKKEP20}. This involves simulating a stochastic differential equation (SDE) driving the process $\parens{X_t\given t\in \R_{\geq 0}}$ defined via 
\[ \d{X_t}=\E*_{Y \sim \mu, g\sim \Normal{0, tI}}{Y \given tY+g=X_t}\d{t}+\d{B_t}, \]
and initial condition $X_0=0$, where $\parens*{B_t\given t\in \R_{\geq 0}}$ is standard Brownian motion in $\R^n$. The \emph{drift term} in this equation is exactly the output of the Gaussian denoiser. Faithfully simulating stochastic localization until a large time $t$ gives us (approximate) samples since $X_t/t$ will be distributed as $\mu*\Normal{0,I/t}$ \cite[see, e.g.,][]{EM22}. To get an algorithm, however, one needs to resort to an approximate SDE-solving method, the simplest of which is Euler-Maruyama, which proceeds by discretizing time into $0=t_0<t_1<\dots<t_N$ and proceeding as in \cref{alg:sequential-SL}.

\begin{Algorithm}[H]
    $X_{t_0}\gets 0$\;
    \For{$i=0,\dots,N-1$}{
        sample $g\sim \Normal{0, (t_{i+1}-t_i)I}$\;
        $X_{t_{i+1}}\gets X_{t_i}+(t_{i+1}-t_i)\cdot \E*_{\omega \sim \mu, g'\sim \Normal{0, t_iI}}{\omega \given t_i\omega+g'=X_{t_i}}+g$\;
    }
    \Return{$X_{t_N}/{t_N}$}\;
    \caption{Euler-Maruyama discretization of the stochastic localization SDE \label{alg:sequential-SL}}
\end{Algorithm}

\Textcite{BDDD23}, and in parallel \textcite{CDG23}, showed that to sample $\epsilon^2$-closely in KL divergence from $\mu*\Normal*{0,\delta \cdot I},$ one can use a schedule with $N=\Otilde*{n\cdot \frac{\log^2(1/\delta)}{\epsilon^2}}$, assuming the normalizing assumption that $\tr{\cov{\mu}}=n$. However, similar to the autoregressive case, this standard algorithm is sequential, again raising the question of potential parallel speedups.

\begin{figure}
\begin{Columns}
\Column
\Tikz*[scale=0.3]{
    \begin{scope}[shift={(-2,0)}]
        \begin{axis}[hide axis, unbounded coords = jump, view={120}{35}]
            \TowerPlot{{0,0,0,0},{0,5,3,0},{0,2,1,0},{0,0,0,0}}
        \end{axis}
    \end{scope}
    \begin{scope}[shift={(12,4)}]
        \begin{axis}[hide axis, unbounded coords = jump, view={120}{35}]
            \TowerPlot{{0,0,0,0},{0,5,0,0},{0,2,0,0},{0,0,0,0}}
        \end{axis}
    \end{scope}
    \begin{scope}[shift={(12,0)}]
        \begin{axis}[hide axis, unbounded coords = jump, view={120}{35}]
            \TowerPlot{{0,0,0,0},{0,0,3,0},{0,0,1,0},{0,0,0,0}}
        \end{axis}
    \end{scope}
    \begin{scope}[shift={(12,-4)}]
        \begin{axis}[hide axis, unbounded coords = jump, view={120}{35}]
            \TowerPlot{{0,0,0,0},{0,0,0,0},{0,2,1,0},{0,0,0,0}}
        \end{axis}
    \end{scope}
    \draw[decorate, line width=1, decoration={brace, amplitude=10}] (12, -4) -- (12, 8);
    \draw[line width=1, -stealth] (5, 2) -- (11, 2);
    \node at (1, -1) {$\mu$};
    \node at (15.5, -5) {possible $\nu$};
    \node at (3, -4) {$\nu = \text{pinning of }\mu$};
}
\Column
\Tikz*[scale=0.3]{
    \begin{scope}[shift={(-2,0)}]
        \begin{axis}[hide axis, domain=-3:3, domain y=-3:3]
            \addplot3[surf] {10*normal(x, y, 2)};
        \end{axis}
    \end{scope}
    \begin{scope}[shift={(12,4)}]
        \begin{axis}[hide axis, domain=-3:3, domain y=-3:3]
            \addplot3[surf] {10*normal(x-1, y-1, 1)};
        \end{axis}
    \end{scope}
    \begin{scope}[shift={(12,0)}]
        \begin{axis}[hide axis, domain=-3:3, domain y=-3:3]
            \addplot3[surf] {10*normal(x-1, y+1, 1)};
        \end{axis}
    \end{scope}
    \begin{scope}[shift={(12,-4)}]
        \begin{axis}[hide axis, domain=-3:3, domain y=-3:3]
            \addplot3[surf] {10*normal(x+1, y+1, 1)};
        \end{axis}
    \end{scope}
    \draw[decorate, line width=1, decoration={brace, amplitude=10}] (12, -4) -- (12, 8);
    \draw[line width=1, -stealth] (5, 2) -- (11, 2);
    \node at (1, -1) {$\mu$};
    \node at (15.5, -5) {possible $\nu$};
    \node at (3, -4) {$\d{\nu} \propto  \text{Gaussian} \cdot \d{\mu}$};
}
\end{Columns}
\caption{\label{fig:pinnings-and-gaussians}In autoregressive models, the oracle returns means --- that is marginal probabilities --- of pinned distributions. In contrast, in diffusion models, the oracle returns means of distributions $\nu$ whose density w.r.t.\ $\mu$ is a Gaussian.}
\end{figure}
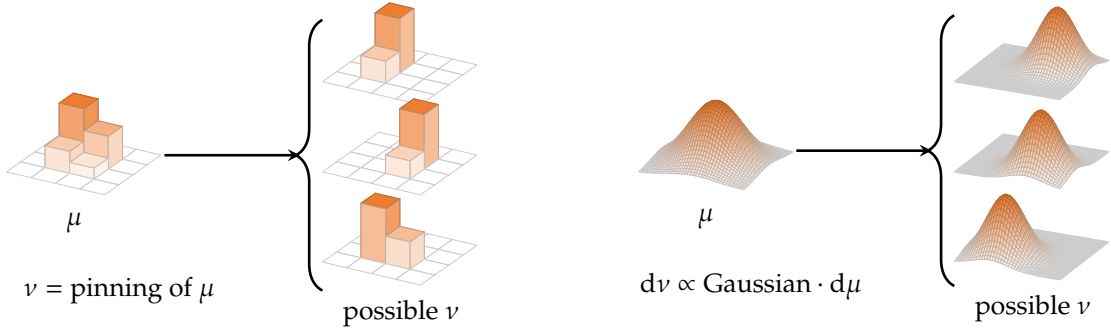

To summarize, given either a \emph{coordinate denoiser} or a \emph{Gaussian denoiser} for an $n$-dimensional distribution $\mu$, standard algorithms allow sampling from $\mu$ (approximately in the denoising diffusion case) using $\Otilde{n}$ oracle calls. However, these calls are inherently sequential, leaving the massive potential of parallel hardware---common in practical implementations---untapped. This motivates our central question:
\begin{quote}
    Suppose we interact with these oracles in rounds, sending polynomially many queries in each round. How many rounds suffice to produce a sample?
\end{quote}
This question is motivated by the fact that in practice, where neural networks play the role of these oracles, the oracles are run on hardware like GPUs capable of massive parallelism. There are also compelling theoretical motivations. For problems such as sampling uniformly random spanning trees, planar perfect matchings, or Eulerian tours,  the underlying counting oracles (i.e., denoisers) can be implemented efficiently in parallel (as \Class{NC} algorithms), that is they run in $\poly\log(n)$ time given $\poly(n)$ processors \cite[see][]{AHLVXY23}.\footnote{This further extends to problems that admit deterministic (approximate) counting algorithms based on two key approaches: decay of correlation and tree recursion of \Textcite{Wei06} and polynomial interpolation of \Textcite{Bar16}. See discussions in \cite{AHSS20} for more details.} It is natural to ask if sampling itself can be achieved in $\poly\log(n)$ time for these problems. For planar perfect matchings, a case where parallel counting is known but parallel sampling remains an open question, our results provide a direct improvement when combined with the analysis of \textcite{AGR24}.

Recently, \textcite{AGR24} showed that in the \emph{any-order autoregressive case}, there are more parallelizable sampling algorithms. Namely, they showed that sampling can be achieved in $O(n^{2/3}\cdot \poly\log(n, q))$ rounds. They complemented this with a lower bound showing $\Omegatilde{n^{1/3}}$ rounds are necessary for some distributions. Our first contribution significantly improves their upper bound:
\begin{theorem}[Informal, main result for any-order autoregressive oracles]\label{thm:main-autoregressive-informal}
    Given access to the coordinate denoiser \eqref{eq:conditional marginal def} for a distribution $\mu$ supported on $[q]^n$, one can sample from $\mu$ using $O\parens*{n^{1/2} \sqrt{\log q} \log^3 n}$ rounds and $O(q n \log n)$ oracle queries\footnote{If we were to use the convention that a query returns the full marginal distribution of one $X_i$ conditioned on a subset $X_S$ of other coordinates (in other words the oracle returns $q$ probability values), then the number of oracle calls our algorithm would make would shrink by a factor of $q$, i.e., we would have $O(n\log n)$ queries in expectation. This is because we always query the \emph{full} marginal distribution of some variable, so our oracle calls come in groups of $q$.} in expectation.
\end{theorem}

\begin{Algorithm}
$\sigma \gets\text{sample uniformly random permutation of } [n]$\;
$\rho\gets \frac{1}{\log_2(n)}$\;
$Z_{[n]}:=(Z_{\sigma(1)}, \cdots, Z_{\sigma(n)} )\gets \RS{[n], Z_{<[n]} =\emptyset}$ \; 
\Return{$(Z_1, \cdots, Z_n)$}\;
\;
\Function{$\RS{S = \set{a+1,\cdots, b},Z_{<S} = ( Z_{\sigma(1)}, \cdots, Z_{\sigma(a)})}$}{
    $\Blue{X} : = (\Blue{X_{\sigma(a+1)}}, \cdots, \Blue{X_{\sigma(b)}} ) \gets\text{sample from }\nu_S\parens*{\cdot\given Z_{<S}} : = \bigotimes_{r=a+1}^b \pi_{\sigma(r)}(\cdot | Z_{\sigma(1)}, \cdots, Z_{\sigma(b)} )$\;
    Compute
    $\dd{\mu_S}{\nu_S}\parens*{\Blue{X}\given Z_{<S}} \gets \prod_{s=a+1}^b \frac{\pi_{\sigma(s)}(\Blue{X}_{\sigma(s)} | Z_{\sigma(1)}, \cdots ,  Z_{\sigma(a)}, \Blue{X}_{\sigma(a+1)}, \cdots, \Blue{X}_{\sigma(s-1)} )}{ \pi_{\sigma(r)}(\Blue{X}_{\sigma(s)} | Z_{\sigma(1)}, \cdots, Z_{\sigma(a)} ) } $ \;
    $u\gets \min\set*{1,\dd{\mu_S}{\nu_S}\parens*{\Blue{X}\given Z_{<S}} }$\;
    \eIf{coin flip with bias $u$ comes up heads}{
        \Return{$\Blue{X}$} \tcp{Note: if $b\leq a+1$ then $u=1$ and the algorithm terminates here} \;
    }{
    	\For{$r=0, 1, 2, \dots$}{
    		\ParallelFor{$i=  \lceil (1+\rho)^r\rceil,\lceil (1+\rho)^r\rceil+2, \dots, \lceil (1+\rho)^{r+1}\rceil -1$}{
    			$\Orange{Y^{(i)}} = (\Orange{Y^{(i)}_{\sigma(a+1)}}, \cdots, \Orange{Y^{(i)}_{\sigma(b)}})\gets \fallback{S, Z_{<S}}$\;
                Compute
    $\dd{\nu_S}{\mu_S}\parens*{\Orange{Y^{(i)}}\given Z_{<S}} \gets \prod_{s=a+1}^b \frac{ \pi_{\sigma(s)}(\Orange{Y^{(i)}_{\sigma(s)}} | Z_{\sigma(1)}, \cdots, Z_{\sigma(a)} ) }{\pi_{\sigma(s)}(\Orange{Y^{(i)}_{\sigma(s)}} | Z_{\sigma(1)}, \cdots ,  Z_{\sigma(a)}, \Orange{Y^{(i)}_{\sigma(a+1)}}, \cdots,\Orange{Y^{(i)}_{\sigma(s-1)}} )} $ \;
    			$p^{(i)}\gets \max\set*{0, 1-\dd{\nu_S}{\mu_S}\parens*{\Orange{Y^{(i)}}\given Z_{<S}}}$\;
    			$c^{(i)}\gets \text{coin flip with bias }p^{(i)}$\;
    		}
    		\If{any $c^{(i)}$ has resulted in heads}{
    			$i^*\gets \text{index of first head}$\;
    			\Return{$\Orange{Y^{(i^*)}}$}\;
    		}
    	}
    }
}\;

\Function{$\fallback{S = \set{a+1,\cdots, b}, Z_{<S} = ( Z_{\sigma(1)}, \cdots, Z_{\sigma(a)})}$}{
    $T_1 \gets \set{a+1, \cdots, \lfloor \frac{a+b}{2}\rfloor }, T_2 \gets \set{ \lfloor \frac{a+b}{2}\rfloor+1 , \cdots , b}$\;
    $Z\gets Z_{<S}$\;
    \For{$i=1,2$}{
        $Z_{T_i}\gets \RS{T_i, Z}$\;
        append $Z_{T_i}$ to the end of $Z$\;
    }
    \Return{$Z_S$}\;
}
\caption{Recursive coordinate denoising scheme for $\pi:[q]^n \to \R_{\geq 0}.$\label{alg:main algo coordinate}}
\end{Algorithm}

\begin{Algorithm}
 $L \gets \lceil\log_2 \frac{R}{\delta} \rceil +1$, $N\gets  2^{\lceil \log_2  \max\set{\frac{R^2}{\delta}, \frac{R^4}{\delta^2}, \frac{n}{\epsilonTV^2}} \rceil }$ \;
 $\mathbf{T}_0 \gets 0,\mathbf{T}_1 \gets R^{-1} 2^0, \cdots, \mathbf{T}_{L-2} \gets R^{-1} 2^{(L-1)}, \mathbf{T}_L \gets 1/\delta$\; 
 \For{$i= 0, \cdots, L-1$}{
 \For{$j =0,\cdots, N$} {
 $\delta_i \gets  \frac{\mathbf{T}_{i+1} - \mathbf{T}_i}{N}$\;
 $t_{(i-1)N + j} \gets  \mathbf{T}_i + j\times \delta_i $ }}\;
 $\rho\gets \frac{1}{\log_2(N)}$\;
$Z_{[NL]} = (Z_1, \cdots, Z_{NL}) \gets \fallback{[NL], Z_{<[NL]} =\emptyset, L}$ \; 
\Return{$\delta(Z_1 + \cdots + Z_{NL})$}\;
\;
\Function{$\RS{S = \set{a+1,\cdots, b},Z_{<S} =(U_{t_{1}} - U_{t_{0}} , \cdots ,  U_{t_{a}} - U_{t_{a-1}} ) }$}{
    $\Blue{X} := (U_{t_{a+1}} - U_{t_{a}} , \cdots ,  U_{t_{b}} - U_{t_{b-1}} )    \gets\text{sample from }\nu_S\parens*{\cdot\given Z_{<S}} : = \bigotimes_{r=a+1}^b \Normal{(t_r - t_{r-1}) f(t_a, U_{t_a}), t_r -t_{r-1} } $\;
    Compute $U_{t_s} \gets Z_1 + Z_2 + \cdots + Z_a + \Blue{X_1} + \cdots + \Blue{X_{s-a}}$ for $s\in \set{a,\cdots, b}$\;
    Compute
    $\dd{\mu_S}{\nu_S}\parens*{\Blue{X}\given Z_{<S}} \gets \prod_{s=a+1}^b \frac{\Normal{ (t_s-t_{s-1})f (t_{s-1}, U_{t_{s-1}}), t_s-t_{s-1}  }}{\Normal{(t_s - t_{s-1}) f(t_a, U_{t_a}), t_s -t_{s-1} } } $ \;
    $u\gets \min\set*{1,\dd{\mu_S}{\nu_S}\parens*{\Blue{X}\given Z_{<S}} }$\;
    \eIf{coin flip with bias $u$ comes up heads}{
        \Return{$\Blue{X}$} \tcp{Note: if $b\leq a+1$ then $u=1$ and the algorithm terminates here} \;
    }{
    	\For{$r=0, 1, 2, \dots$}{
    		\ParallelFor{$i=  \lceil (1+\rho)^r\rceil,\lceil (1+\rho)^r\rceil+2, \dots, \lceil (1+\rho)^{r+1}\rceil -1$}{
    			$\Orange{Y^{(i)}} := (U^{(i)}_{t_{a+1}} -U^{(i)}_{t_{a}}, \cdots, U^{(i)}_{t_{b}} - U^{(i)}_{t_{b-1}}) \gets \fallback{S, Z_{<S},2}$\;
                Compute $U_{t_s}^{(i)} \gets Z_1 + Z_2 + \cdots + Z_a + \Orange{Y^{(i)}_1} + \cdots + \Orange{Y^{(i)}_{s-a}}$ for $s\in \set{a,\cdots, b}$\;
                Compute
    $\dd{\nu_S}{\mu_S}\parens*{\Orange{Y^{(i)}}\given Z_{<S}} \gets \prod_{s=a+1}^b \frac{\Normal{(t_s - t_{s-1}) f(t_a, U^{(i)}_{t_a}), t_s -t_{s-1} } }{\Normal{ (t_s-t_{s-1})f (t_{s-1}, U^{(i)}_{t_{s-1}}), t_s-t_{s-1}  }}$ \;
    			$p^{(i)}\gets \max\set*{0, 1-\dd{\nu_S}{\mu_S}\parens*{\Orange{Y^{(i)}}\given Z_{<S}}}$\;
    			$c^{(i)}\gets \text{coin flip with bias }p^{(i)}$\;
    		}
    		\If{any $c^{(i)}$ has resulted in heads}{
    			$i^*\gets \text{index of first head}$\;
    			\Return{$\Orange{Y^{(i^*)}}$}\;
    		}
    	}
    }
}\;

\Function{$\fallback{S = \set{a+1,\cdots, b}, Z_{<S}, k} $}{
    $T_1 \gets \set{a+1, \cdots,  a+ \lfloor \frac{b-a}{k}\rfloor }, T_2 \gets \set{ a+ \lfloor \frac{b-a}{k}\rfloor + 1 , \cdots ,  a+ 2 \lfloor \frac{b-a}{k}\rfloor }, \cdots , T_k \gets \set{ a+ (k-1)\lfloor \frac{b-a}{k}\rfloor + 1 , \cdots ,  b }$\;
    $Z\gets Z_{<S}$\;
    \For{$i=1,2, \cdots, k$}{
        $Z_{T_i}\gets \RS{T_i, Z}$\;
        append $Z_{T_i}$ to the end of $Z$\;
    }
    \Return{$Z_S$}\;
}
\caption{Recursive Gaussian denoising scheme for $\mu:\R^n \to \R_{\geq 0}$ with accuracy parameters $ \delta, \epsilonTV>0,$ assuming $\sup\set*{\norm{x- \E_\mu{X}} \given x\in \supp(\mu)}\leq R$, given Gaussian denoiser oracle $f(t,x) =\E*_{Y \sim \mu, g\sim \Normal{0, tI}}{Y \given tY+g=x}$ \label{alg:main algo gaussian}}
\end{Algorithm}

The pseudocode is given in \cref{alg:main algo coordinate}.
For a more precise statement, see \cref{thm:autoregression main detail} and \cref{remark:density vs. sampling query}.  
A high-probability bound on the number of rounds can also be obtained, with a slightly worse dependence on $n$, for example, $n^{1/2 + o(1)}$ instead of $\tilde{O}(n^{1/2})$ (see \cref{thm:autoregression with high probability} and \cref{remark:autoregression with high probability parameter choice}).
Beyond the oracle calls, the internal computations performed by our algorithm can also be efficiently parallelized; when implemented on a standard parallel model (\Class{PRAM}), the overall parallel runtime matches the round complexity up to polylogarithmic factors and the number of operations match the query complexity (see \cref{remark:implementation details coordinate denoiser}). Furthermore, our result is robust to noise in the oracle (see \cref{remark:approximate autoregression oracle}).

As our second contribution, we show that under a mild bounded-support assumption on the target distribution~$\mu$, denoising diffusion admits a substantial parallel speedup, achieving a runtime that scales as $n^{1/2},$ up to polylogarithmic factors in the accuracy parameters.
\begin{theorem}[Informal, main result for denoising diffusion oracles]\label{thm:main-diffusion-informal}
    Given access to the Gaussian denoiser \eqref{eq:conditional mean def} for a distribution $\mu$ with $\sup\set*{\norm{x- \E_\mu{X}}\given x\in \supp(\mu)} \leq R,$ and accuracy parameters $\epsilonTV, \delta>0,$ one can sample from a distribution that is $\epsilonTV$-close to $\mu\ast \Normal*{0, \delta\cdot I}$ in total variation distance in the following number of rounds in expectation: \[O\parens*{n^{1/2}\cdot \poly\log(R, n, \delta^{-1}, \epsilonTV^{-1})}.\]
\end{theorem}
The pseudocode is given in \cref{alg:main algo gaussian}. For further details, see \cref{thm:full denoising polylog}. Our algorithm can be implemented in the PRAM model where the overall parallel runtime matches the round complexity up to polylogarithmic factors (see \cref{remark:implementation details gaussian denoiser}) and can also handle error in the Gaussian denoiser (see \cref{remark:approximte gaussian denoiser}).

Very recently, \cite{hu2025diffusionmodelssecretlyexchangeable} showed a parallel algorithm that needed $\Otilde*{n^{2/3}/{\epsilonTV^{4/3}}}$ rounds to sample within $\epsilonTV$ total variation distance of $\mu\ast \Normal*{0, \delta\cdot I}$, under the normalizing assumption $ \tr{\cov{\mu}} = O(n)$.  Their runtime scales as $n^{2/3}$ which is worse than our $n^{1/2}$ bound.  Notably, in the high-accuracy regime $\epsilonTV = 1/\poly(n)$—which is the standard setting for approximate sampling—this runtime is not even necessarily sublinear in $n$. In contrast, our algorithm runs in $\Otilde{n^{1/2}}$ rounds, for any accuracy parameters $\epsilonTV, \delta$ that are $2^{- \poly\log(n)},$ assuming $R = 2^{\poly\log (n)}$. Our result offers the first parallel speedup for denoising diffusion in the high-accuracy regime without strong structural assumptions i.e., smooth score function along the Ornstein-Uhlenbeck process \cite{GCC24,CRYR24}, and addresses an open question posed by \textcite{GCC24} in the parallel setting.

\subsection{Related work}

Parallelizing sampling algorithms has been a subject of interest both in theoretical computer science and in machine learning and AI. In theoretical computer science, the interest can be traced back to at least \textcite{MVV87} who gave a parallel algorithm to \emph{find} perfect matchings and asked if one can also \emph{sample} a uniformly random one in parallel; this question remains open, even on planar graphs where parallel counting (\Class{NC}) is possible \cite[see, e.g.,][]{AGR24}. Our results, specifically \cref{thm:main-autoregressive-informal}, improve the state of the art for sampling planar perfect matchings from the previous parallel runtime of $\Otilde{n^{1/3}}$ \cite{AGR24} to $\tilde{O}(n^{1/4})$ by just directly plugging into the analysis done by \textcite{AGR24}.

Perhaps motivated by the generality of algorithms used to tackle sampling, e.g., general-purpose Markov chains such as Glauber dynamics or chains based on Langevin dynamics, sampling has been extensively studied in an \emph{oracle setup}, where one can access a distribution by asking questions from an oracle. This is also the setting we adopt.

Classically, sampling via oracles has been studied in settings where the oracle can answer \emph{local queries}. For example, for distributions with density $\mu$ on $\R^n$, this might involve mapping an input point $x\in \R^n$ to $\log \mu(x)$ or $\nabla \log \mu(x)$, or for discrete distributions $\mu$ on a space like $[q]^n$ this might involve mapping $x\in [q]^n$ to something proportional to $\mu(x)$. Local oracles are often readily available for explicit distributions, and they are what standard Markov chains like Glauber dynamics or chains based on Langevin dynamics need.

Parallelizing Markov chains has been studied by several works. \Textcite{Ten95} showed how to parallelize Markov chains with at most $\poly(n)$ states and $\poly(n)$ mixing time to run in $\poly\log(n)$ time. Parallelizing more useful Markov chains such as Glauber dynamics --- which have an exponentially large state space --- to run in $\poly\log(n)$ time was shown to be possible under certain structural assumptions on the underlying distribution $\mu$ \cite{FHY21,LY22}. Similarly, on the continuous side, Markov chains that sample from distributions $\mu$ on $\R^n$ have been parallelized under certain structural assumptions such as log-concavity or isoperimetric inequalities satisfied by $\mu$ \cite{SL19,ACV24}. These assumptions on $\mu$ are natural as they are needed to obtain even polynomial-time sequential sampling algorithms via local oracles, let alone parallel sampling algorithms.

In this work, our focus is on sampling with little or no assumptions on $\mu$. While this is impossible via \emph{local oracles}, at least polynomial-time sequential sampling is possible via the following \emph{global oracles}, as discussed before: the any-order autoregressive oracle and the denoising diffusion oracle. For the former, the standard algorithm with $n$ steps is folklore, while for the latter, the algorithm is based on an appropriate discretization  \cite{BDDD23} of the stochastic localization of \textcite{Eld13}.

In the any-order autoregressive case, the only work we are aware of that obtains theoretically guaranteed parallel speedups is \cite{AGR24}, which gives a $O(n^{2/3}\poly\log(n, q))$ time algorithm. 

For the denoising diffusion case, very recently, \cite{hu2025diffusionmodelssecretlyexchangeable} gives a $\Otilde{{n^{2/3}}/{\epsilonTV^{4/3}}}$-time parallel algorithm to sample within $\epsilonTV$ total variation distance of $ \mu\ast \Normal*{0, \delta\cdot I}$, under the minimal assumption that $\tr{\cov{\mu}}\leq O(n).$
With much stronger structural assumptions\footnote{However, these assumptions are still milder than those needed for parallelizing Markov chains}, $\poly\log(n)$ time sampling via a denoising diffusion oracle was shown to be possible in two cases: for semi-log-concave distributions \cite{AHLVXY23,ACV24}, and for distributions with a smooth score function along the Ornstein-Uhlenbeck process \cite{GCC24, CRYR24}. For instance, when the score function along the Ornstein-Uhlenbeck process is $L$-smooth, \cite{GCC24} shows a parallel algorithm that approximate samples within $\epsilon_{TV}$ total variation distance in $\Otilde{L \cdot \poly\log(L n/\epsilon_{TV})}$ time. It is worth noting, however, that in general $L$ can be very large; for instance, converting a distribution supported on a ball of radius $R$ into one with an $L$-smooth score typically yields $L \approx R^2$.

A recent line of work investigates the ODE version of denoising diffusion (obtained from the corresponding SDE), where the goal is to achieve a sequential runtime with better dependency on the dimension $n$, for example, sublinear in $n$. However, these results either have super-linear dependence on $n$ or poor dependence on other parameters \cite{beyler2025convergence,li2024sharp}, or else require strong structural assumptions \cite{chen2023probability,gao2024convergence,kremling2025non}. Among these, the current state-of-the-art is due to \cite{chen2023probability}, who show that if the score function along the Ornstein–Uhlenbeck process is $L$-smooth, then an ODE-based algorithm can approximately sample within $\epsilon_{TV}$ total variation distance in $\Otilde{L^2 \sqrt{n}/\epsilon_{TV}}$ sequential steps. In general, the smoothness parameter $L$ and the inverse accuracy parameter $1/\epsilon_{TV}$ may be large e.g. $\poly(n)$, so this runtime is not necessarily sublinear in $n$.

There have been numerous proposals with surprising success to speed up sampling in generative AI models using parallelization. In diffusion models, Picard iterations were observed to yield speedups in practice \cite{SBESA23}. Picard iteration is a technique that approximately solves an ODE or an SDE by iterating an operation whose fixed point is the solution trajectory, and has been the basis of many theoretical works as well \cite{SL19, ACV24, GCC24, CRYR24}. On the other side, for autoregressive models, a different widely applied parallelization technique, known as \emph{speculative decoding} \cite{SSU18,CBILSJ23} is the dominant approach in practice. Speculative decoding sequentially uses a faster but less accurate oracle, called the \emph{draft model}, to produce samples and ``verifies'' these samples using the accurate but slower main oracle. The theoretical work of \cite{AGR24} for any-order autoregressive models is also based on a similar draft-and-verify approach, and we use the same high-level approach in this paper. As a remark, unlike the practical setup of speculative decoding, we do not use tiers of oracles; rather we show that an oracle can be its own draft model, or in other words, an oracle can \emph{autospeculate}.

\subsection{Techniques}
\label{sec:techniques}

Our main algorithm borrows, but crucially alters, elements of a parallelization technique known as speculative decoding \cite{SSU18,CBILSJ23}, widely used in practice for autoregressive models such as large language models. At a high level, in speculative decoding, to sample from a distribution $\mu$ on $[q]^n$, a separate less accurate but faster autoregressive oracle --- imagine for a different distribution $\nu$ --- is used to generate samples while the main oracle verifies these samples. In more detail, the oracle for $\nu$ is used to sequentially generate $X_1,\dots,X_n$. These are speculations for a sample from $\mu$. One then computes for all $i$ in parallel $\P_{X\sim \mu}{X_i\given X_{[i-1]}}$ and based on these probabilities we decide to accept each $X_i$ or reject it. The algorithm proceeds up to the first rejected $X_i$, samples according to a corrected marginal for coordinate $i$, and the process is restarted after pinning $X_1,\dots,X_i$. This repeats until all coordinates have been generated.

\Textcite{AGR24} showed that for \emph{any-order} autoregressive models, instead of using a separate oracle in speculative decoding, one can use the same oracle to produce speculations. More precisely, if $\nu$ is defined to be the product distribution with the same marginals as $\mu$, then it is easy to sample in parallel from $\nu$, and these samples can then be used as speculations verified by the oracle in a parallel manner. For details of the algorithm, see \cite{AGR24}. While resulting in a theoretically guaranteed speedup, \textcite{AGR24} also showed a barrier: they showed their analysis for the resulting speculation-based algorithm, a runtime of $O(n^{2/3}\cdot \poly\log(n, q))$, is essentially tight; the algorithm had to be changed to obtain an improvement.

Our main observation is that verifying samples \emph{iteratively} as in the original practical works on speculative decoding \cite{SSU18,CBILSJ23}, and in the theoretical algorithm of \cite{AGR24} is wasteful. To illustrate, imagine two product distributions $\mu=\operatorname{Ber}(1/2)^{\otimes n}$ and $\nu=\operatorname{Ber}(1/2+\epsilon)^{\otimes n}$. If we sample $(X_1,\dots,X_n)\sim \nu$, then each $X_i$ is accepted with probability $\simeq 1-\Theta(\epsilon)$. So overall, we are unlikely to accept the entire set of $X_1,\dots,X_n$ unless $\epsilon\ll 1/n$. However, this is far from an optimal coupling between $\mu$ and $\nu$. Indeed we can see that $\dTV{\mu, \nu}\leq \sqrt{\frac{1}{2}\DKL{\mu\river \nu}}=O(\sqrt{n\epsilon^2})=O(\sqrt{n}\cdot \epsilon)$. So in principle, we should be able to couple $\mu$ and $\nu$ in a way that most samples from $\nu$ are accepted as samples from $\mu$, as long as $\epsilon\ll 1/\sqrt{n}$. The suboptimality of the iterative acceptance procedure is at the heart of the barrier faced by prior methods.

To overcome the barrier faced by prior works, we introduce a new procedure that we call \emph{speculative rejection sampling}. See \cref{fig:speculative-rejection} for a visual. Imagine $\mu$ and $\nu$ are distributions on an arbitrary space, and $\nu$ is speculatively close to $\mu$, in total variation distance. Assuming that we can sample from $\mu$, $\nu$ and also compute $\dd{\mu}{\nu}$ at any desired point, speculative rejection sampling produces a sample $\sim \mu$ as follows:
\begin{Algorithm}[H]
    sample $\Navy{x}\sim \nu$\;
    accept and return $\Navy{x}$ with probability $\min\set*{1, \dd{\mu}{\nu}(\Navy{x})}$\;
    \While{not accepted}{
        sample $\Orange{y}\sim \mu$\;
        accept and return $\Orange{y}$ with probability $\max\set*{0, 1-\dd{\nu}{\mu}(\Orange{y})}$\;
    }
    \caption{Speculative rejection sampling\label{alg:srs}}
\end{Algorithm}
It can be easily shown that the output of this process is a faithful sample from $\mu$, see \cref{prop:srs-is-correct}. Our assumption is that directly sampling $\sim \mu$ is expensive, in the sense that it needs a sequential algorithm, but sampling $\sim \nu$ or computing $\dd{\mu}{\nu}$ is easy, in the sense that it can be efficiently parallelized. So the goal of speculative rejection sampling is to minimize the need for directly sampling $\Orange{y}\sim \mu$. The first sample, $\Navy{x}$, in this algorithm, is accepted with probability $1-\dTV{\mu, \nu}$ and if that happens we never need to sample $\Orange{y}\sim \mu$. But unfortunately if $\Navy{x}$ fails, each $\Orange{y}$ is only accepted with probability $\dTV{\mu, \nu}$, so on average the while loop takes $1/\dTV{\mu, \nu}$ many iterations. All in all, speculative rejection sampling needs on average $\dTV{\mu, \nu}/\dTV{\mu, \nu}=1$ direct samples $\Orange{y}\sim \mu$! Thus, at first sight, speculative rejection sampling seems pointless; one can simply ignore $\nu$ and directly sample $\sim \mu$.

A key observation is that with access to parallelism, many iterations of the while loop can be run simultaneously. 
Ideally, we would like to run the entire while loop in parallel, so that its total cost matches the cost of drawing a single sample $\Orange{y}\sim \mu.$ In this ideal setting, speculative rejection sampling would pay this cost only $\dTV{\mu, \nu}$ fraction of the time, making it an appealing speedup whenever $\dTV{\mu, \nu}$ is small. However, fully parallelizing the loop would require generating infinitely many candidate processes, leading to an unbounded number of oracle queries.

To approximate this ideal behavior while keeping the number of queries finite, we run \textbf{batches} of the while loop iterations in parallel, with geometrically increasing batch sizes, until one $\Orange{y}$ from a batch gets accepted. Surprisingly, this simple trick yields a near-ideal speedup in the number of rounds, while increasing the number of queries by only a logarithmic factor.

The last missing piece is to sample $\Orange{y}\sim\mu$. For this, we use a divide-and-conquer recursion: we partition the problem into smaller subproblems by splitting the set of coordinates in the autoregressive case and the time interval in the denoising diffusion case, and recursively apply the same sampling procedure to each part, processing these recursive calls sequentially across the parts. In \cref{sec:algorithm}, we tackle both the autoregression case and denoising diffusion case under a single unified approach, where the divide-and-conquer step corresponds to the $\fallback{}$ in \cref{alg:rs2}.

\begin{figure}
\Tikz*{
    \begin{axis}[hide axis, no markers, smooth, domain=-2.5:2.5, height=5cm, width=8cm]
        \addplot[line width=1, Black, name path=A1, domain=-2.5:0] {exp(-(x+0.5)^2/2)};
        \addplot[line width=1, Black, name path=A2, domain=0:2.5] {exp(-(x+0.5)^2/2)};
        \addplot[line width=1, Black, dashed, name path=B1, domain=-2.5:0] {exp(-(x-0.5)^2/2)};
        \addplot[line width=1, Black, dashed, name path=B2, domain=0:2.5] {exp(-(x-0.5)^2/2)};
        \addplot[draw=none, name path=C1, domain=-2.5:0] {0};
        \addplot[draw=none, name path=C2, domain=0:2.5] {0};
        \addplot[LightOrange] fill between[of=A1 and B1];
        \addplot[LightNavy] fill between[of=B1 and C1];
        \addplot[LightNavy] fill between[of=A2 and C2];
        \node at (axis cs: -2,0.5) {$\mu$};
        \node at (axis cs: 2,0.5) {$\nu$};
    \end{axis}
}
\caption{\label{fig:speculative-rejection} Speculative rejection sampling first attempts to produce an $\Navy{x}$ from the area common to both distributions by sampling $\sim \nu$ and rejecting with a carefully tuned probability; if unsuccessful, it switches to sampling $\Orange{y}$ from the defect region by sampling $\sim \mu$ and accepting with a carefully tuned probability. Sampling from the defect region is continued until acceptance.}
\end{figure}
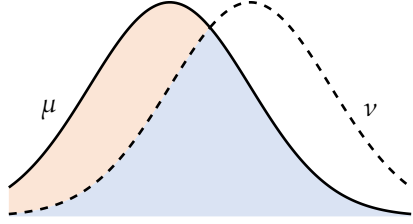

We apply speculative rejection sampling to any-order autoregressive models by setting $\nu$ to be a product distribution with the same marginals as $\mu$. For denoising diffusion models, we take $\mu$ to be the path measure of the stochastic localization SDE, and $\nu$ to be the path measure of a similar SDE, but with constant drift throughout.  Initially, in both cases $\dTV{\mu, \nu}$ may be large; however, as the recursion proceeds and the subproblems become smaller, the total variation distance $\dTV{\mu', \nu'}$ in recursive calls decreases accordingly.

The expected number of queries can be bounded using a simple inductive argument; the main challenge lies in analyzing the runtime, that is, bounding the expected number of rounds.

At the heart of our analysis are inequalities, \cref{lem:pinning lemma for induction} and \cref{prop:mean difference bound via Ito formula}, that bound the total variation distance between the distribution $\mu'$s produced by the recursion and their corresponding speculation distributions $\nu'$s. These inequalities are related in spirit and proof to the so-called pinning lemmas \cite{Mon08,RT12} and analogs of pinning lemmas for stochastic localization \cite{Eld20,EM22}. 

As a remark, even after bounding the total variation distances between distributions and their speculations, the runtime analysis still remains challenging. This is because different recursive calls take a random amount of time, and in a batch of parallel recursive calls, we need to wait until all of them are finished. In other words, the runtime (i.e. the number of rounds) of a batch is determined by the slowest call, which is typically hard to control. 

To handle this, we analyze the runtime by rewriting it as the sum of the cost of the while-loop of all recursive subproblems across all levels of the recursion. 
While this may appear loose, the combination of our batching strategy and the structure of the recursion ensures that each subproblem is only processed a small, poly-logarithmic number of times in expectation. Recall that for a given subproblem with target distribution $\mu'$ and speculative distribution $\nu'$, the while loop only executes with probability $d_{TV}(\mu',\nu').$ Consequently, the expected runtime is roughly the sum of the total variation distance between $\mu'$ and $\nu'$ across subproblems, which can be controlled by \cref{lem:direct kl divergence bound} and \cref{prop:mean difference bound via Ito formula}.

We provide a common analysis of the runtime and number of queries in \cref{sec:runtime}. We then specialize to any-order autoregressive models in \cref{sec:autoregression} and to denoising diffusion in \cref{sec:diffusion-models}.

%% file: prelims.tex
\section{Preliminaries}\label{sec:prelim}

When a random variable $X$ is distributed according to the probability distribution $\mu$, i.e., when for every measurable set $A$, we have $\P{X\in A}=\mu(A)$, we write $X\sim \mu$ and $\mu=\law{X}$.

\begin{definition}[Radon-Nikodym derivative]\label{def:random-nikodym}
    For distributions $\mu$ and $\nu$ on the same space $\Omega$, we denote the Radon-Nikodym derivative of $\nu$ w.r.t.\ $\mu$ by $\dd{\nu}{\mu}:\Omega\to \R_{\geq 0}$. This is a measurable function such that for any measurable set $A$
    \[ \nu(A)=\int_{A} \dd{\nu}{\mu}\d{\mu}. \]
    The derivative exists whenever $\nu$ is absolutely continuous w.r.t.\ $\mu$, that is $\nu\ll \mu$, and is unique up to a set of measure $0$ according to $\mu$. For simplicity, and with some abuse of notation, we just treat it as a uniquely defined object throughout.
\end{definition}

\begin{definition}[Product distributions]
    For distributions $\mu_1,\dots,\mu_n$ on spaces $\Omega_1,\dots,\Omega_n$ we define $\mu_1\otimes \cdots \otimes \mu_n$ to be the product distribution, supported on $\Omega_1\times\cdots\times \Omega_n$ defined by setting
    \[ \mu_1\otimes\cdots\otimes \mu_n(A_1\times \cdots\times A_n)=\mu_1(A_1)\cdots \mu_n(A_n)\]
    for all measurable $A_1,\dots,A_n$.
\end{definition}

We use $\Normal{m, \Sigma}$ to denote the normal/Gaussian distribution with mean $m$ and covariance matrix $\Sigma$.

\subsection{Information theory}

Here we state several definitions and well-known facts from information theory. For a comprehensive reference see \cite{Gra11}.

\begin{definition}[KL-divergence]
    For two distributions $\mu, \nu$ on the same space $\Omega$, we define the Kullback-Leibler (KL) divergence as
    \[ \DKL{\nu\river \mu}=\E*_{x\sim \mu}{\dd{\nu}{\mu}(x)\cdot \log\parens*{\dd{\nu}{\mu}(x)}}=\E*_{x\sim \nu}{\log\parens*{\dd{\nu}{\mu}(x)}}. \]
\end{definition}

\begin{definition}[Entropy]
    Let $X$ and $Y$ be random variables on finite spaces $\Omega_1$ and $\Omega_2$, respectively. The entropy of random variable $X$ is defined to be
    \[\H{X} =-\sum_{x\in \Omega_1} \P{X=x} \log\parens*{\P{X=x}}. \]
    The conditional entropy of $X$ conditioned on $Y$ is defined to be
    \[ \H{X\given Y}= \sum_{y\in \Omega_2} \P{Y=y}\cdot \H{X \given Y=y} = \H{X,Y} - \H{Y}.\]
\end{definition}
\begin{proposition}\label{prop:simple-entropy-bound}
    Let $X$ be a random variable supported on $[q]$. Then $0\leq \H{X} \leq \log q$.
\end{proposition}

\begin{definition}[Mutual information]
    Let $X, Y$ be random variables. The mutual information between $X$ and $Y$ is defined as
    \[\I{X;Y} = \DKL{\law{X,Y} \river \law{X}\otimes \law{Y}}=\E*_{Y}{\DKL*{\law{X\given Y}\river \law{X}}}.\]
    On finite spaces, we can also write
    \begin{equation}\label{eq:mutual information equivalent def}
    \I{X;Y} =  \H{X} - \H{X\given Y} = \H{X}+\H{Y} - \H{X,Y}.
    \end{equation}
\end{definition}

There are several generalizations of mutual information to multiple variables. One that plays a crucial role for us is total correlation.
\begin{definition}[Total correlation \cite{Wat60}]
    Let $X_1,\cdots, X_n$ be random variables. Then we define the total correlation between $X_1,\dots,X_n$ as
    \[ \C{X_1,\dots,X_n}=\DKL { \law{X_1, \dots, X_n}\river \law{X_1} \otimes \law{X_2} \otimes \cdots \otimes \law{X_n}}.\]
    On finite spaces we have
    \[ \C{X_1,\dots,X_n}=\sum_{i=1}^n \H{X_i} - \H{X_1,\dots, X_n}.\]
\end{definition}

The following plays the role of a ``potential function'' in our proofs. It is also known in the literature as the sum of the total correlation and the ``dual total correlation'' (Lemma 4.3 of \cite{austin2020multi}).
\begin{definition}[Potential function]\label{def:potential function}
    For $n\geq 2$, let $X_1,\cdots, X_n$ be random variables on $\Omega_1, \dots, \Omega_n$ respectively. For $S \subseteq [n],$ let $X_S$ denote the concatenation of $X_i$ for $i\in S.$ Define 
    \[\phi(X_{S}) = \sum_{i\in S} \I*{X_i; X_{S\setminus \set{i}}} =\sum_{i\in S } \parens*{\H{X_i} - \H{X_i\given X_{S\setminus \set{i}}}} \]
    We will sometimes work with random variables of the form $\tilde{X}_i = \parens{X_i\given Y}\forall i$ for some random variable $Y$. In this case,
    we write 
$\phi\parens{X_S \given Y}$ as shorthand for $ \phi(\parens{X_i|Y}_{i\in S})$.
\end{definition}

The following shows that the potential function upper bounds the KL-divergence of a joint distribution over $n $ variables and the product distribution with the same marginals.
\begin{proposition}\label{prop:potential function upper bound kl}
    For $n\geq 2$, let $X_1,\cdots, X_n$ be random variables on $\Omega_1, \dots, \Omega_n$ respectively. For $S \subseteq [n],$ let $X_S$ denote the joint distribution over $X_i$ for $i\in S.$
    Then
    \[\C{X_1,\dots,X_n}=\DKL { \law{X_1, \cdots, X_n}\river \law{X_1} \otimes \cdots \otimes \law{X_n}}\leq \phi(X_{[n]}), 
    \]
    where $\phi$ is as defined in \cref{def:potential function}.
    In particular, we have
    \begin{equation} \label{eq:conditional mutual information ineq}
    \H{X_1, X_2, Y} + \H{Y}\leq \H{X_1,Y} + \H{X_2, Y}.
    \end{equation}
    This can also be derived using non-negativity of the conditional mutual information i.e., $\I{X_1;X_2\given Y}\geq 0.$ 
\end{proposition}
\begin{proof} The second statement is obtained by applying the first statement for $n=2$ and $ \tilde{X}_i= \parens{X_i \given Y}$. We prove the first statement.
    Rewriting the LHS and the RHS using \cref{eq:mutual information equivalent def} gives
    \[ (n-1) \H{X_{[n]}} \leq \sum_{i=1}^n \H{X_{[n]\setminus \set{i}}}\]
    We prove this by induction on $n$ for $n\geq 2.$ The base case $n=2$ is equivalent to $\H{X_1}+ \H{X_2} - \H{X_1,X_2} = \I{X_1;X_2} \geq 0$, which is true by definition of mutual information.

    Suppose the induction hypothesis is true for $ n-1\geq 2.$
    Fix an arbitrary $i\in [n].$  For $j\neq i$, let $Y_j  = \parens{X_j \given X_i}$. The induction hypothesis to $\set{Y_j \given j\neq i}$, and note that $Y_{[n]\setminus i} = \parens{X_{[n]\setminus i} \given X_i}$ gives
\[ (n-2) \H{X_{[n]\setminus i} \given X_i} \leq  \sum_{j\neq i} \H{X_{[n]\setminus \set{i,j}}\given X_i} \]
The above combined with the fact that $ \H{U\given V} = \H{U,V} - \H{V}$ (see \cref{eq:mutual information equivalent def}) implies
\[ (n-2) \H{X_{[n]}} + \H{X_i} \leq \sum_{j\neq i} \H{X_{[n]\setminus \set{j}}}.   \]
Summing these inequalities over all $i\in [n],$ and using the fact that $ \sum_{i=1}^n \H{X_i} \geq \H{X}$, which is implied by $\sum_{i=1}^n \H{X_i} -\H{X} =\DKL{\law{X}\river \law{X_1} \otimes \cdots\otimes \law{X_n}} \geq 0$, concludes the proof.
\end{proof}
The following is a consequence of \cref{eq:conditional mutual information ineq}. 
We have a simple upper bound on the potential function.
\begin{proposition}\label{prop:potential function trivial bound}
    Let $X $ be the joint distribution on $n$ random variables $X_1, \cdots, X_n$ taking value in $ [q].$ Then 
    $\phi(X)\leq n\log q.$
\end{proposition}
\begin{proof}
    This is simply by noting that $ \I{X_i;X_{S\setminus {i}}} \leq \H{X_i}\leq \log q.$
\end{proof}
\begin{lemma}[New pinning lemma]\label{lem:pinning lemma for induction}
    For $n\in \N^*$, let $X_1,\cdots, X_{2n}$ be random variables on $\Omega_1, \dots, \Omega_{2n}$ respectively. Let $\sigma \sim \mathcal{S}_{2n}$ be a random permutation, $A_1 = \set{1,\dots, n}, A_2 = \set{n+1, \cdots, 2n}$. Then
    \[\E*_{\sigma} { \phi\parens{X_{\sigma(A_1)}}  + \phi\parens{X_{\sigma(A_2)}\given  X_{\sigma(A_1)}}}   \leq \frac{1}{2} \phi(X_{[2n]}) \]
    where $\sigma(S) := \set{\sigma(i) \given i\in S}$ and $\phi$ is as defined in \cref{def:potential function}.
    More generally, for $k\in \N_{\geq 2}$ and random variables $X_1, \cdots,X_{nk}$ and let $A_r = \set{n(r-1)+1, \cdots, nr}\forall r\in [k],$ we have:
    \[\E*_{\sigma\sim \mathcal{S}_{nk}} { \sum_{r=1}^k \phi\parens{X_{\sigma(A_r)} \given X_{\sigma(A_1 \cup \cdots \cup A_{r-1})}} }   \leq \frac{1}{k} \phi(X_{[nk]}) \]
\end{lemma}
\begin{proof}
Fix a $\sigma$ and let $S_1 = \sigma(A_1), S_2 =\sigma(A_2)$.
    We can write
    \begin{align*}
        &\phi(X_{S_1})  + \phi\parens{X_{S_2}\given X_{S_1}}\\
        = &\sum_{i\in S_1} \H{X_{i}} - \sum_{i\in S_1} \H{X_{i}\given X_{S_1\setminus \set{i}}} +\sum_{i\in S_2} \H{X_{i}\given X_{S_1}} - \sum_{i\in S_2} \H{X_{i} \given X_{ S_1 \cup S_2\setminus \set{i}}}\\
        = &\sum_{i\in S_1} \H{X_{i}} - \sum_{i\in S_1} (\H{X_{S_1}} - \H{X_{S_1\setminus \set{i}}})+ \sum_{i\in S_2} (\H{X_{S_1 \cup \set{i}}} - \H{X_{S_1}}) - \sum_{i\in S_2} \H{X_{i} \given X_{ [2n]\setminus \set{i}}}.
    \end{align*}
    Now, note that each $i$ is in $S_1$ ($S_2$ resp.) with probability $1/2.$ Thus
    \[\E*_{\sigma} {\sum_{i\in S_1} \H{X_i} - \sum_{i\in S_2} \H{X_{i} \given X_{ [2n]\setminus \set{i}}}} = \frac{1}{2} \sum_{i\in [2n]} \H{X_i} - \H{X_{i} \given X_{ [2n]\setminus \set{i}}} = \frac{1}{2} \phi(X_{[2n]})\]
    Thus we only need to prove
    \begin{equation}\label{eq:pinning concise}
       \sum_{S\in \binom{[2n]}{n}} \parens*{2n \H{X_{S}} - \sum_{i\in [2n]} \H{X_{S \Delta \set{i}}}}  \geq 0 
    \end{equation}
    where $U\Delta V$ denote the symmetric difference between sets $U, V.$ For $ k\in [2n],$ let $ Q_k =\sum_{S\in \binom{[2n]}{k}} \H{X_S}$ then the above is equivalent to
    \begin{equation}\label{eq:pinning main}
       2n Q_n \geq (n+1) (Q_{n-1} + Q_{n+1}) 
    \end{equation}
    Fixed an arbitrary set $ T $ of size $n-1$ and distinct elements $i, j \not \in T.$
   \cref{eq:conditional mutual information ineq} implies
   \[  \H{X_{T\cup \set{i}}} + \H{X_{T\cup\set{j}}} \geq \H{X_T} + \H{X_{T\cup \set{i,j}}}.  \]
    Summing over $T$ and ordered tuple $(i,j)$ where $i\neq j$ gives
 \begin{align*}
     0&\leq \sum_{\abs{T}=n-1} \sum_{(i,j); i,j\not\in T} ( \H{X_{T\cup \set{i}}} + \H{X_{T\cup\set{j}}} - \H{X_T} -\H{X_{T\cup \set{i,j}}}) \\
    &= 2  \sum_{\abs{S}=n}\sum_{i\in S, j\not\in S} \H{X_S} - \sum_{\card{T} = n-1} \sum_{(i,j); i,j\not\in T} \H{X_T} - \sum_{\card{W} = n+1} \sum_{i\neq j; i,j\in W} \H{X_W}\\
    &= 2n^2 Q_n - (n+1) n (Q_{n-1} + Q_{n+1})
 \end{align*}   
 thus \cref{eq:pinning main} is true.

 For the general case, we can reduce the problem to the following inequality:
 \[\forall r\in [k-1]: \sum_{\sigma} ( 2n \H{X_{\sigma(A_r)}} - \sum_{i\in A_r\cup A_{r+1}} \H{X_{\sigma(A_r)\Delta \sigma(i)}}\geq 0\]
 which is true by the above proof, since the inequality becomes  \cref{eq:pinning concise} when summing over $\sigma$ s.t. the value of $ \sigma$ on $(A_r\cup A_{r+1})^c$ is fixed. 
\end{proof}

%% file: algorithm.tex
\section{Recursive speculative rejection sampling} \label{sec:algorithm}

In this section, we present our parallel sampling algorithm. Although coordinate denoisers and Gaussian denoisers may seem very different, we handle both cases in an almost unified way. For clarity, we suggest having the coordinate denoiser (the any-order autoregressive case) in mind when reading the description of our algorithm for the first time.

Throughout this section, we work with a distribution $\mu$ on a space $\Omega_1\times \Omega_2\times \cdots \times \Omega_N$, and our goal is to sample $Z$ from $\mu$. For a subset $S\subseteq [N]$, we use the shorthand $\Omega_S$ to denote $\prod_{i\in S}\Omega_i$. So $\mu$ is a distribution on $\Omega_{[N]}$, and for $Z\in \Omega_{[N]}$, we have $Z_{S}\in \Omega_S$.

The distribution $\mu$ is derived from our target distribution in the following ways:

\paragraph{Coordinate denoising scheme.} Here we have a distribution $\pi$ on $[q]^n$ accessible via a coordinate denoiser i.e. conditional marginals of $\pi$. We let $N=n$, and $\Omega_1=\dots=\Omega_N=[q]$. The distribution $\mu$ will be identical to $\pi$ but after a \emph{random permutation} on the coordinates. In other words, we first choose a uniformly random permutation $\sigma$ on $[n]$, and then let for all $x\in [q]^n$:
\[ \mu(x_{1},\dots,x_{n})=\pi(x_{\sigma(1)},\dots,x_{\sigma(n)}).\]
\paragraph{Gaussian denoising scheme.} Here we have a distribution $\pi$ on $\R^n$ accessible via a Gaussian denoiser. We first choose a discretization schedule $0=t_0<t_1<\dots<t_N$ that would guarantee the accuracy of the Euler-Maruyama discretization of the stochastic localization SDE (\cref{alg:sequential-SL}). Name the resulting stochastic process from running \cref{alg:sequential-SL}, $X_{t_0},X_{t_1},\dots,X_{t_N}$. We let $\Omega_1=\dots=\Omega_N=\R^n$ and we let $\mu$ be the distribution of the difference sequence $\Delta X\coloneq \parens*{X_{t_1}-X_{t_0},X_{t_2}-X_{t_1},\dots,X_{t_N}-X_{t_{N-1}}}$.

Note that a sample from $\mu$ can be easily translated to a sample from $\pi$ in the coordinate denoiser case, and an approximate sample from $\pi$ in the Gaussian denoiser case. In both cases, the standard sequential algorithm proceeds to sample a $Z=(Z_1,\dots,Z_N)\sim \mu$ one $Z_i$ at a time. So in iteration $i$, one samples $Z_i$ conditioned on $Z_1,\dots,Z_{i-1}$. At the core of our parallel algorithm is the following idea: we pretend that some subsets of $Z_i$s are (conditionally) independent but correct for this inaccurate independence assumption through speculative rejection sampling. Independent random variables can be sampled separately in parallel, which powers our parallel speedup.

As mentioned, the independence of subsets of $Z_i$s almost never holds exactly, and one has to correct for it. Speculative rejection sampling corrects for this by occasionally falling back on a slower sampling algorithm. For example, one can use the standard sequential algorithm as the fallback. However, to get the maximum possible speedup, we use multiple layers of fallbacks. This comes in the form of splitting a subset of $Z_i$s into smaller subsets, each of which we again pretend to be independent (conditioned on prior subsets). Those subsets are themselves partitioned into even smaller subsets, and so on. To make this more concrete, we organize these partitions into a tree defined as follows.
\begin{definition}[Fallback tree]
	A fallback tree is a tree whose leaf nodes correspond to $\set{1,\dots,N}$, and all leaf nodes have the same height (i.e., same distance from the root). Moreover, for any node, the leaf nodes in its subtree form a contiguous subset $S = \set{i, i+1, \dots, j} \subseteq [N]$. For convenience, we identify each node with its corresponding subset $S$.\footnote{This is a mild abuse of notation, since a node with only one child would share the same subset as its child. However, such cases rarely occur in the fallback trees we use, so we ignore this distinction for simplicity.}
\end{definition}

Note that we have a natural ordering on the children of a node $S$ in a fallback tree; the sets representing them must be a partitioning of $S$ into contiguous blocks $T_1,\dots,T_k$, and we can after rearranging the indices, let $T_1$ be the block with the smallest elements, $T_2$ the next one, and so on. For a node $S=\set{i,\dots,j}$, we use the shorthand $<S$ to denote the set of indices before $S$, that is $<S$ denotes $\set{1,\dots,i-1}$ and similarly we use the shorthand $>S$ to denote the set of indices after $S$, that is $>S$ denotes $\set{j+1,\dots,N}$.

For any node $S$ in a fallback tree, we can define a conditional distribution $\mu_S\parens{\cdot \given Z_{<S}}$, which assigns to any $Z_{<S}\in \Omega_{<S}$ a distribution on $\Omega_S$. Formally, this is the distribution of $Z_S$ conditioned on $Z_{<S}$ assuming that $Z\sim \mu$. The key to our parallel speedup is \emph{speculations} for these conditional distributions; these are (conditional) distributions that are speculatively close (in total variation distance) to $\mu_S\parens{\cdot \given Z_{<S}}$, but are much easier to sample from in parallel.
\begin{definition}[Speculation scheme]
	For a fallback tree, a speculation scheme is an assignment of a collection of conditional distributions $\nu_S\parens*{\cdot\given Z_{<S}}$ to every node $S$. Each $\nu_S$ is a function that receives a ``conditioning'' $Z_{<S}\in \Omega_{<S}$ and gives a distribution on $\Omega_S$.
\end{definition}
We note that although the conditional distributions $\mu_S$ are all compatible with one global distribution $\mu$, this is not necessarily the case for a speculation scheme. There does not have to be a ``global $\nu$'' from which all of the $\nu_S$ are derived.

Before describing our algorithm, we first summarize the main assumptions underlying a speculation scheme to build some intuition. These assumptions will be made precise later, when we instantiate our algorithm for the coordinate and Gaussian denoiser cases.
\begin{assumption}[Informal main assumptions about the speculation scheme (see \cref{assumption:formal} for the formal version)]\label{assumption:informal}
	We assume the speculation scheme to satisfy the following properties:
	\begin{enumerate}
		\item For each node $S$ and conditioning $Z_{<S}\in \Omega_{<S}$, sampling $Z_S\sim \nu_S\parens*{\cdot\given Z_{<S}}$ is parallelizable, i.e., it takes $O(1)$ round of calls to the denoiser oracle.
		\item For each node $S$, conditioning $Z_{<S}\in \Omega_{<S}$, and sample $Z_S\in \Omega_S$, computing the density (Radon-Nikodym derivative, see \cref{def:random-nikodym}) of $\mu_S\parens*{\cdot\given Z_{<S}}$ w.r.t.\ $\nu_S\parens*{\cdot \given Z_{<S}}$ at the point $Z_S$, which we represent as $\dd{\mu_S}{\nu_S}\parens*{Z_S\given Z_{<S}},$ is parallelizable, i.e., it takes $O(1)$ rounds of calls to the denoiser oracle.
            \item At the leaf nodes, our speculations are exact. In other words, for $\card{S}=1$, we have $\nu_S=\mu_S$.
        \end{enumerate}
\end{assumption}
The first assumption helps us produce samples from $\nu_S$ fast, and the second assumption helps us quickly accept/reject them via speculative rejection sampling. As we will see, when we reject, we have to fall back on the children of $S$. The third assumption ensures that at the leaf nodes we never reject speculations.

We are now ready to present our algorithm, recursive speculative rejection sampling, or $\RS{}$ for short. The pseudocode can be found in \cref{alg:rs2}. The algorithm consists of two recursive functions named $\RS{}$ and $\fallback{}$, each of which receives a node $S$ in the fallback tree along with a conditioning $Z_{<S}\in \Omega_{<S}$, and returns a sample from $\mu_S\parens*{\cdot\given Z_{<S}}$. The helper function $\fallback{}$ does so by following a sequential sampling strategy applied to the children of $S$, while $\RS{}$ uses a version of speculative rejection sampling (as described in \cref{sec:techniques} and \cref{alg:srs}).

\begin{Algorithm}
\Function{$\RS{S,Z_{<S}}$}{
    $\Blue{X}\gets\text{sample from }\nu_S\parens*{\cdot\given Z_{<S}}$\;
    $u\gets \min\set*{1,\dd{\mu_S}{\nu_S}\parens*{\Blue{X}\given Z_{<S}}}$\;
    \eIf{coin flip with bias $u$ comes up heads}{
        \Return{$\Blue{X}$}\;
    }{
    	\For{$r=0, 1, 2, \dots$}{
    		\ParallelFor{$i=  \lceil (1+\rho)^r\rceil,\lceil (1+\rho)^r\rceil+1, \dots, \lceil (1+\rho)^{r+1}\rceil - 1$}{
    			$\Orange{Y^{(i)}}\gets \fallback{S, Z_{<S}}$\;
    			$p^{(i)}\gets \max\set*{0, 1-\dd{\nu_S}{\mu_S}\parens*{\Orange{Y^{(i)}}\given Z_{<S}}}$\;
    			$c^{(i)}\gets \text{coin flip with bias }p^{(i)}$\;
    		}
    		\If{any $c^{(i)}$ has resulted in heads}{
    			$i^*\gets \text{index of first head}$\;
    			\Return{$\Orange{Y^{(i^*)}}$}\;
    		}
    	}
    }
}\;

\Function{$\fallback{S, Z_{<S}}$}{
    $T_1,\dots,T_k\gets\text{children of }S\text{ in the fallback tree, in increasing order}$\;
    $Z\gets Z_{<S}$\;
    \For{$i=1,\dots,k$}{
        $Z_{T_i}\gets \RS{T_i, Z}$\;
        append $Z_{T_i}$ to the end of $Z$\;
    }
    \Return{$Z_S$}\;
}
\caption{Recursive speculative rejection sampling with parameter $\rho>0$\label{alg:rs2}}
\end{Algorithm}

In more detail, the $\fallback{}$ sampling strategy simply goes over the children $T_1,\dots,T_k$ of the node $S$ in order, and in each iteration samples $Z_{T_i}$ conditioned on all the previously sampled values (including the initial $Z_{<S}$), by recursively calling $\RS{}$ on $T_i$. It is easy to see that assuming $\RS{}$ correctly samples from its target distribution, then so does $\fallback{}$.

The function $\RS{}$ uses speculative rejection sampling (see \cref{sec:techniques} and \cref{alg:srs}) with the speculation distribution being $\nu_S\parens*{\cdot\given Z_{<S}}$ and the target distribution being $\mu_S\parens*{\cdot \given Z_{<S}}$. Namely, it first samples $X$ from the speculation distribution $\nu_S\parens*{\cdot\given Z_{<S}}$, and accepts it with probability $\min\set*{1,\dd{\mu_S}{\nu_S}\parens*{X\given Z_{<S}}}$. If rejected however, it has to use a slower sampling strategy ($\fallback{}$) to produce a stream of i.i.d.\ samples $\Orange{Y^{(i)}}\sim \mu_S\parens*{\cdot\given Z_{<S}}$ accepting each with probability $\max\set*{0, 1-\dd{\nu_S}{\mu_S}\parens*{\Orange{Y^{(i)}}\given Z_{<S}}}$. As mentioned in \cref{sec:techniques}, we do not want to produce the i.i.d.\ samples $\Orange{Y^{(i)}}$ sequentially, or else the whole scheme would be defeated. In an ideal unachievable world, all the $\Orange{Y^{(i)}}$ would be produced in parallel, but this is impossible as there are infinitely many $\Orange{Y^{(i)}}$. So instead $\RS{}$ uses a geometric scaling scheme, where we start with producing one $\Orange{Y^{(i)}}$s and geometrically ramp this up, in each iteration producing roughly $(1+\rho)$ times as many as the previous iteration. Here $\rho>0$ is a small parameter, to be set later. This ensures two properties hold: first, that we do not produce many more $\Orange{Y^{(i)}}$ than necessary---we can only produce at most $1+\rho$ times as many as $i^*$; second, we sufficiently parallelize the production of $\Orange{Y^{(i)}}$---we only produce $\simeq \log(i^*)/\rho$ batches of $\Orange{Y^{(i)}}$, with each batch being produced in parallel.

In the remainder of this section, we justify the correctness of \cref{alg:rs2} and explain how to instantiate speculation schemes from a coordinate or Gaussian denoiser. We will analyze the number of rounds and queries in the ensuing \cref{sec:runtime,sec:autoregression,sec:diffusion-models}.

\subsection{Correctness}

\begin{proposition}
	Assume that the speculation scheme $\nu_S$ matches $\mu_S$ for all leaf nodes $S$. Then, for each non-leaf node $S$ and conditioning $Z_{<S}\in \Omega_{<S}$, the output of $\RS{S, Z_{<S}}$ and the output of $\fallback{S, Z_{<S}}$ are also distributed according to $\mu_S\parens*{\cdot\given Z_{<S}}$.
\end{proposition}
\begin{proof}
	The proof follows by induction on the height of the nodes $S$ of the fallback tree. The base case, when $S$ is a leaf node, follows from the assumption that the speculation scheme matches the target distribution at the leaf nodes. This means that when $\Blue{X}$ is generated, it is accepted with probability $p=1$, because $\dd{\mu_S}{\nu_S}=1$ at the leaf nodes $S$. So the returned sample is distributed according to $\nu_S\parens*{\cdot\given Z_{<S}}=\mu_S\parens*{\cdot\given Z_{<S}}$.
	
	The induction step follows from the correctness of sequential sampling (for the $\fallback{}$ function) and the correctness of speculative rejection sampling (for the $\RS{}$ function). The former is an easy exercise, while the latter is proved below in \cref{prop:srs-is-correct}.
\end{proof}

It just remains to prove that speculative rejection sampling, i.e., \cref{alg:srs}, is correct.
\begin{proposition}\label{prop:srs-is-correct}
	The output of \cref{alg:srs} is always distributed as $\mu$. Furthermore, the chance that $\Blue{x}$ is accepted is $1-\dTV{\mu, \nu}$ and the chance that each $\Orange{y}$ is accepted is $\dTV{\mu, \nu}$.
\end{proposition}
\begin{proof}
	Note that the chance that a random $y\sim \mu$ is accepted is
	\[ \E*_{y\sim \mu}{\max\set*{0, 1-\dd{\nu}{\mu}(y)}}=\E*_{y\sim \mu}{\max\set*{0, \dd{\nu}{\mu}(y)-1}}. \]
	This follows from the fact that
	\[ \E*_{y\sim \mu}{1-\dd{\nu}{\mu}(y)}=1-1=0. \]
	So we can rewrite the chance of acceptance as
	\[ \frac{1}{2}\E*_{y\sim \mu}{\max\set*{0, 1-\dd{\nu}{\mu}(y)}+\max\set*{0, \dd{\nu}{\mu}(y)-1}}=\frac{1}{2}\E*_{y\sim \mu}{\abs*{1-\dd{\nu}{\mu}(y)}}=\dTV{\mu, \nu}. \]
	
	A similar calculation shows that the chance that we accept $x\sim \nu$ in the first step of the algorithm is
	\[ \E*_{x\sim \nu}{\min\set*{1, \dd{\mu}{\nu}(x)}} = 1-\E*_{x\sim \nu}{\max\set*{0, 1-\dd{\mu}{\nu}(x)}}= 1-\dTV{\mu, \nu}.  \]
	Let $\pi_1$ be the distribution of $x$ conditioned on it being accepted, and $\pi_2$ the distribution of the first $y$ accepted from the while loop. The output distribution $\pi$ will be the mixture $(1-\dTV{\mu, \nu})\cdot \pi_1+\dTV{\mu, \nu}\cdot \pi_2$. It is a well-known fact about rejection sampling that
	\[\dd{\pi_1}{\nu}(x)=\frac{\min\set*{1,\dd{\mu}{\nu}(x)}}{1-\dTV{\mu, \nu}}, \quad \dd{\pi_2}{\mu}(y)=\frac{\max\set*{0, 1-\dd{\nu}{\mu}(y)}}{\dTV{\mu, \nu}}. \]
	Therefore for any event $E$, we can write $\pi(E)$ as
	\begin{align*} (1-\dTV{\mu, \nu})\pi_1(E)+\dTV{\mu, \nu}\pi_2(E)&=\int_E \min\set*{1,\dd{\mu}{\nu}(x)}\d\nu(x)+\int_E \max\set*{0,1-\dd{\nu}{\mu}(y)}\d\mu(y)\\
	&=\int_E \min\set*{\dd{\nu}{\mu}(x),1}\d\mu(x)+\int_E \max\set*{0,1-\dd{\nu}{\mu}(y)}\d\mu(y)\\
	&=\int_E \parens*{\min\set*{\dd{\nu}{\mu}(x),1}+\max\set*{0,1-\dd{\nu}{\mu}(x)}}\d\mu(x)=\int_E \d\mu(x)=\mu(E).
	\end{align*}
	Here we used the fact that $\min\set{\alpha, 1}+\max\set{0, 1-\alpha}=1$. This proves that $\pi=\mu$.
\end{proof}

\subsection{Speculation scheme from coordinate denoiser}\label{sec:spec-scheme-coordinate}

Here we briefly describe, given an arbitrary fallback tree, how to attach a speculation scheme based on a coordinate denoiser. Our informal goal is to satisfy \cref{assumption:informal}. Note that the details of the fallback tree directly affect the runtime of our algorithm, but we discuss these details in \cref{sec:runtime,sec:autoregression}.

Assume that $\pi$ is a distribution on $[q]^n$ accessible via a coordinate denoiser. As mentioned at the beginning of this section, we first choose a random permutation $\sigma$ on $[n]$ and let $\mu$ be the distribution $\pi$ after permuting the coordinates with $\sigma$. Now assume we have a fallback tree with leaf nodes $[N]=[n]$. We need to describe for a node $S$, what $\nu_S\parens*{\cdot\given Z_{<S}}$ looks like.

\begin{definition}[Coordinate denoiser speculation scheme] 
	For each $i\in S$, let $\mu_i\parens*{\cdot \given Z_{<S}}$ denote the distribution of $Z_i$ conditioned on $Z_{<S}$, assuming that $Z\sim \mu$. We define $\nu_S\parens*{\cdot \given Z_{<S}}$ by the equation
	\[\nu_S\parens*{Z_S\given Z_{<S}}=\prod_{i\in S}\mu_i\parens*{Z_i \given Z_{<S}}.\]
	In other words $\nu_S\parens*{\cdot\given Z_{<S}}=\bigotimes_{i\in S}\mu_i\parens*{\cdot\given Z_{<S}}$ is a product distribution with the same marginals as $\mu_S\parens*{\cdot \given Z_{<S}}$.
\end{definition}

Now we verify \cref{assumption:informal} for this speculation scheme.
\begin{enumerate}
	\item Note that $\nu_S\parens*{\cdot\given Z_{<S}}$ is a product distribution, each of whose components is a distribution $\mu_i\parens*{\cdot\given Z_{<S}}$ on $[q]$. Each component can be retrieved by a single call to the coordinate denoiser. So all of the components can be retrieved in parallel in just $1$ round, and then we can produce $Z_S\sim \nu_S\parens*{\cdot\given Z_{<S}}$ by sampling each $Z_i$ for $i\in S$ independently (in parallel).

	\item By a similar reasoning as before, we can compute all the components of $\nu_S\parens*{\cdot\given Z_{<S}}$ in $1$ round of calls to the coordinate denoiser, and then for a given $Z_S\in \Omega_S$, we can easily compute $\nu_S\parens*{Z_S\given Z_{<S}}$. The more nontrivial part is how to compute $\mu_S\parens*{Z_S\given Z_{<S}}$; once we do that, we can compute the Radon-Nikodym derivatives by simply dividing the probabilities. To compute $\mu_S\parens*{Z_S\given Z_{<S}}$, we use the fact that we are given the entire point $Z_S$. So if $S=\set{a, a+1,\dots, b}$, we can write
	\[ \mu_S\parens*{Z_S\given Z_{<S}}=\prod_{i=a}^b \mu_i\parens*{Z_i\given Z_{[i-1]}},  \]
	and use the fact that all terms in the product can be retrieved from the coordinate denoiser in $1$ round of parallel calls.
	\item It is trivial from the definition that when $\card{S}=1$, the distributions $\nu_S\parens*{\cdot\given Z_{<S}}$ and $\mu_S\parens*{\cdot\given Z_{<S}}$ are equal.
\end{enumerate}

\subsection{Speculation scheme from Gaussian denoiser}\label{sec:spec-scheme-gaussian}

Here we briefly describe, given an arbitrary fallback tree, how to attach a speculation scheme based on a Gaussian denoiser. Our informal goal is to satisfy \cref{assumption:informal}. Note that the details of the fallback tree directly affect the runtime of our algorithm, but we discuss these details in \cref{sec:runtime,sec:diffusion-models}. 

Assume that $\pi$ is a distribution on $\R^n$ accessible via a Gaussian denoiser. The Gaussian denoiser, which here we denote with the function $f$, gives the drift term in the stochastic localization SDE:
\[ f(t, X_t)\coloneq \E*_{X\sim \pi, g\sim \Normal{0, tI}}{X\given tX+g=X_t}. \]

As mentioned at the beginning of this section, we first choose a discretization schedule $0=t_0<t_1<\dots<t_N$ for the stochastic localization SDE, and let $X_{t_0},\dots, X_{t_N}$ be the Euler-Maruyama discretized process (i.e., the random variables produced in \cref{alg:sequential-SL}). As mentioned at the beginning of this section, we let $\mu$ be the distribution of the difference sequence $\Delta X=\parens*{X_{t_1}-X_{t_0},X_{t_2}-X_{t_1},\dots,X_{t_N}-X_{t_{N-1}}}$.\footnote{We could alternatively define $\mu$ to be the distribution of $\parens*{X_{t_1},\dots,X_{t_N}}$, but the difference sequence is notationally more convenient.}

We also define $\Delta t$ to be the difference sequence for discretization times, i.e., $\Delta t\coloneq (t_1-t_0, t_2-t_1,\dots, t_N-t_{N-1})$. So $(\Delta X)_i$ and $(\Delta t)_i$ mean $X_{t_i}-X_{t_{i-1}}$ and $t_i-t_{i-1}$, respectively. For a set $S\subseteq [N]$, we use the shorthand $\sum Z_S$ to denote $\sum_{i\in S}Z_i$, and similarly we use $\sum (\Delta t)_S$ to denote $\sum_{i\in S}(\Delta t)_i$. For example, if $S=\set{a,a+1,\dots,b}$ is a node in the fallback tree, and $Z\sim \mu$ is obtained by taking the difference sequence $\Delta X$ for $X_{t_0},\dots,X_{t_N}$, then $\sum Z_{<S}=X_{t_{a-1}}$ and $\sum (\Delta t)_{<S}=t_{a-1}$.
\begin{definition}[Gaussian denoiser speculation scheme]
	Given a node $S$ in the fallback tree and $Z_{<S}$, we define the distribution $\nu_S\parens*{\cdot\given Z_{<S}}$ to be a product of normal distributions (itself a normal distribution):
	\[ \bigotimes_{i\in S} \Normal*{(\Delta t)_i v, (\Delta t)_i I}\]
	where $v=f\parens*{\sum (\Delta t)_{<S}, \sum Z_{<S}}$ is the drift of the stochastic localization SDE immediately before $S$. In other words $\nu_S\parens{\cdot\given Z_{<S}}$ is obtained by keeping the drift of the stochastic localization SDE constant and equal to $v$ starting at the immediate discretization point before $S$.
\end{definition}
Now we verify \cref{assumption:informal} for this speculation scheme.
\begin{enumerate}
	\item Note that $\nu_S\parens*{\cdot\given Z_{<S}}$ is a product distribution defined entirely in terms of a single drift vector $v=f\parens*{\sum (\Delta t)_{<S}, \sum Z_{<S}}$. We can calculate $v$ in a single call to the Gaussian denoiser, and then in parallel sample each component of the product distribution.
	\item Let $\Leb$ denote the Lebesgue measure. Given a node $S$, $Z_{<S}$, and $Z_S$, we can easily calculate $\dd{\nu\parens*{Z_S\given Z_{<S}}}{\Leb}$ as the product of Gaussian densities composing $\nu_S\parens*{\cdot\given Z_{<S}}$ by just knowing the drift $v=f\parens*{\sum (\Delta t)_{<S}, \sum Z_{<S}}$, which takes just one call to the Gaussian denoiser oracle. We can also compute $\dd{\mu\parens*{Z_S\given Z_{<S}}}{\Leb}$ by the formula
	\[ \dd{\mu\parens*{Z_S\given Z_{<S}}}{\Leb}=\prod_{i\in S} \dd{\Normal*{(\Delta t)_i f\parens*{\sum (\Delta t)_{[i-1]}, \sum Z_{[i-1]}}, (\Delta t)_i I}}{\Leb}(Z_i). \]
	Note that we need to know all the drifts $f\parens*{\sum (\Delta t)_{[i-1]}, \sum Z_{[i-1]}}$ for all $i$, but this only takes $1$ round of parallel calls to the Gaussian denoiser. Given densities w.r.t.\ the Lebesgue measure, the ratios give us the Radon-Nikodym derivative.
	\item When $\card{S}=1$, we only have one discretization step, and therefore keeping the drift constant throughout this step yields the same process as the original Euler-Maruyama discretization scheme. Therefore for $\card{S}=1$, we have $\nu_S\parens*{\cdot\given Z_{<S}}$ and $\mu_S\parens*{\cdot\given Z_{<S}}$ are equal.
\end{enumerate}

%% file: improved-analysis-polylog.tex
\section{Runtime analysis} \label{sec:runtime}
To analyze our main algorithms (\cref{alg:rs2}) we use two convenient measures: round complexity and query complexity. Here, query complexity counts the total number of queries made to the appropriate denoiser, and round complexity measures the number of parallel rounds in which these queries are sent. Because we access the denoisers indirectly through a speculation scheme, we assume the following:

\begin{assumption}[Main assumptions about the speculation scheme (formal version of \cref{assumption:informal})]\label{assumption:formal}
	We assume the speculation scheme satisfies the following properties:
	\begin{enumerate}
		\item For each node $S$ and conditioning $Z_{<S}\in \Omega_{<S}$, sampling from $\nu_S\parens*{\cdot\given Z_{<S}}$ takes $O(1)$ rounds and $O(\card{S})$ queries to the denoiser oracle.
		\item For each node $S$, conditioning $Z_{<S}\in \Omega_{<S}$, and sample $Z_S\in \Omega_S$, computing $\dd{\nu_S}{\mu_S}\parens*{Z_S\given Z_{<S}}$ takes $O(1)$ rounds and $O(\card{S})$ queries to the denoiser oracle.
            \item For each leaf node $S,$ i.e. $\card{S}=1$, $\nu_S=\mu_S$.
        \end{enumerate}
\end{assumption}

These assumptions are made without loss of generality, as all bounds translate accordingly. See \cref{remark:implementation details coordinate denoiser} and \cref{remark:implementation details gaussian denoiser} for verification in the coordinate denoising and Gaussian denoising settings, respectively.

\begin{definition}[Recursive call index]\label{def:recursive call index}
When analyzing $\RS{}$, it is convenient to index all the calls for a particular node $S$ in the fallback tree using some indices $t$. If we were to run $\RS{}$ sequentially (i.e., replacing the parallel for loop with a sequential for loop), then we would index the first recursive call for a node $S$ with $t=1$, the second recursive call with $t=2$, and so on. We use the same indexing convention in the parallel implementation, since the set of recursive calls is identical regardless of whether they are executed sequentially or in parallel.
\end{definition}

\subsection{Notation}

For a node $S$ in the fallback tree, let $\mathcal{C}_\ell(S)$ denote the set of all descendants of $S$ at distance $\ell$ from $S$, and let $C_{\ell}(S)=\card{\mathcal{C}_\ell(S)}$. We write $\mathcal{C}(S_\ell)$ and $C(S_\ell)$ as shorthand for $\mathcal{C}_1(S_\ell)$ and $C_1(S_\ell).$ We note that $ C(S_\ell)$ is the branching factor of the node $S_\ell.$
\subsection{Bounding the number of queries}

    \begin{lemma}\label{lem:expected number of queries}
 Consider a fallback tree
 with $N$ leaf nodes and height $h$.
 Let $Q$ be the query complexity of $\RS{}$ at the root with parameter $\rho > 0$. Under \cref{assumption:formal}, the expected query complexity is bounded by:
\[\E{Q  } \leq O\left( \frac{(1+\rho)^{h+1}-1}{\rho}\cdot N\right).  \]
\end{lemma}
\begin{proof}
   
    Let $\textbf{root}$ be the root of the fallback tree. W.l.o.g., let us assume that in \cref{assumption:formal}, the number of queries called locally at any node $S $ in the fallback tree is $ \card{S};$ this is without loss of generality  since all bounds will translate up to constant factors.
   We show by induction on the height $h$ that
   \begin{equation}
       \E{Q } \leq N \sum_{\ell=0}^h (1+\rho)^\ell = \frac{(1+\rho)^{h+1}-1}{\rho} \cdot N 
   \end{equation}

When $h=0,$ $\textbf{root}$ is a leaf node and the number of queries called at $\textbf{root}$ is $1,$ so the claim is trivially true. 

Let $I$ be the indicator of the event that $\Blue{X}$ is not accepted in the execution of $\RS{}$ at the root, and $\mathcal{M}$ be the set of recursive call indices of the fallback calls issued by the root node. Let $S_1, \cdots,S_k$ be the children of $\textbf{root}$ and $N_i = \card{S_i}$ be the number of leaf nodes in the subtree rooted at $S_i$. For a fallback call with recursive call index $t$ at the root, let $ Q^{i,t}$ denote the number of queries of the recursive call to $\RS{S_i, \cdot}.$

 Let $\mu, \nu$ be the target and the speculative distribution at the root respectively. By \cref{prop:srs-is-correct}, the event $I$
 happens with probability $\dTV{\mu,\nu}. $ If $ \dTV{\mu,\nu} =0$ then the event $I$ does not happen, and $ Q \leq N $ as desired. Below, assume $\dTV{\mu,\nu}  >0. $
   We have the recursion
    \begin{align*}
\E{Q }   &= N +  \E*{I\cdot  \sum_{i=1}^k \sum_{t\in \branchset{}} Q^{i, t} } \\
   &=_{(1)} N+ \E*{I} \E*{  \sum_{i=1}^k \sum_{t\in \branchset{}} Q^{i, t}   }\\
&=_{(2)} N+   \dTV{\mu,\nu}\cdot \sum_{i=1}^k \E*{\sum_t      \E{Q^{i, t} \cdot\1[t\in \branchset{}] }}   \\
&\leq_{(3)} N + \dTV{\mu,\nu}\cdot \sum_{i=1}^k \E*{\card{\branchset{}} \parens*{\sum_{\ell=0}^{h-1} (1+\rho)^\ell} N_i}\\
&= N + \dTV{\mu,\nu}\E{\card{\branchset{}}}\cdot  \parens*{\sum_{\ell=0}^{h-1} (1+\rho)^\ell} N\\
&\leq_{(4)} {\sum_{\ell=0}^{h} (1+\rho)^\ell} N.
    \end{align*}
   where (1) is because the indicator $I$ is independent of the events in the "Else" clause in \cref{alg:rs2}, (2) is because of \cref{prop:srs-is-correct} that $\E{I}=\dTV{\mu,\nu}$, (3) is because \[\E{Q^{i,t} \cdot \1[t\in \mathcal{M}] } \leq \E{Q^{i,t}} \cdot \1[t\in \mathcal{M}]\leq N_i\parens*{\sum_{\ell=0}^{h-1} (1+\rho)^\ell} \cdot \1[t\in \mathcal{M}]\] by the induction hypothesis for the children of the root,
   and (4) is because 
   \begin{equation}\label{eq:bound on the number of threads}
       \E{\card{\mathcal{M}}}\leq \frac{1+\rho}{\dTV{\mu,\nu}}
   \end{equation}
   Let the index of the first fallback call that accepts be $t_*\geq 1,$ and let $r$ be the smallest $r\in\Z$ so that $t_* \leq \lceil (1+\rho)^{r}\rceil-1.$ Note that since $ t\geq 1,$ we have $r \geq 1.$
   The total number of fallback calls generated by the root is \[\card{\mathcal{M}} = \lceil (1+\rho)^{r}\rceil-\lceil (1+\rho)^0\rceil  = \lceil (1+\rho)^{r}\rceil - 1.\] 
   By the definition of $r$, $t_* \geq \lceil (1+\rho)^{r-1}\rceil$, we have
   \[\card{\mathcal{M}} = \lceil (1+\rho)^{r}\rceil -1  \leq \lceil (1+\rho)\lceil (1+\rho)^{r-1}\rceil \rceil -1 \leq (1+\rho)\lceil (1+\rho)^{r-1}  \rceil \leq (1+\rho) t_*.\] 
    Thus, we have shown $ \card{\mathcal{M}}  \leq(1+\rho) t_*.$
   Taking expectation over $t_*,$ and noting that $\E{t_*} = \frac{1}{ \dTV{\mu, \nu}}$ by \cref{prop:srs-is-correct} since each  fallback call at the root accepts with probability $ \dTV{\mu, \nu}$, yields the desired bound. 
    \end{proof} 

     \subsection{Bounding the number of rounds}

     \begin{lemma} \label{lem:expected round bound}
         Consider a fallback tree with height $h,$ $N$ leaf nodes, and maximum branching factor $D.$   
 Let
 $T$ be the number of rounds of $\RS{}$ at the root with parameter $\rho >0.$
 Under \cref{assumption:formal}, the expected number of rounds is bounded by:
         \begin{equation}
             \E{T}
             \leq O\parens*{ 1 + D  \sum_{\ell = 0}^{h-1} (1+\rho)^{\ell}  \parens*{ \parens*{2 +\frac{\log N + 1}{\log(1+\rho)}} \E*_{Z\sim \mu} {\sum_{S_\ell \in \mathcal{C}_\ell (\textbf{root})}  \dTV*{\mu_{S_\ell}\parens*{\cdot\given Z_{<S_\ell}},\nu_{S_\ell}\parens*{\cdot\given Z _{<S_\ell}}}} + \frac{1}{e\log (1+\rho) } } }
         \end{equation}
        
     \end{lemma}
    \begin{proof}

   We will show that
\begin{equation}\label{eq:expected round bound first eq}
     \E{T} \leq O\parens*{1 + D \sum_{\ell = 0}^{h-1} (1+\rho)^{\ell}  \E*_{Z\sim \mu} {\sum_{S_\ell \in \mathcal{C}_\ell (\textbf{root})}  \varphi(\dTV*{\mu_{S_\ell}\parens*{\cdot\given Z_{<S_\ell}},\nu_{S_\ell}\parens*{\cdot\given Z_{<S_\ell}}})}}
\end{equation}
 where $\varphi:[0,1]\to \R$ is defined by \[\varphi(x) = \begin{cases} x(2+\log_{1+\rho} \frac{1}{x}) &\text{ if } x >0, \\ 0 &\text{ if } x = 0.\end{cases}\]

 Given \cref{eq:expected round bound first eq}, the lemma follows 
    by applying the following helper proposition 
    \begin{proposition}\label{prop:expected round helper}

Define the function $\tilde{\varphi}:[0,+\infty]\to \R, \tilde{\varphi} (x) =\begin{cases} x \log \frac{1}{x},&\text{ if } x>0 \\ 0 &\text{ if } x= 0\end{cases}$ where $\log$ denotes the natural logarithm. 
 
We have the following inequality:
\begin{align*}
&\E*_{Z\sim \mu} {\sum_{S_\ell \in \mathcal{C}_\ell(\textbf{root})}\tilde{\varphi} (\dTV*{\mu_{S_\ell}\parens*{\cdot\given Z_{<S_\ell}},\nu_{S_\ell}\parens*{\cdot\given Z_{<S_\ell}}}) } \\
\leq  &\E*_{Z\sim \mu} {\sum_{S_\ell \in \mathcal{C}_\ell(\textbf{root})}\dTV*{\mu_{S_\ell}\parens*{\cdot\given Z_{<S_\ell}},\nu_{S_\ell}\parens*{\cdot\given Z_{<S_\ell}}} }(\log |\mathcal{C}_\ell (\textbf{root})| +1)  + \frac{1}{e}.
\end{align*}
\end{proposition}
The proof of
    \cref{prop:expected round helper} is deferred to \cref{sec:misssing proofs}.

Now, we prove \cref{eq:expected round bound first eq}.
We first define some relevant notations.
    For a node $S,$ and recursive call index $t$ (see \cref{def:recursive call index}), in the call to node $S$ with index $t$, let $I^{S,t}$ be the indicator for the event that $\Blue{X}$ is not accepted, $\mathcal{M}^{S,t}$ be the set of fallback recursive call indices issued by this call, and $m^{S,t}$ be the number of batches\footnote{This corresponds to the counter $r$ in the algorithm}. Note that, if $ S$ is a leaf node, then $I^{S,t} $ and $m^{S,t} $ are always $0.$ 
    
     W.l.o.g., let us assume that in \cref{assumption:formal}, the number of rounds called locally at each node is exactly $1;$ this is without loss of generality  since all bounds will translate up to constant factors. 
    Consider nodes $S_\ell$ and $S_r$ where $S_r$ is a descendant fo $S_\ell$, and index $t_\ell$ of node $S_\ell.$ We define
  \begin{equation}\label{eq:definition of Z in round induction}
      W^{(S_\ell,t_\ell)\to S_r} = \sum_{t_{\ell+1}\in \mathcal{M}^{S_{\ell}, t_\ell},\: t_{\l+2}\in \mathcal{M}^{S_{\l +1},  \: t_{\l +1}},\: \dots,\: t_r\in \mathcal{M}^{S_{r-1}, t_{r-1}}} I^{S_\ell, t_\ell}I^{S_{\ell+1}, t_{\ell+1}}\cdots I^{S_r, t_r} m^{S_r,t_r}. 
  \end{equation}
 For intuition, think of $ W^{(S_\ell,t_\ell)\to S_r} $ as the sum of the cost of the while-loop of all calls to $\RS{S_r,\cdot} $ generated by the call to $\RS{S_\ell,\cdot}$ with recursive call index $t_\ell.$
 
    Note that  we have the recursive relation
    \begin{equation}\label{eq:recursive relation for Z}
         W^{(S_\ell,t_\ell)\to S_r} = I^{S_\ell, t_\ell}\cdot \sum_{t_{\ell+1}\in \mathcal{M}^{S_{\ell}, t_\ell}} W^{(S_{\ell+1}, t_{\ell+1}) \to S_r}
    \end{equation}
   
    We will prove that for any level $\ell$-th descendant $S_\ell $ of the root and any index $t_\ell,$ that the number of rounds used by the call to node $S_\ell$ of index $t_\ell$, denoted by $T^{S_\ell,t_\ell}$, is bounded by:
    \begin{equation}\label{eq:round complexity main induction}
        T^{S_\ell, t_\ell}\leq 1 + \sum_{r =\ell}^{h-1}\sum_{S_r \in \mathcal{C}_{r-\ell} (S_\ell)} C(S_r) W^{(S_\ell,t_\ell)\to S_r}.
    \end{equation}

The base case is $\ell = h,$ at which point the statement is trivially true.
Indeed, the statement is then equivalent to $T^{S_h, t_h}\leq 1 $, which is true by the assumption that each leaf node costs $1$ round.

Suppose we have proved the statement for all $\ell \geq 1.$ We show that the statement holds for $\ell =0$, i.e., at the root $S_0 \equiv\textbf{root}$ with index $t_0\equiv 0.$

In the execution of $\RS{}$ at the root node, let $I\equiv I^{S_0,t_0}$ be the event that $X$ is not accepted,  $\mathcal{M}\equiv \mathcal{M}^{S_0,t_0}$ be the set of recursive call indices of the fallback calls, and $m$ be the number of batches.
    For $0\leq j< m$, let $\mathcal{M}^{ (j)}$ denote the $j$-th batch of parallel recursive calls to the $\fallback{}$ method at that root. Note that within each batch, recursive calls are made in parallel, but each batch is only executed sequentially after the previous batch. Thus, we can bound the round complexity of the root by
        \begin{align*}
        T &\leq 1 + I \cdot
    \sum_{0\leq j<m} \max\set*{\sum_{S_1 \in \mathcal{C}(\textbf{root})}T^{S_1, t_1}\given t_1 \in \mathcal{M}^{(j)}}\\
    &\leq_{(1)} 1 + I \cdot
     \sum_{0\leq j<m} \max\set*{   \sum_{S_1 \in \mathcal{C}(\textbf{root})}  (1+  \sum_{r =1}^{h-1}\sum_{S_r \in \mathcal{C}_{r-1} (S_1)} C(S_r) W^{(S_1,t_1)\to S_r} ) \given t_1 \in \mathcal{M}^{(j)}}\\
     &= 1 + I \cdot \sum_{0\leq j<m}\parens*{ C(S_0) + \max\set*{  \sum_{S_1 \in \mathcal{C}(\textbf{root})}    \sum_{r =1}^{h-1}\sum_{S_r \in \mathcal{C}_{r-1} (S_1)} C(S_r) W^{(S_1,t_1)\to S_r}    \given t_1 \in \mathcal{M}^{(j)}}}\\
     &= 1 + I \cdot \sum_{0\leq j<m} \max\set*{  \sum_{S_1 \in \mathcal{C}(\textbf{root})}    \sum_{r =1}^{h-1}\sum_{S_r \in \mathcal{C}_{r-1} (S_1)} C(S_r) W^{(S_1,t_1)\to S_r}    \given t_1 \in \mathcal{M}^{(j)}}\\
     &\leq_{(2)} 1 +  I \cdot m \cdot C(S_0) + I\cdot \sum_{t_1 \in \mathcal{M}^{S,t}}\sum_{S_1 \in \mathcal{C}(\textbf{root})}  \sum_{r =1}^{h-1}\sum_{S_r \in \mathcal{C}_{r-1} (S_1)} C(S_r) W^{(S_1,t_1)\to S_r}   \\
     &= 1+ I\cdot m\cdot C(S_0) + \sum_{S_1 \in \mathcal{C}(\textbf{root})} \sum_{r =1}^{h-1}\sum_{S_r \in \mathcal{C}_{r-1} (S_1)}    (I \cdot\sum_{t_1 \in \mathcal{M}^{S,t}}W^{(S_1,t_1)\to S_r} ) C(S_r)\\
     &=_{(3)} 1 +  W^{(S_0,t_0) \to S_0}  C(S_0) +\sum_{r=1}^h \sum_{S_r \in \mathcal{C}_{r} (S_0)} W^{(S_0,t_0)\to S_r}C(S_r)
\end{align*}
where (1) is by applying the induction hypothesis for each children $S_1$ of $S,$ (2) is due to
\begin{align*}
    &\sum_{0\leq j<m} \max\set*{  \sum_{S_1 \in \mathcal{C}(\textbf{root})}    \sum_{r =1}^{h-1}\sum_{S_r \in \mathcal{C}_{r-1} (S_1)} C(S_r) W^{(S_1,t_1)\to S_r}    \given t_1 \in \mathcal{M}^{(j)}} \\
    \leq &\sum_{0\leq j<m} \sum_{t_1\in \mathcal{M}^{(j)}} \sum_{S_1 \in \mathcal{C}(\textbf{root})}    \sum_{r =1}^{h-1}\sum_{S_r \in \mathcal{C}_{r-1} (S_1)} C(S_r) W^{(S_1,t_1)\to S_r}   \\
    = &  \sum_{t_1 \in \mathcal{M}^{S,t}}\sum_{S_1 \in \mathcal{C}(\textbf{root})}  \sum_{r =1}^{h-1}\sum_{S_r \in \mathcal{C}_{r-1} (S_1)} C(S_r) W^{(S_1,t_1)\to S_r} 
\end{align*}
and (3) is due to the recursive relation for $W$ (\cref{eq:recursive relation for Z}). So, we have proven \cref{eq:round complexity main induction} at the root, which completes the inductive argument.

We finish the proof by showing that, for any descendant $S_r\in \mathcal{C}_{r-\ell}(S_\ell)$ of $S_\ell:$
\begin{equation}\label{eq:bounding Z}
    \E{W^{(S_\ell,t_\ell)\to S_r} } \leq (1+\rho)^{r-\ell} \E*_{Z\sim \mu} {\varphi(\dTV*{\mu_{S_r}\parens*{\cdot\given Z_{<S_r}},\nu_{S_r}\parens*{\cdot\given Z_{<S_r}}} ) } 
\end{equation}
We induct on $\ell\in \set{0,\cdots, h}$ and $r \in \set{\ell, \cdots h}.$ 
We first check the base cases $r=\ell$ for all $\ell \in \set{0,\cdots, h}.$ \cref{eq:bounding Z} is equivalent to 
\begin{equation}\label{eq:bounding Z base case}
    \E{I^{S_r,t_r} m^{S_r,t_r}}\leq 
    \E*_{Z\sim \mu} {\varphi(\dTV*{\mu_{S_r}\parens*{\cdot\given Z_{<S_r}},\nu_{S_r}\parens*{\cdot\given Z_{<S_r}}} ) } 
\end{equation}
\begin{proof}[Proof of \cref{eq:bounding Z base case}]
To prove \cref{eq:bounding Z base case}, we show that for every input $Z_{< S_r}$ to the call to node $S_r$ with index $t_r,$ with $\mu = \mu_{S_r}\parens*{\cdot\given Z_{<S_r}}, \nu = \nu_{S_r}\parens*{\cdot\given Z_{<S_r}}$ and $d_{TV} = \dTV {\nu ,\mu},$ 
\[ \E{I^{S_r,t_r} m^{S_r,t_r}}\leq \begin{cases} 0 &\text{if } d_{TV} =0 \\
d_{TV} (2 + \log_{1+\rho} \frac{1}{d_{TV}})  &\text{if } d_{TV} > 0\end{cases}.\]
\cref{eq:bounding Z base case} follow since by \cref{prop:srs-is-correct}, the input $Z_{<S_r}$ has the same distribution as $Z'_{<S_r}$ where $Z'$ is sampled from $\mu.$  

Indeed, if $d_{TV} =0,$ then by \cref{prop:srs-is-correct}, the event $I^{S_r,t_r}$ happens with probability $d_{TV} = 0$,  so we have $\E{I^{S_r,t_r} m^{S_r,t_r} }= 0.$

Now, suppose $d_{TV} =0 >0.$
Note that the event $I^{S_r,t_r}$ is independent of the events in "Else" clause of \cref{alg:rs2}, and $\E{I^{S_r,t_r}=1} = \dTV{\mu,\nu} =d_{TV}  $ by \cref{prop:srs-is-correct}. Furthermore, \[m^{S_\ell,t_\ell}  \leq  1 + \log_{1+\rho} |\mathcal{M}^{S_\ell,t_\ell}|.\] Therefore, we obtain: 
\begin{align*}
    \E{I^{S_r,t_r} m^{S_r,t_r}} &= \E{I^{S_r,t_r}} \E{m^{S_r,t_r}} \leq  d_{TV} \E{1 + \log_{1+\rho} |\mathcal{M}^{S_\ell,t_\ell}|}\\ &\leq_{(1)} d_{TV} (1+ \log_{1+\rho} (\E{|\mathcal{M}^{S_\ell,t_\ell}|}) \leq_{(2)} d_{TV}(1+ \log_{1+\rho} \frac{1+\rho}{d_{TV}}) = d_{TV} (2 + \log_{1+\rho} \frac{1}{d_{TV}})
\end{align*}
where (1) is by Jensen's inequality, and (2) is by $\E{|\mathcal{M}^{S_\ell,t_\ell}|} \leq \frac{1+\rho}{d_{TV}} $ (see \cref{eq:bound on the number of threads}).
 \end{proof}

Now, we prove \cref{eq:bounding Z} by induction.
Suppose we have proved \cref{eq:bounding Z} for all $\ell \geq 1$ and $r\geq \ell.$ We show that \cref{eq:bounding Z} holds for $\ell = 0$ i.e. at the node $S.$ Let $S_1 \in \mathcal{C}(S)$ be the unique child of $S$ which is also an ancestor of $S_r,$ we use the recursive relation for $W$ (\cref{eq:recursive relation for Z}) to obtain
\begin{align*}
      \E{W^{(S_0,t_0)\to S_r}} = \E{I^{S_0,t_0}\cdot \sum_{t_1 \in \mathcal{M}^{S_0,t_0}} W^{(S_1, t_1)\to S_r} } = \E{I^{S_0,t_0}} \cdot \E{ \sum_{t_1 \in \mathcal{M}^{S_0,t_0}} W^{(S_1, t_1)\to S_r} }
\end{align*}
where the second equality is because the indicator $I^{S_0,t_0}$ is independent of the events in the "Else" clause of \cref{alg:rs2}. Using the fact that $ \E{I^{S_0,t_0}} = \dTV{\mu,\nu}$ where $\mu,\nu$ are the target and the speculative distribution at $S_0$ respectively, and the induction hypothesis for $S_1$, we obtain
\begin{align*}
    &\E{W^{(S_0,t_0)\to S_r} } \\
    &= \dTV{\mu,\nu}\cdot \E{ \sum_{t_1 \in \mathcal{M}^{S_0,t_0}} W^{(S_1, t_1)\to S_r} m^{S_r,t_r}} \\
    &\leq \dTV{\mu,\nu} \E*{ \card{\mathcal{M}^{S_0,t_0}}(1+\rho)^{r-1} \E*_{Z\sim \mu} {\varphi(\dTV*{\mu_{S_r}\parens*{\cdot\given Z_{<S_r}},\nu_{S_r}\parens*{\cdot\given Z_{<S_r}}})  } } \\
    &=_{(1)}  \dTV{\mu,\nu} \E*{ |\mathcal{M}^{S_0,t_0} |} (1+\rho)^{r-1} \E*_{Z\sim \mu} {\varphi(\dTV*{\mu_{S_r}\parens*{\cdot\given Z_{<S_r}},\nu_{S_r}\parens*{\cdot\given Z_{<S_r}}} )  } \\
    &\leq_{(2)} (1+\rho)^{r} \E*_{Z\sim \mu} {\varphi(\dTV*{\mu_{S_r}\parens*{\cdot\given Z_{<S_r}},\nu_{S_r}\parens*{\cdot\given Z_{<S_r}}})  } 
\end{align*}
where in (1), we pull out the inner expectation as it is invariant of the events in the algorithm, and in (2), we use the fact that $  \E*{ |\mathcal{M}^{S_0,t_0} |} \leq \frac{1 + \rho}{\dTV{\mu,\nu}} $ (see \cref{eq:bound on the number of threads}).
\end{proof}

We show an alternative bound on the number of rounds that depends poly-logarithmically on the number of queries, which will be useful to obtain high-probability tail bound for the number of rounds (see \cref{thm:autoregression with high probability}).
\begin{lemma}\label{lem:worst case bound on number of rounds}
Consider a fallback tree with height $h,$ $N$ leaf nodes, and maximum branching factor $D.$
 Let
 $Q, T$ be the number of queries and rounds of $\RS{}$ at the root with parameter $\rho >0.$
 Under \cref{assumption:formal}, the number of rounds is bounded by:
 \[T \leq O((1 + \log_{1+\rho} Q)^{h} N D h). \]
\end{lemma}
\begin{proof}
We use the same notations and setup as in proof of \cref{lem:expected round bound}.
     Consider nodes $S_\ell$ and $S_r$ where $S_r$ is a descendant fo $S_\ell$, and index $t_\ell$ of node $S_\ell.$ We define
     \begin{equation}
    F^{(S_\ell,t_\ell)\to S_r }  = \max\set*{m^{S_\ell, t_\ell}I^{S_\ell, t_\ell}m^{S_{\ell+1}, t_{\ell+1}}I^{S_{\ell+1}, t_{\ell+1}}\cdots m^{S_r, t_r}I^{S_r, t_r}\given t_{\ell+1}\in \mathcal{M}^{S_\l, t_\l}, \dots, t_r\in \mathcal{M}^{S_{r-1}, t_{r-1}}}.
 \end{equation}
 Note that we have the recursive relation
 \begin{equation}\label{eq:recursive relation for F}
     F^{(S_\ell,t_\ell)\to S_r } = m^{S_\ell, t_\ell}I^{S_\ell, t_\ell} \cdot \max\set{F^{(S_{\ell+1},t_{\ell+1})\to S_r }  \given t_{\ell+1} \in  \mathcal{M}^{S_\l, t_\l}} 
 \end{equation}
 W.l.o.g., let us assume that in \cref{assumption:formal}, the number of rounds called locally at each node is exactly $1;$ this is without loss of generality  since all bounds will translate up to constant factors.
     We will prove for any level $\ell$-th descendant $S_\ell $ of the root and any index $t_\ell,$ and $T^{S_\ell,t_\ell}$ being the number of rounds used by the call to node $S_\ell$ of index $t_\ell$, that
    \begin{equation}\label{eq:round complexity induction for worst case bound}
        T^{S_\ell, t_\ell}\leq  1 + \sum_{r =\ell}^{h-1}\sum_{S_r \in \mathcal{C}_{r-\ell} (S_\ell)} C(S_r) F^{(S_\ell,t_\ell)\to S_r} 
    \end{equation}
    We show how to derive the lemma from \cref{eq:round complexity induction for worst case bound}.  For any node $S_\ell$ and index $t_\ell,$ each fallback call issued by the call at node $S_\ell$ with index $t_\ell$ costs at least 1 query, so 
\[\card{\mathcal{M}^{S_\ell, t_\ell}}\leq Q. \]
Next, we have:
\[m^{S_\ell,t_\ell} \leq 1 + \log_{1+\rho} \card{\mathcal{M}^{S_\ell, t_\ell}} \leq 1 + \log_{1+\rho} Q \]
thus, using the fact that the fallback tree contains at most $N(h+1)$ nodes, we 
\[ T\leq O(1+\sum_{r=0}^{h-1} C(S_r) F^{(S_0,t_0)\to S_r})  \leq O((1 + \log_{1+\rho} Q)^{h} N D h).  \]
The proof of \cref{eq:round complexity induction for worst case bound} is similar to the proof of \cref{eq:round complexity main induction} in \cref{lem:expected round bound}.

The base case is $\ell = h$ and $S_\ell$ is a leaf node, at which point the statement is equivalent to $T^{S_h, t_h}\leq 1 $, which is true by the assumption that each leaf node costs $1$ round.

In the execution of $\RS{}$ at the root node, let $I \equiv I^{S_0,t_0}$ be the event that $X$ is accepted,  $\mathcal{M} \equiv \mathcal{M}^{S_0,t_0}$ be the set of fallback calls, and $m\equiv m^{S_0,t_0}$ be the number of batches.
    For $0\leq j< m$, let $\mathcal{M}^{ (j)}$ denote the $j$-th batch of parallel recursive calls to the $\fallback{}$ method at that root. Note that within each batch, recursive calls are made in parallel, but each batch is only executed sequentially after the previous batch. Thus, we can bound the round complexity of the root by
        \begin{align*}
        T &\leq 1 + I \cdot
    \sum_{0\leq j<m} \max\set*{\sum_{S_1 \in \mathcal{C}(\textbf{root})}T^{S_1, t_1}\given t_1 \in \mathcal{M}^{(j)}}\\
    &\leq_{(1)} 1 + I \cdot
     \sum_{0\leq j<m} \max\set*{   \sum_{S_1 \in \mathcal{C}(\textbf{root})}  (1+  \sum_{r =1}^{h-1}\sum_{S_r \in \mathcal{C}_{r-1} (S_1)} C(S_r) F^{(S_1,t_1)\to S_r} ) \given t_1 \in \mathcal{M}^{(j)}}\\
     &= 1 + I \cdot \sum_{0\leq j<m}\parens*{ C(S_0) + \max\set*{  \sum_{S_1 \in \mathcal{C}(\textbf{root})}    \sum_{r =1}^{h-1}\sum_{S_r \in \mathcal{C}_{r-1} (S_1)} C(S_r) F^{(S_1,t_1)\to S_r}    \given t_1 \in \mathcal{M}^{(j)}}}\\
     &= 1 +  I \cdot m \cdot C(S_0) + I \cdot \sum_{0\leq j<m} \max\set*{  \sum_{S_1 \in \mathcal{C}(\textbf{root})}    \sum_{r =1}^{h-1}\sum_{S_r \in \mathcal{C}_{r-1} (S_1)} C(S_r) F^{(S_1,t_1)\to S_r}    \given t_1 \in \mathcal{M}^{(j)}}\\
     &\leq_{(2)} 1 +  I \cdot m \cdot C(S_0) + I \cdot m \cdot \max\set*{\sum_{S_1 \in \mathcal{C}(\textbf{root})}  \sum_{r =1}^{h-1}\sum_{S_r \in \mathcal{C}_{r-1} (S_1)} C(S_r) F^{(S_1,t_1)\to S_r} \given t_1 \in \mathcal{M}}  \\
     &\leq 1+ I\cdot m\cdot C(S_0) + \sum_{S_1 \in \mathcal{C}(\textbf{root})} \sum_{r =1}^{h-1}\sum_{S_r \in \mathcal{C}_{r-1} (S_1)}    ( I \cdot m\cdot   \max \set*{ F^{(S_1,t_1)\to S_r} \given t_1 \in \mathcal{M}} C(S_r)   )   \\
     &=_{(3)} 1 +  F^{(S_0,t_0) \to S_0}  C(S_0) +\sum_{r=1}^{h-1} \sum_{S_r \in \mathcal{C}_{r} (S_0)} F^{(S_0,t_0)\to S_r}C(S_r)
\end{align*}
where (1) is by applying the induction hypothesis for each children $S_1$ of $S,$ (2) is due to
\begin{align*}
    &\sum_{0\leq j<m} \max\set*{  \sum_{S_1 \in \mathcal{C}(\textbf{root})}    \sum_{r =1}^{h-1}\sum_{S_r \in \mathcal{C}_{r-1} (S_1)} C(S_r) F^{(S_1,t_1)\to S_r}    \given t_1 \in \mathcal{M}^{(j)}} \\
    \leq &m \cdot \max_{0\leq j<m} \max\set*{  \sum_{S_1 \in \mathcal{C}(\textbf{root})}    \sum_{r =1}^{h-1}\sum_{S_r \in \mathcal{C}_{r-1} (S_1)} C(S_r) F^{(S_1,t_1)\to S_r}    \given t_1 \in \mathcal{M}^{(j)}}   \\
    = &  m \cdot \max\set*{  \sum_{S_1 \in \mathcal{C}(\textbf{root})}    \sum_{r =1}^{h-1}\sum_{S_r \in \mathcal{C}_{r-1} (S_1)} C(S_r) F^{(S_1,t_1)\to S_r}    \given t_1 \in \mathcal{M}}
\end{align*}
and (3) is due to the recursive relation for $F$ (\cref{eq:recursive relation for F}). So, we have proven \cref{eq:round complexity main induction} at the root, which completes the inductive argument.
\end{proof}

%% file: new-autoregression.tex
\section{Parallelizing any-order autoregressive sampling}\label{sec:autoregression}
In this section, we instantiate and analyze $\RS{}$ with respect to the coordinate denoiser. Consider a target distribution  $\pi: [q]^n \to \R_{\geq 0}.$ Suppose we are given access to the conditional marginal distributions  $\pi_i(\cdot |X_S) $ of $\pi$ via the following types of queries:
\begin{itemize}
    \item Density query: Given a subset $S\subseteq [n],$ an index $i\not\in S,$ a configuration $x_S\in [q]^S$ such that $\pi_S(x_S) = \P_{X\sim \pi}{X_S = x_S}> 0$, and $ x_i\in [q]$,  return \[ \pi_i\parens{x_i \given X_S = x_S} =\P_{X\sim \pi} {X_i= x_i\given X_S = x_S}.\]
    Note that this is exactly the coordinate denoiser oracle for $\pi.$
    \item Sampling query: Given a subset $S\subseteq [n]$, an index $i\not\in S,$ and a configuration $x_S\in [q]^S$ such that $\pi_S(x_S) = \P_{X\sim \pi}{X_S = x_S}> 0,$  return a sample $ x_i \sim \pi_i\parens{\cdot \given X_S = x_S} .$ 
    
\end{itemize}
\begin{theorem}\label{thm:autoregression main detail}
There exists an algorithm that, given oracle access to the conditional marginals $\pi_i(\cdot \mid X_S)$ of a distribution $\pi: [q]^n \to \mathbb{R}_{\geq 0}$ via the queries described above, outputs a sample exactly distributed according to $\pi$. The expected number of rounds is $O(\sqrt{n \log q} \log^3 n)$ and the expected number of queries is $O(n \log n)$ queries.
\end{theorem}
We can also get a high probability bound on the number of rounds at the cost of a slightly worse dependency on $n.$
\begin{theorem}\label{thm:autoregression with high probability}
Consider $\pi:[q]^n \to \R_{\geq 0}.$
For every parameter $h \in \mathbb{N}_{\geq 1}$ where $h^{2h} \leq n$, there exists an algorithm that, given oracle access to the conditional marginals $\pi_i(\cdot \mid X_S)$ of $\pi$ via the queries described above, outputs a sample exactly distributed according to $\pi$.
The expected number of rounds is $O(n^{\frac{1}{2} + \frac{1}{2h}} h^2 \log n \sqrt{\log q})$ and the expected number of queries is $O(nh).$

Moreover, for any $\epsilon > 0$, with probability at least $1 - \epsilon$, the number of rounds is at most
\[ O\parens*{  n^{\frac{1}{2} +\frac{1}{2h}} \parens*{h^2\log n \sqrt{\log q} + (4h)^h \log^h\parens*{\frac{n}{\epsilon}}} \cdot \log \parens*{\frac{1}{\epsilon}} }.\]
\end{theorem}
\begin{remark}\label{remark:autoregression with high probability parameter choice}
    Different choice of the parameter $h$ yields different trade-offs. Take $h = O(1),$ then  with probability $\geq 1-\epsilon$, the number of rounds is $O (n^{\frac{1}{2} + \frac{1}{2h} } \poly\log\frac{n}{\epsilon}).$ For $h = \sqrt{\log n}$ and $ \epsilon  = \omega(2^{-2^{\sqrt{\log n}}})$\footnote{This allows $\epsilon$ to decay superpolynomially in $n$}, the number of rounds is $O(n^{1/2+o(1)})$ with probability $1-\epsilon.$ 
\end{remark}
\begin{remark}\label{remark:density vs. sampling query}
    Since each sampling query can be implemented in one round and $q$ density queries by evaluating $\pi_i (x_i |X_S =x_S) $ for all $[x_i]\in [q]$ in parallel, any algorithm using both sampling and density query (e.g. those in \cref{thm:autoregression main detail,thm:autoregression with high probability}) can be converted into those using only density queries (which exactly corresponds to the coordinate denoiser oracle) with the same number of rounds, where the number of queries increased by a factor of $q.$ 
\end{remark}
\begin{remark}\label{remark:implementation details coordinate denoiser}
We verify \cref{assumption:formal} and discuss the implementation of our algorithm in the PRAM model. We recall the definitions of $\nu_S\parens{\cdot \given Z_{<S}},\mu_S\parens{\cdot \given Z_{<S}} $ from \cref{sec:spec-scheme-coordinate}. Given access to the density and sampling query as defined in the beginning of this section, we can sample from $ \nu_S\parens{\cdot \given Z_{<S}},$ and compute the probabilities $ \nu_S\parens{Z_S \given Z_{<S}}, \mu_S\parens{\cdot \given Z_{<S}}$ in $1$ round, using $|S|$ density and sampling queries, $O(\log |S|) = O(\log n)$ parallel time\footnote{Since $\nu_S\parens{\cdot \given Z_{<S}}$ and $ \mu_S\parens{\cdot \given Z_{<S}}$ are the products of $|S|$ conditional marginals}, and $O(|S|)$ operations.

Finally, we show how to implement the sampling oracle given the coordinate denoiser oracle (which is exactly the density oracle).
     To sample from $\pi_i (x_i |X_S =x_S) ,$ we evaluate the $q$ conditional marginals $\hat{p}_i := \pi_i (x_i |X_S =x_S) $ for all $x_i \in [q]$ in parallel. We then construct the cumulative distribution function of $\pi_i (\cdot |X_S =x_S) $ i.e. the array $[ \hat{p}_1, \hat{p}_1+\hat{p}_2, \cdots, \hat{p}_1 + \cdots + \hat{p}_q].$ Next, we draw $\hat{u} \sim \Uniform{[0,1]},$ and output the unique $r\in [q]$  s.t. $ \sum_{j=1}^{r-1} \hat{p}_j  \leq  \hat{u} < \sum_{j=1}^{r} \hat{p}_j. $ The procedure takes $O(1)$ rounds and $O(q)$ conditional marginal evaluations.  In the PRAM model, this step can be implemented in $O(\log q)$ parallel time and $O(q)$ operations by (1) computing each entry in $[ \hat{p}_1, \hat{p}_1+\hat{p}_2, \cdots, \hat{p}_1 + \cdots + \hat{p}_q]$ in $O(\log q)$ parallel time and (2) search for $j$ in the sorted array $[ \hat{p}_1, \hat{p}_1+\hat{p}_2, \cdots, \hat{p}_1 + \cdots + \hat{p}_q]$ of length $q$ in $O(\log q)$ time. This can be done by constructing a binary tree with $q$ leaves where the leaves correspond to $ \hat{p}_1, \cdots, \hat{p}_q$ and each node stores the sum of the leaves in its subtree. The tree is built bottom-up, and the search is performed by traversing from the root to leaves: given a value $\hat{u}\in [0,1],$ starting from the root, recursively move to the left child if $\hat{u}$ is strictly less than the value stored there, and to the right child otherwise, and output the index of the leaf that the algorithm arrives at. 

     Hence, \cref{thm:autoregression main detail} implies a sampling algorithm for $\pi:[q]^n \to\R_{\geq 0}$ that,  in expectation, uses $O(\sqrt{n \log q} \log^3 n)$ rounds, $O(q n \log n)$ coordinate denoiser queries, $O(\sqrt{n \log^3 q} \log^4 n)$ parallel time and $O(q n \log n)$ total operations in the PRAM model. Similar statements can be made for \cref{thm:autoregression with high probability}.
\end{remark}
Our results also apply to the case when the oracle queries are approximate, as discussed in \cref{remark:approximate autoregression oracle}.

\begin{proof}[Proof of \cref{thm:autoregression main detail}]
W.l.o.g. assume $n\geq 2$ and is a power of $2;$ otherwise, we can replace $n$ with $n':= 2^{\lceil\log n\rceil} ,$ which satisfies $n\leq n'\leq 2n,$ by appending fixed entries to vectors in $[q]^n$, defining a new distribution $\pi':[q]^{n'}\to\R_{\geq 0}$ whose oracle queries reduce to those of $\pi$.

    Let the fallback tree be the full binary tree with $n$ leaf nodes and height $ h=  \log n \in \N_{\geq 1}.$ 
    
    Our algorithm samples a uniformly random permutation $\sigma$ of $[n],$ then runs $\RS{}$ at the root of the fallback tree (assuming density and sampling query access), with parameter $\rho = 1/h.$ For the verification of \cref{assumption:formal}, see \cref{remark:implementation details coordinate denoiser}. Let $Q,T$ be the number of queries and rounds of $\RS{}$ at the root of the fallback tree. By \cref{lem:expected number of queries}, the expected queries complexity is bounded by \[\E_{\sigma}{\E{Q}}\leq O(n\log n).\]

    To bound the expected number of rounds, we need the following inequality, that for any level $\ell, 0\leq \ell \leq h-1$:
    \begin{equation}\label{eq:main potential bound for autoregression}
        \E*_{\sigma}{ \E*_{Z\sim \mu} {\sum_{S_\ell \in \mathcal{C}_\ell (\textbf{root})}  \dTV*{\mu_{S_\ell}\parens*{\cdot\given Z_{<S_\ell}},\nu_{S_\ell}\parens*{\cdot\given Z_{<S_\ell}}}}}  \leq O(n\sqrt{\log q}).
    \end{equation}
    Given \cref{eq:main potential bound for autoregression},
    by \cref{lem:expected round bound}, the expected number of rounds is bounded by:
    \begin{align*}
             \E_{\sigma}{\E{T}} &\leq 1 + 2  \sum_{\ell = 0}^{h-1} (1+\rho)^{\ell}  \left( (2 +\frac{\log n + 1}{\log(1+\rho) }) \E*_{\sigma}{\E*_{ Z\sim \mu} {\sum_{S_\ell \in \mathcal{C}_\ell (\textbf{root})}  \dTV*{\mu_{S_\ell}\parens*{\cdot\given Z_{<S_\ell}},\nu_{S_\ell}\parens*{\cdot\given Z_{<S_\ell}}}}} + \frac{1}{e\log(1+\rho)} \right)\\
             &\leq O(\sqrt{n \log q }\log^3 n)
         \end{align*}
         Finally, we prove \cref{eq:main potential bound for autoregression}.  We will use \cref{lem:pinning lemma for induction}.
        Fix level $\ell , 0
        \leq \ell \leq h-1.$
         Let $S_1,\cdots, S_k$ be all the level-$\ell$ descendants of the root in $\mathcal{C}_\ell(\textbf{root}),$ where $\forall r \in [k]: S_r = \set{\tilde{n}(r-1) +1, \cdots, \tilde{n}}$ and $ \tilde{n} = \frac{n}{k}\in \N_{\geq 1}.$
         
         Consider $(X_1, \cdots, X_n) \sim \pi.$
          As in \cref{def:potential function}, for $r\in [k],$ let \[\phi (X_{\sigma(S_r) } | X_{\sigma(S_1) \cup \cdots \cup \sigma(S_{r-1}) }) =\sum_{i\in \tilde{\sigma}(S_r)} \I*{ (X_i |  X_{\tilde{\sigma}(S_1) \cup \cdots \cup \tilde{\sigma}(S_{r-1}) } ) ;  (X_{\sigma(S_r) \setminus \set{i} } |  X_{\sigma(S_1) \cup \cdots \cup \sigma(S_{r-1}) } )  }.\] 
          
         For any $r\in [k]$, by Pinsker's inequality, the definition of $\mu_{S_r}$ and $\nu_{S_r},$ and \cref{prop:potential function upper bound kl},  we have: \[\dTV* {\mu_{S_r}\parens*{\cdot\given Z_{<S_r}}, \nu_{S_r}\parens*{\cdot\given Z_{<S_r}}} \leq \sqrt{\DKL*{\mu_{S_r}\parens*{\cdot\given Z_{<S_r}} || \nu_{S_r}\parens*{\cdot\given Z_{<S_r}}} } \leq \sqrt{\phi (X_{\sigma(S_r) } | X_{\sigma(S_1) \cup \cdots \cup \sigma(S_{r-1}) })}\] 
         thus
         \begin{align*}
           \E*_{\sigma}{ \E*_{Z\sim \mu} {\sum_{S_\ell \in \mathcal{C}_\ell (\textbf{root})}  \dTV*{\mu_{S_\ell}\parens*{\cdot\given Z_{<S_\ell}},\nu_{S_\ell}\parens*{\cdot\given Z_{<S_\ell}}}}}&= \E*_\sigma{\sum_{r=1}^k \dTV* {\mu_{S_r}\parens*{\cdot\given Z_{<S_r}}, \nu_{S_r}\parens*{\cdot\given Z_{<S_r}}}   }\\
           &\leq \E*_{\sigma} {\sum_{r=1}^k\sqrt{\phi (X_{\sigma(S_r) } | X_{\sigma(S_1) \cup \cdots \cup \sigma(S_{r-1}) })  }}\\
             &\leq_{(1)} \sqrt{k  \E*_{\sigma} {\sum_{r=1}^k\phi (X_{\sigma(S_r) } | X_{\sigma(S_1) \cup \cdots \cup \sigma(S_{r-1}) })  } }\\
             &\leq_{(2)} \sqrt{\phi(X_{[n]})} \\
             &\leq_{(3)} \sqrt{n \log q}
         \end{align*} 
    where (1) is by Holder's inequality, (2) is by \cref{lem:pinning lemma for induction}, and (3) is by \cref{prop:potential function trivial bound}.

\end{proof}

\begin{proof}[Proof of \cref{thm:autoregression with high probability}]
     W.l.o.g. assume $D=n^{\frac{1}{2h}}\in \N;$ otherwise, we replace $n $ with $ n': = \lceil n^{\frac{1}{2h}} \rceil^{2h} ,$ which satisfies $n\leq n' \leq ( 1+ \frac{1}{n^{1/2h}})^h n\leq (1+\frac{1}{h})^h n = O(n)$ by padding the distribution $\pi $ similar to in the proof of \cref{thm:autoregression main detail}.
    
    Consider the fallback tree with $n$ leaf nodes, where the root node has $k:= \sqrt{n} = D^h$ children, and all other nodes in the tree have $D$ children. Let $S^{(1)}, \cdots, S^{(k)}$ be the children of the root. Each of the subtree rooted at $S^{(i)}$ for $i\in \set{1,\cdots k}$ has $N^{(i)}:= k$ leaf nodes, and height $h.$ 

    Our algorithm samples a uniformly random permutation $\sigma$ of $[n],$ then runs  the $\fallback{}$ in \cref{alg:rs2} corresponds to the coordinate denoiser scheme at the root of the fallback tree, with parameter $\rho = 1/h;$ let $Q, T$ be the number of queries and rounds. For $i\in \set{1,\cdots k},$  let $T^{(i)}, Q^{(i)}$ be the number of rounds and queries used by the call to $\RS{}$ at the node $S^{(i)}.$ Clearly, $T =\sum_{i=1}^k T^{(i)}$ and  $Q = \sum_{i=1}^k Q^{(i)}. $ Using the linearity of expectation and applying \cref{lem:expected number of queries} for each $Q^{(i)}$ gives
    \[ \E{Q} = \E{\sum_{i=1}^k Q^{(i)}} = \sum_{i=1}^k \E{Q^{(i)}} \leq O\parens*{\frac{(1+ \rho)^{h+1}-1}{\rho} \cdot \sum_{i=1}^k N^{(i)}} = O(n h). \]

   Fix \( \epsilon > 0 \). For each \( i \in \{1, \ldots, k\} \), let \( \mathcal{G}^{(i)} \) be the event
   the call to $\RS{}$ at node \( S^{(i)} \) uses at most
\[
C := \frac{(1+ \rho)^{h+1} - 1}{\rho} \cdot \frac{2c n}{\epsilon}
\]
queries, where $c$ is the implicit constant appearing in the bound on the expected number of queries in \cref{lem:expected number of queries}. Let \( \mathcal{G} \) be the event that \( \mathcal{G}^{(i)} \) occurs for all \( i \in \{1, \ldots, k\} \). By \cref{lem:expected number of queries}, we have
\[
C \geq \frac{2 \mathbb{E}[Q^{(i)}] \cdot k}{\epsilon}.
\]
Then, by Markov's inequality, each event \( \mathcal{G}^{(i)} \) occurs with probability at least \( 1 - \frac{\epsilon}{2k} \). Applying the union bound, the event \( \mathcal{G} \) occurs with probability at least \( 1 - \frac{\epsilon}{2} \).

Let $\tilde{T}^{(i)}:=\1[\mathcal{G}^{(i)} ] T^{(i)} .$ This is a random variable that takes the same value as $T^{(i)} $ if the event $ \mathcal{G}^{(i)}$ occurs, and take the value $0$ otherwise. By \cref{lem:worst case bound on number of rounds}, 
\begin{equation}\label{eq:uniform bound on number of rounds under good event}
    \tilde{T}^{(i)}\leq C_1: =c_1 n^{\frac{1}{2} +\frac{1}{2h}} (4h \log \frac{n}{\epsilon})^h).
\end{equation}
where $c_1$ is the implicit constant appearing in \cref{lem:worst case bound on number of rounds}.
By \cref{lem:expected round bound} and linearity of expectation,
the expected number of rounds is bounded by:
\begin{equation}\label{eq:expected round bound in high probability}
    \E_{\sigma}{\E{T}} = \E_{\sigma}{\E{\sum_{j=1}^k T^{(j)}}} \leq C_2:= c_2  n^{\frac{1}{2} +\frac{1}{2h}} h^2 \log n \sqrt{\log q}
\end{equation}
for some positive constant $c_2.$

Let $C_3 =\max\set{C_1, C_2}.$ Next, we prove by induction that for any $r \in \set{1,\cdots, k}$ and any input $x$ to node $S^{(i)}$,
\begin{equation} \label{ineq:inductive hypothesis for high probability bound}
    \P{ \sum_{j=r}^{k}\tilde{T}^{(j)}\geq 3 m C_3  \given Z_{<S^{(r)}} =x }  \leq 2^{-m}  
    \end{equation}
    where the probability is w.r.t the sub-string $\tilde{\sigma}$ of the random permutation $\sigma$ that corresponds to the set $S^{(r)}, \cdots, S^{(k)} $ and the randomness in the algorithms at nodes $S^{(j)}$ for $j\in \set{r,\cdots, k}.$

Given \cref{ineq:inductive hypothesis for high probability bound}, the high-probability bound on the number of rounds follows, since 
\begin{equation}\label{ineq:high probability bound}
    \P{T \geq 3m C_3} \leq \P{\bar{\mathcal{G}}} +  \P{\1[\mathcal{G}] \cdot T \geq 3m C_3}  \leq \frac{\epsilon}{2} + \P{\sum_{j=1}^{k}\tilde{T}^{(j)}\geq 3m C_3} \leq  \frac{\epsilon}{2} +2^{-m} \leq \epsilon
\end{equation}
when we choose $m = \lceil\log \frac{2}{\epsilon}\rceil.$

    The proof of \cref{ineq:inductive hypothesis for high probability bound} is similar to \cite[Theorem 28]{AGR24}.

\cref{ineq:inductive hypothesis for high probability bound} is trivially true for all $r\in [k]$ when $m\leq 0.$ 
Suppose we prove \cref{ineq:inductive hypothesis for high probability bound} for all $r\in [k]$ and $m-1\geq 0.$ We prove \cref{ineq:inductive hypothesis for high probability bound}  for $r=1$ and $m;$ the proof for $r\in \set{2, \cdots k}$ is analogous.

    For $i\in \set{0,\cdots, k-1},$
 let $\mathcal{A}_i$ be the event that $\sum_{j=1}^{i}\tilde{T}^{(j)} \leq 2C_3$ but $ \sum_{j=1}^{i+1}\tilde{T}^{(j)} > 2C_3,$ Intuitively, this is the event that after $  2C_3$ rounds, the algorithm has finished the execution of \emph{only} the first $i$ children of the root.

        Observe that the algorithm hasn't terminated after $2C_3$ rounds only if at least one of the events $\mathcal{A}_1, \cdots, \mathcal{A}_{k-1}$ occurs.  
    Due to \cref{eq:uniform bound on number of rounds under good event}, $\P{\mathcal{A}_0}=0.$ 
    Note that by \cref{eq:expected round bound in high probability},
    \[\E_{\sigma}{\E{\sum_{j=1}^k\tilde{T}^{(j)}}} \leq \E_{\sigma}{\E{\sum_{j=1}^k T^{(j)}}} \leq C_2\leq C_3.\] 
    By Markov's inequality and the fact that the events $(\mathcal{A}_i)_i$ are mutually exclusive, 
    \begin{equation}\label{eq:sum of Ai in high probability bound}
        \sum_{i= 1}^{k-1} \P{\mathcal{A}_i} \leq 1/2.
    \end{equation}
Conditioned on  $\mathcal{A}_i,$
we have
\[\sum_{j=1}^{k}\tilde{T}^{(j)} = \sum_{j=1}^{i}\tilde{T}^{(j)} + \tilde{T}^{(i+1)} +\sum_{j=i+2}^{k}\tilde{T}^{(j)} \leq    2C_3+ C_3 +\sum_{j=i+2}^{k}\tilde{T}^{(j)} \]
where we bound the first term using the definition of $\mathcal{A}_i$ and the second by \cref{eq:uniform bound on number of rounds under good event}. 

Thus, conditioned on $ \mathcal{A}_j,$ $ T\geq 3m C_3 $ only if $\sum_{j=i+2}^{k} \tilde{T}^{(j)}\geq 3(m-1) C_3 $ and we will show that this happens with probability $\leq 2^{-(m-1)} $ by the induction hypothesis. 

Indeed, conditioned on an input $ \hat{x}$ of $S^{(i+2)},$ the execution of $\RS{S^{(j)}, Z_{<S^{(j)}}}$ for $j\geq i+2$ are independent of the event $\mathcal{A}_i,$  
thus
\begin{align*}
    \P{\sum_{j=i+2}^{ k} \tilde{T}^{(j)} \geq 3(m-1)C_3 \given \mathcal{A}_i } &= \sum \P{Z_{<S^{(i+2)}} = \hat{x}}\P{\sum_{j=i+2}^{ k} \tilde{T}^{(j)} \geq 3(m-1)C_3 \given Z_{<S^{(i+2)}} = \hat{x}, \mathcal{A}_i} \\
    &=\sum \P{Z_{<S^{(i+2)}} = \hat{x}}\P{\sum_{j=i+2}^{ k}\tilde{T}^{(j)} \geq 3(m-1 C_3)\given Z_{<S^{(i+2)}} = \hat{x} }\\
    &\leq 2^{-(m-1)}.
\end{align*}
Thus
\[ \P{\sum_{j=1}^{k}\tilde{T}^{(j)}\geq 3C_3 m}\leq  \sum_{i=1}^{k-1} \P{\sum_{j=i}^{k}\tilde{T}^{(j)}\geq 3(m-1)C_3 \given \mathcal{A}_i } \P{\mathcal{A}_i}\leq \sum_{i=1}^{k-1} 2^{-(m-1)} \P{\mathcal{A}_i} \leq 2^{-m},\]
where the last inequality is due to \cref{eq:sum of Ai in high probability bound}.
This finishes the proof of \cref{ineq:inductive hypothesis for high probability bound} when $ r =1.$ For $r\in \set{2,\cdots,k},$ the proof \cref{ineq:inductive hypothesis for high probability bound} is analogous: we only need to replace $\E_{\sigma}{\E{\sum_{j=1}^k\tilde{T}^{(j)}}} \leq C_3$ with $\E_{\tilde{\sigma}}{\E{\sum_{j=r}^k\tilde{T}^{(j)} | Z_{<r} = x} }\leq C_3,$ where $\tilde{\sigma}$ is the substring of the random permutation $\sigma$ corresponds to the set $S^{(r)}, \cdots, S^{(k)} .$ 
\end{proof}

\begin{remark}\label{remark:approximate autoregression oracle}
Suppose we do not have access to the exact conditional distributions $\pi_i(\cdot \mid X_S)$, but instead to distributions $p_{i,S}$ satisfying 
\begin{equation}\label{eq:approximate coordinate denoiser assumption worst case}
    \forall x_S: d_{\mathrm{TV}}\big(\pi_i(\cdot \mid x_S),\, p_{i,S} (\cdot| x_S) \big) \leq \tilde{\epsilon}_{\mathrm{TV}}.
\end{equation}
Assuming $\tilde{\epsilon}_{\mathrm{TV}}\, n \leq \epsilonTV$, oracle access to $\{p_{i,S}\}_{i,S}$ allows us to construct approximate sampling algorithm(s) for $\pi$ whose output distribution is $\epsilon_{\mathrm{TV}}$-close to $\pi$. Assuming in addition that $ \tilde{\epsilon}_{\mathrm{TV}} \leq n^{-3/2}$, these algorithms achieve the same guarantees on the number of rounds and oracle queries as stated in \cref{thm:autoregression main detail,thm:autoregression with high probability}.

Indeed, in the coordinate denoiser scheme, after sampling the random permutation $\sigma$, for each node $S =\set{a+1, \cdots, j},$
we simply replace the speculative distribution \[\nu_S(\cdot |Z_{<S}) = \bigotimes_{r=a+1}^b \pi(X_{\sigma(r)} | Z_{\sigma(1)}, \cdots, Z_{\sigma(a)} )\] with \[\hat{\nu}_S(\cdot |Z_{<S}) =\bigotimes_{r=a+1}^b p_{\sigma (r), \sigma([a])} (X_{\sigma(r)} | Z_{\sigma(1)}, \cdots, Z_{\sigma(a)} ) \] and the target distribution \[\mu_S(\cdot |Z_{<S}) = \bigotimes_{r=a+1}^b \pi(X_{\sigma(i)} | Z_{\sigma(1)}, \cdots,  Z_{\sigma(a)}, X_{\sigma(a+1)}, \cdots , X_{\sigma(i-1)} )\] with  \[ \hat{\mu}_S (\cdot|Z_{<S}) = \bigotimes_{r=a+1}^b p_{\sigma (i), \sigma([i-1])}(X_{\sigma(i)} | Z_{\sigma(1)}, \cdots,  Z_{\sigma(a)}, X_{\sigma(a+1)}, \cdots , X_{\sigma(i-1)} ) .\]

The number of queries is unaffected. To bound the number of rounds, we only need to verify that
   \begin{align*}
           &\E*_{\sigma}{ \E*_{Z\sim \hat{\mu}_{[n]}} {\sum_{S_\ell \in \mathcal{C}_\ell (\textbf{root})}  \dTV*{\hat{\mu}_{S_\ell}\parens*{\cdot\given Z_{<S_\ell}},\hat{\nu}_{S_\ell}\parens*{\cdot\given Z_{<S_\ell}}}}} 
           \leq O(\sqrt{n \log q}) 
         \end{align*}
which is true since by triangle inequality for the total variation distance, we have:
\begin{align*}
           &\E*_{\sigma}{ \E*_{Z\sim \hat{\mu}_{[n]}} {\sum_{S_\ell \in \mathcal{C}_\ell (\textbf{root})}  \dTV*{\hat{\mu}_{S_\ell}\parens*{\cdot\given Z_{<S_\ell}},\hat{\nu}_{S_\ell}\parens*{\cdot\given Z_{<S_\ell}}}}}\\ 
           \leq  &\E*_{\sigma}{ \E*_{Z\sim \mu} {\sum_{S_\ell \in \mathcal{C}_\ell (\textbf{root})}  \dTV*{\hat{\mu}_{S_\ell}\parens*{\cdot\given Z_{<S_\ell}},\hat{\nu}_{S_\ell}\parens*{\cdot\given Z_{<S_\ell}}}}} + \tilde{\epsilon}_{TV} n^2 \\
           \leq &\E*_{\sigma}{ \E*_{Z\sim \mu} {\sum_{S_\ell \in \mathcal{C}_\ell (\textbf{root})}  \dTV*{\mu_{S_\ell}\parens*{\cdot\given Z_{<S_\ell}},\nu_{S_\ell}\parens*{\cdot\given Z_{<S_\ell}}}}} + \tilde{\epsilon}_{TV} n^2 + O(\tilde{\epsilon}_{TV} n) \\
           \leq &O(\sqrt{n \log q})
         \end{align*}
where the first inequality is by a simple change of measure inequality $\E_{P} {U} \leq \E_{\tilde{P}} {U} + d_{TV}(P, \tilde{P})\sup_{u\sim P} \abs{u}$ and the facts that $d_{TV}(\mu, \hat{\mu}_{[n]}) \leq \tilde{\epsilon}_{TV} n,$ and $\sum_{S_\ell \in \mathcal{C}_\ell (\textbf{root})}  \dTV*{\hat{\mu}_{S_\ell}\parens*{\cdot\given Z_{<S_\ell}},\hat{\nu}_{S_\ell}\parens*{\cdot\given Z_{<S_\ell}}}\leq n ,$ the second is by triangle inequality, and the third inequality is by \cref{eq:main potential bound for autoregression} and the assumption $\tilde{\epsilon}_{TV} \leq n^{-3/2}.$
 \end{remark}

%% file: corrected-improved-diffusion-polylog.tex
\section{Parallelizing denoising diffusion}\label{sec:diffusion-models}
In this section, we instantiate and analyze $\RS{}$ with respect to the Gaussian denoiser.

We recall the stochastic localization process with respect to a distribution $\pi:\R^n \to \R_{\geq 0}$, 
\begin{equation}\label{eq:stochastic localization}
 \d{\bar{X}_t} = f(t,\bar{X}_t) \d{t} + \d{B_t},
\end{equation}
where
$f(t,x)=\E*_{Y \sim \mu, g\sim \Normal{0, tI}}{Y \given tY+g=x}$ is the Gaussian denoiser. % 
Let $\pi_t(x)$ be the distribution of $U$ that is sampled from $ \pi $ conditioned on $ t U + \Normal*{0,tI} =x,$ and let $\Sigma_t(\cdot)=\text{cov}(\pi_t(\cdot)).$

A standard Ito calculus computation yields:
\[d \Sigma_t = - \Sigma_t^2 dt + \text{martingale} \quad \text{and} \quad d f(t,\cdot) = \Sigma_t dB_t \]
which implies the following identity:
\begin{equation}\label{eq:ito mean difference}
\forall t_1 \geq t_0: \E*{\norm{f(t_1 ,\bar{X}_{t_1}) - f(t_0 ,\bar{X}_{t_0})}^2 \given \bar{X}_{t_0}} =   \E*{\int_{t_0}^{t_1} \tr{\Sigma_t^2(\bar{X}_t) } dt \given \bar{X}_{t_0}} = \E*{\tr{\Sigma_{t_0} (\bar{X}_{t_0})} - \tr{\Sigma_{t_1} (\bar{X}_{t_1})} \given \bar{X}_{t_0}}  
\end{equation}

We consider the Euler discretization of \cref{eq:stochastic localization} with endpoints $t_0 \leq t_1 \cdots\leq t_N.$
For $t \in [t_r, t_{r+1})$ , let $t- = t_r$, then the Euler discretization is defined by 
\begin{equation} \label{eq:discretize sde}
    \d X_t = f(t-, X_{t-}) \d t + \d B_t.
\end{equation}

We will prove the following:
\begin{theorem}\label{thm:full denoising polylog}
Consider a distribution $\pi:\R^n \to \R_{\geq 0}$ where $\supp(\pi)\subseteq B(\E_\pi{X},R)$ i.e. $\sup_{x\in \supp(\pi)} \norm{x- \E_\pi{X}} \leq R.$ For any $\delta, \epsilonTV> 0,$ there exists an algorithm that, given oracle query access to the Gaussian denoiser $\set{f(t,\cdot)},$  outputs a sample from a distribution that is $\epsilonTV$-close to $\pi\ast\mathcal{N}(0,\delta \cdot I )$ in total variation distance. The expected number of rounds is $O(\sqrt{n} \log^4 \parens*{\frac{n R}{\epsilonTV \delta}})$ rounds and the expected number of queries is $O(\max\set{\frac{n}{\epsilonTV^2}, \frac{R^2}{\delta}, \frac{R^4}{\delta^2} }\log^2 \parens*{\frac{n R}{\epsilonTV \delta}}).$
\end{theorem}
\begin{remark}\label{remark:implementation details gaussian denoiser}
We verify \cref{assumption:formal} and discuss the implementation of our algorithm in the PRAM model. We recall the definitions of $\nu_S\parens{\cdot \given Z_{<S}},\mu_S\parens{\cdot \given Z_{<S}} $ from \cref{sec:spec-scheme-coordinate}. Given access to the density and sampling query as defined in the beginning of this section, we can sample from $ \nu_S\parens{\cdot \given Z_{<S}},$ and compute the probabilities $ \nu_S\parens{Z_S \given Z_{<S}}, \mu_S\parens{\cdot \given Z_{<S}}$ in $1$ round, using $O(|S|)$ Gaussian denoiser queries, $O(\log |S|) = O(\log n)$ parallel time, and $O(|S|)$ operations in the PRAM model\footnote{We make a minor modification to \cref{alg:main algo gaussian} so that $\RS{S=\set{a+1,\ldots,b}, Z_{<S}}$ takes $Z_1 + \cdots + Z_a$ as an additional input and returns $Z_1 + \cdots + Z_b$ as an additional output.}, under the convention that adding two \(n\)-dimensional vectors counts as one operation (rather than \(O(n)\)). Thus, \cref{thm:full denoising polylog} can be implemented in the PRAM model where the parallel runtime matches the round complexity up to poly-logarithmic factors and the number of operations matches the query complexity. 
\end{remark}
The following lemma bounds the distance between the target and speculative distributions in the Gaussian denoiser scheme.
\begin{lemma}\label{lem:direct kl divergence bound}
Fix $k < j.$ 
    Let $ (X_t)_{t_k\leq t \leq t_j}$ be the process defined by the \cref{eq:discretize sde} conditioned on $X_{t_k} = x_{t_k}$. Let $(\tilde{X}_t)_t$ evolve according to the constant-drift SDE with drift $f(t_k, x_{t_k}).$ Let $\delta_i = t_{i+1} -t_i.$
    Then
    \begin{align*}
        \DKL*{ X_{t_k},\cdots, X_{t_j} \river \tilde{X}_{t_k},\cdots, \tilde{X}_{t_j}} \leq \frac{1}{2}\E*{\sum_{i=k}^{j-1} \delta_i \norm*{f(t_i,X_{t_i}) - f(t_k,x_{t_k}) }^2 } 
    \end{align*}
\end{lemma}
\begin{proof}
   Note that, conditioned on $X_{t_i}=x_{t_i}, $ $X_{t_{i+1}}$ is distributed according to $ x_{t_i} + \delta_i f(t_i,x_{t_i}) + E_i$ where $E_i\sim \mathcal{N}(0, \delta_i)$ and $E_i$ is independent of $ x_{t_i}.$ Thus, letting $X = (X_{t_k},\cdots, X_{t_j})$ and $\tilde{X} = (\tilde{X}_{t_k},\cdots, \tilde{X}_{t_j})$ , we can write
\begin{equation}
\begin{split}
    \DKL*{X\river \tilde{X} } &= \E*{-\sum_{i=k}^{j-1} \frac{1}{2\delta_i} \parens*{\norm{X_{t_{i+1}} - X_{t_i} - \delta_i f(t_i,X_{t_i})}^2   
  -  \norm{X_{t_{i+1}} - X_{t_i} - \delta_i f(t_k,x_{t_k})}^2  } } \\
  &= \E*{-\sum_{i=k}^{j-1} \langle X_{t_{i+1}} - X_{t_i} - \delta_i f(t_i,X_{t_i}), f(t_i,X_{t_i}) -f(t_k,x_{t_k}) \rangle + \frac{1}{2}\sum_i \delta_i \norm{f(t_i,X_{t_i}) - f(t_k,x_{t_k}) }^2 }\\
  &=\frac{1}{2}\E*{\sum_{i=k}^{j-1} \delta_i \norm{f(t_i,X_{t_i}) - f(t_k,x_{t_k}) }^2 }
\end{split}
\end{equation}
\end{proof}
In an idealized setting, when the discretized SDE \eqref{eq:discretize sde} faithfully follows the stochastic localization SDE \eqref{eq:stochastic localization}, the total distance between the target and speculative distributions across all nodes of the fallback tree can be bounded via the following proposition.
\begin{proposition}\label{prop:mean difference bound via Ito formula}
    Let $ (\bar{X}_t)_{t}$ be the process defined by the stochastic localization SDE \eqref{eq:stochastic localization}.
    Fix $t_0 < t_1 < \cdots < t_N .$
    Let $\delta_j = t_{j+1}-t_j.$
    Then
    \begin{align*}
        \E*{\sum_{i=k}^{j-1}  \delta_i \norm{f(t_i,\bar{X}_{t_i}) - f(t_k,x_{t_k}) }^2 \given  \bar{X}_{t_k} = x_{t_k} } 
        &\leq (t_{j}- t_k) 
       \E*{ 
  \int_{t_k}^{t_j}\tr{\Sigma_{s}^2 (\bar{X}_s)} \d s \given \bar{X}_{t_k} = x_{t_k}}\\
  &= (t_j - t_k)
 \E*{\tr{\Sigma_{t_k}(x_{t_k})} -\tr{ \Sigma_{t_j} (\bar{X}_{t_j})}    \given \bar{X}_{t_k} = x_{t_k}}
    \end{align*}
    Consider a subsequence $k_0 < k_1 <\cdots < k_M = t_N $ of the sequence $t_0 < t_1 <\cdots < t_N ,$ where $k_0 = t_0$ and $k_M=t_N.$ 
    For $i\in\set{0,\cdots, M-1},$ let $H_i:=\set{j\in \set{0,\cdots, N } \given k_{i}\leq t_j \leq k_{i+1} -1}.$ We have:
    \[ \E*{\sum_{i=0}^{M-1} \sqrt{\E*{\sum_{j\in H_i} \delta_j\norm{f(t_j,\bar{X}_{t_j}) - f(k_i,\bar{X}_{{k_i}})}^2 \given \bar{X}_{{k_i}} } } \given \bar{X}_{t_0} = x_{t_0} } \leq \sqrt{(t_N-t_0) \cdot \tr{\Sigma_{t_0} (x_{t_0})}} \]
\end{proposition}
 \begin{proof}
By \cref{eq:ito mean difference}, we have:
     \begin{align*}
    \E*{\sum_{i=k}^{j-1}  \delta_i \norm{f(t_i,\bar{X}_{t_i}) - f(t_k,x_{t_k}) }^2 }  &= \E*{\sum_{i=k}^{j-1} \delta_i \int_{t_k}^{t_i}\tr {\Sigma_{s}^2 (\bar{X}_s)} \d s \given \bar{X}_{t_k} = x_{t_k}}\\ 
    &= \E*{\int_{t_k}^{t_{j-1}} \parens*{\sum_{i:t_i \geq s} \delta_i }\tr{\Sigma_{s}^2 (\bar{X}_s)} \d s \given \bar{X}_{t_k} = x_{t_k} }\\
    &\leq (t_j - t_k)\E*{ 
\int_{t_k}^{t_j}\tr{\Sigma_{s}^2 (\bar{X}_s)} \d s \given \bar{X}_{t_k} = x_{t_k}} \\
 &=  (t_j - t_k)\E*{ 
 \tr{\Sigma_{t_k}(x_{t_k})} -\tr{\Sigma_{t_j} (\bar{X}_{t_j})}   \given \bar{X}_{t_k} = x_{t_k}}
\end{align*}
The second statement is obtained by combining the first, Holder's inequality, and Jensen's inequality. 
Indeed,
\begin{align*}
    &\E*{\sum_{i=0}^{M-1} \sqrt{\E*{\sum_{j\in H_i} \delta_j\norm{f(t_j,\bar{X}_{t_j}) - f(k_i,\bar{X}_{{k_i}})}^2 \given \bar{X}_{{k_i}} } } \given X_{t_0} = x_{t_0} }\\
    \leq &\E*{\sum_{i=0}^{M-1} \sqrt{(k_{i+1} - k_i)\cdot \E*{ 
\int_{k_i}^{k_{i+1}}\tr{\Sigma_{s}^2 (\bar{X}_s)} \d s \given \bar{X}_{k_i}}  }  \given \bar{X}_{t_0} = x_{t_0}}\\
    \leq&\sqrt{t_N - t_0} \cdot \sqrt{\E*{\sum_{i=0}^{M-1}\int_{k_i}^{k_{i+1}} \tr{\Sigma_s^2 (\bar{X}_s)} \d s \given \bar{X}_{t_0}=x_{t_0} } }\\
    =&\sqrt{t_N - t_0} \cdot \sqrt{\tr{\Sigma_{t_0}(x_{t_0})} - \tr{\E{\Sigma_{t_N}(\bar{X}_{t_N})\given \bar{X}_{t_0} =x_{t_0} }}}\\
    \leq&\sqrt{(t_N - t_0 )\cdot\tr{\Sigma_{t_0}(x_{t_0})}}
\end{align*}
 \end{proof}

Under the mild assumption that the distribution $\pi$ has bounded support, a similar bound on the total distance between the target and speculative distributions holds, as given in the following proposition.
\begin{proposition}\label{lem:change of measure for bounded support}
Let $(\bar{X}_t)_t$ be the process defined by the stochastic localization SDE \eqref{eq:stochastic localization} and $(X_t)_t$ be the process defined by the discretized SDE \eqref{eq:discretize sde} with respect to the distribution $\pi$. Suppose $ \supp(\pi)\subseteq B(0,R). $ Let $\delta_i = t_{i+1} -t_i.$
  We have
  \begin{align*}
  &\frac{1}{2}\E{\sum_{i=k}^{j-1} \delta_i \norm{f(t_i,X_{t_i}) - f(t_k,x_{t_k}) }^2 \given X_{t_k} =x_{t_k} } \\
  &\leq (t_j -t_k) \parens*{  R^2 d_{TV}^2  ((\bar{X}_t |\bar{X}_{t_k}=x_{t_k})_{t_k\leq t\leq t_j}, (X_t |X_{t_k}=x_{t_k})_{t_k\leq t\leq t_j}) +  \E*{ 
\int_{t_k}^{t_j} \tr{ \Sigma_{s}^2(\bar{X}_{s})} ds   | \bar{X}_{t_k} = x_{t_k}} }  \\
&\leq (t_j -t_k) \parens*{\frac{R^2}{2} \cdot \sum_{i=k}^{j-1}\int_{t_i}^{t_{i+1}} \norm{f(t,\bar{X}_{t}) - f(t_i,\bar{X}_{t_i})}^2 dt  +  \E{\int_{t_k}^{t_j} \tr{\Sigma_{s}^2(\bar{X}_{s})} ds   | \bar{X}_{t_k} = x_{t_k}} }   \\
&\leq \frac{3}{2}(t_j -t_k)^2 R^4 
\end{align*}
\end{proposition}
\begin{proof}
We couple the trajectories $ (X_t)_{t_k \leq t\leq t_j}$ and $(\bar{X}_t)_{t_k \leq t\leq t_j}$ with initial condition $X_{t_k} =\bar{X}_{t_k} = x_{t_k}$ optimally. By triangle and Holder's inequality, we have 
\begin{align*}
    \norm{f(t_i,X_{t_i}) - f(t_k,x_{t_k}) }^2 &\leq (\norm{f(t_i,\bar{X}_{t_i}) - f(t_k,x_{t_k}) } + \norm{f(t_i,\bar{X}_{t_i}) - f(t_i,X_{t_i}) })^2 \\
    &\leq 2 (\norm{f(t_i,\bar{X}_{t_i}) - f(t_k,x_{t_k}) }^2 + \norm{f(t_i,\bar{X}_{t_i}) - f(t_i,X_{t_i}) }^2 )  
\end{align*}
Thus
\begin{equation}\label{eq:change of measure main}\begin{split}
    &\frac{1}{2}\E{\sum_{i=k}^{j-1} \delta_i \norm{f(t_i,X_{t_i}) - f(t_k,x_{t_k}) }^2 \given X_{t_k} =x_{t_k} }\\
    \leq & \E{\sum_{i=k}^{j-1} \delta_i (\norm{f(t_i,\bar{X}_{t_i}) - f(t_k,x_{t_k}) }^2 + \norm{f(t_i,\bar{X}_{t_i}) - f(t_i,X_{t_i}) }^2 ) \given \bar{X}_{t_k} =x_{t_k} } \\
    \leq &\underbrace{\E{\sum_{i=k}^{j-1} \delta_i \norm{f(t_i,\bar{X}_{t_i}) - f(t_k,x_{t_k}) }^2  \given \bar{X}_{t_k} =x_{t_k} }}_{A_1} + \underbrace{\E{\sum_{i=k}^{j-1} \delta_i  \norm{f(t_i,\bar{X}_{t_i}) - f(t_i,X_{t_i}) }^2 \given \bar{X}_{t_k} =x_{t_k} }}_{A_2}
    \end{split}
\end{equation}
By \cref{prop:mean difference bound via Ito formula}, 
the first term is bounded by 
\[A_1 \leq  (t_j -t_k) \E{\int_{t_k}^{t_j} \tr {\Sigma_{s}^2(\bar{X}_{s})} ds   | \bar{X}_{t_k} = x_{t_k}} \leq (t_j-t_k)^2 R^4  \]
where the second inequality is due to $ \tr{\Sigma_{s}^2(\bar{X}_{s})} 
\leq (\tr {\Sigma_{s}(\bar{X}_{s}))}^2 \leq R^4, $ since for any value of $\bar{X}_s,$ $\supp(\pi_s(\bar{X}_s))\subseteq\supp(\pi) \subseteq B(0,R).$ 
Next, we bound the second term.
Note that $ \supp(\pi_t(x))\subseteq \supp(\pi) \subseteq B(0,R) \forall t, x.$ Note also that $ f(t,x) = \E_{U\sim \pi_t(x)}{U},$ so we have: 
\[ \norm{f(t,x)- f(t,y)} = \norm{\E_{U\sim\pi_t(x)} {U} - \E_{U\sim\pi_t(y)} {U}} \leq d_{TV} (\pi_t(x), \pi_t(y) ) R \]
Suppose $t\in [t_k, t_j],$ $ x \sim \bar{X}_{t} |\bar{X}_{t_k}=x_{t_k}$ and $y \sim X_{t} |X_{t_k}=x_{t_k}.$ By the data processing inequality for TV distance, we can bound
\begin{align*}
    d_{TV}^2 (\pi_t(x), \pi_t(y) )&\leq d_{TV}^2((\bar{X}_t |\bar{X}_{t_k} =x_{t_k})_{t_k\leq t\leq t_j}, (X_t |X_{t_k} =x_{t_k})_{t_k\leq t \leq t_j}) \\
    &\leq \frac{1}{2}\E{\sum_{i=k}^{j-1}\int_{t_i}^{t_{i+1}} \norm{f(t,\bar{X}_{t}) - f(t_i,\bar{X}_{t_i})}^2 dt |\bar{X}_{t_k} = X_{t_k}}\\
    &\leq  \frac{1}{2} \E{\sum_{i=k}^{j-1} (t_{i+1}-t_i) (   \tr{\Sigma_{t_i} (\bar{X}_{t_i})} - \tr{\Sigma_{t_{i+1}} (\bar{X}_{t_{i+1}})} |\bar{X}_{t_k} = X_{t_k}} \\
    &\leq \frac{1}{2} (t_j-t_k) \E{ \tr{\Sigma_{t_k} (\bar{X}_{t_k})} -  \tr{\Sigma_{t_j} (\bar{X}_{t_j})}  |\bar{X}_{t_k} = X_{t_k}}\\
    &\leq \frac{1}{2} R^2 (t_j - t_k)
\end{align*}
where the second inequality is due to Pinsker's inequality and \cite[Lemma 9]{chen2022sampling}, the third inequality is due to \cref{prop:mean difference bound via Ito formula} and the last inequality is due to $ \tr{\Sigma_{t_k} (\bar{X}_{t_k})} -  \tr{\Sigma_{t_j} (\bar{X}_{t_k})} \leq \tr{\Sigma_{t_k} (\bar{X}_{t_k})}\leq R^2 .$ 

Hence,
\begin{align*}
    A_2 &\leq \E{\sum_{i=k}^{j-1} \delta_i  R^2 d_{TV}^2((\bar{X}_t |\bar{X}_{t_k} =x_{t_k})_{t_k\leq t\leq t_j}, (X_t |X_{t_k} =x_{t_k})_{t_k\leq t \leq t_j})  \given \bar{X}_{t_k} =x_{t_k} }\\
    & \leq R^2 (t_j-t_k) d_{TV}^2((\bar{X}_t |\bar{X}_{t_k} =x_{t_k})_{t_k\leq t\leq t_j}, (X_t |X_{t_k} =x_{t_k})_{t_k\leq t \leq t_j}) \\
    &\leq \frac{R^2}{2} \cdot \sum_{i=k}^{j-1}\int_{t_i}^{t_{i+1}} \norm{f(t,\bar{X}_{t}) - f(t_i,\bar{X}_{t_i})}^2 dt  \\
    &\leq \frac{1}{2} R^4 (t_j-t_k)^2 
\end{align*}
The result is by combining the bounds on $A_1$, $A_2$ and \cref{eq:change of measure main}.
\end{proof}

Before proving \cref{thm:full denoising polylog}, we show an algorithm that exactly simulate the discretized process \eqref{eq:discretize sde}.
\begin{theorem} \label{thm:parallelizing diffusion using only score function}
Let $(\bar{X}_t)_{t_0 \le t \le t_0 + \Delta}$ be the process defined by the stochastic localization SDE \eqref{eq:stochastic localization} with respect to the distribution $\pi.$  Suppose $\supp(\pi)\subseteq B(\E_\pi{X},R)$ i.e. $\sup_{x\in \supp(\pi)} \norm{x- \E_\pi{X}} \leq R.$ Let $(X_t)_{t_0 \le t \le t_0 + \Delta}$ be the process defined by the discretized SDE \eqref{eq:discretize sde} with endpoints $t_0 \leq t_1 \leq \cdots \leq t_N = t_0 + \Delta.$ Let
\begin{equation}\label{eq:kl continuous discrete bounded}
\hat{E}: = \E{\sum_{i=0}^{N-1} \int_{t_i}^{t_{i+1}} \norm{f(t,\bar{X}_{t}) - f(t_i,\bar{X}_{t_i})}^2 dt  | \bar{X}_{t_0} = x_{t_0}  }
\quad \text{and} \quad \Phi:= \sqrt{\hat{E}} R^2 \Delta + \sqrt{  \hat{E}\Delta R^2 +  \Delta \tr{\Sigma_{t_0}(x_{t_0})}}
\end{equation}

 There exists an algorithm, that given oracle query access to the Gaussian denoiser $\set{f(t_i,\cdot)}_{0\leq i\leq N}$, outputs a sample exactly distributed according to the distribution $ (X_{t_0, t_1, \cdots, t_N} | X_{t_0} = x_{t_0}) .$ The expected number of rounds is $O(\Phi \log^3 N)$ and the expected number of queries is $O(N \log N).$
\end{theorem}
 
\begin{proof}[Proof of \cref{thm:parallelizing diffusion using only score function}]
W.l.o.g.\ assume $ \E_{\pi}{X} = 0;$ otherwise we can query $ \E_{\pi}{X} = f(0,0)$ and shift the distribution $\pi$ by its mean $\E_\pi{X}$. Hence, from this point onward, we assume $\supp(\pi)\subseteq B(0,R).$

W.l.o.g.\ assume $N$ is a power of $2;$ otherwise we can replace $N$ with $N' = 2^{\lceil  \log_2 N \rceil}$ by extending the sequence $(t_0,\cdots, t_N)$ to $ (t_0, \cdots, t_N, t_{N+1}, \cdots, t_{N'})$ where $ t_{k} = t_N$ for $ N<k \leq N'.$
Let the fallback tree be the full binary tree with $N $ leaf nodes and height $h = \log_2 N \in \N_{\geq 1}.$

Our algorithm runs $\RS{}$ with respect to the Gaussian denoiser scheme at the root of the above fallback tree, with parameter $\rho = 1/h.$ Let $Q,T$ be its number of queries and rounds. By \cref{lem:expected number of queries}, $\E{Q} \leq O(N \log N).$ 

 We show that, for any $0\leq \ell \leq h-1:$
 \begin{equation}\label{eq:bound expected TV under good coupling event}
     \E_{Z\sim \mu}{\sum_{S_\ell \in \mathcal{C}_\ell(\textbf{root})}\dTV*{\mu_{S_\ell}\parens*{\cdot\given Z_{<S_\ell}},\nu_{S_\ell}\parens*{\cdot\given Z_{<S_\ell}}} } \leq O(\Phi)
 \end{equation}
 Given \cref{eq:bound expected TV under good coupling event}, we use \cref{lem:expected round bound} (with $D =2$ and $\rho = 1/h$)  to bound the expected number of rounds as follows:
 \begin{equation}\label{eq:expected round bound under good coupling event}
    \E{ T } \leq O( 1+2\sum_{\ell=0}^h (1+1/h)^\ell ( 2\Phi+ \frac{1}{\log (1+1/h)}\cdot ((1+\log N) \Phi + 1/e) ) ) = O(\Phi \log^3 N) 
 \end{equation}

 Finally, we verify \cref{eq:bound expected TV under good coupling event}. Fix an arbitrary level $\ell.$

Let the nodes at that level be $ \set{i_0+1, \cdots, i_1}, \set{i_1+1, \cdots, i_2}, \cdots, \set{i_{M-1} +1,i_M}$ and let $ k_0 = t_{i_0}, k_1=t_{i_1}, \cdots, k_M=t_{i_M}.$ Note that $ t_0 =k_0 $ and $t_N = k_M,$ and $ k_0, \cdots, k_M$ is a sub-sequence of $t_0, \cdots, t_N.$  

For $j\in [N]$, let $\delta_j = t_{j+1}-t_j.$
By Pinsker inequality and \cref{lem:direct kl divergence bound}
\begin{align*}
     &\E_{Z\sim \mu}{\sum_{S_\ell \in \mathcal{C}_\ell(\textbf{root})}\dTV*{\mu_{S_\ell}\parens*{\cdot\given Z_{<S_\ell}},\nu_{S_\ell}\parens*{\cdot\given Z_{<S_\ell}}} }\\ \leq &\E*_{X_{k_1}, \cdots, X_{k_{M-1}}}{\sum_{i=0}^{M-1} \sqrt{\E*{\sum_{j\in H_i} \delta_j\norm{f(t_j,X_{t_j}) - f(k_i,X_{{k_i}})}^2 \given X_{{k_i}} } } \given X_{t_0} = x_{t_0} }
\end{align*}

For $i\in \set{0,\cdots, M-1},$  let $H_i$ be the set of $j\in \set{0,\cdots, N}$ s.t. $ k_i \leq t_j \leq k_{i+1}-1$ and let 
\[E_{i}(z) := \sqrt{\E*{\sum_{j\in H_i} \delta_j\norm{f(t_j,X_{t_j}) - f(k_i,X_{{k_i}})}^2 \given X_{k_i} = z } }.\]

Let $z^{(0)} \equiv x_{t_0}.$
Our goal is to prove
 \begin{equation}\label{eq:main equation to bound potential for denoising diffusion}
     \E*_{z^{(1)}, \cdots, z^{(M-1)} \sim (X_{k_{1}, \cdots, k_{M-1}} | X_{t_{0}}  =x_{t_0}) } { \sum_{i=0}^{M-1} E_{i} (z^{(i)}) } \leq O(\Phi)
 \end{equation}
 
 By \cref{lem:change of measure for bounded support}, for any $ z^{(0)}, \cdots, z^{(M-1)},$ 
 \[ 0\leq \sum_{i=0}^{M-1} E_{i} (z^{(i)}) \leq O(\sum_{i=0}^{k-1} R^2 (k_{i+1} -k_i)) = O(R^2 \Delta)   \]
Using the simple change of measure inequality $\E_{P} {U} \leq \E_{\tilde{P}} {U} + d_{TV}(P, \tilde{P})\sup_{u\sim P} \abs{u} ,$ we have
 \begin{equation}\label{eq:denoising diffusion change of measure}
 \begin{split}
   \E*_{z^{(1)}, \cdots, z^{(M-1)} \sim (X_{k_{1}, \cdots, k_{M-1}} | X_{t_{0}}  =x_{t_0}) } { \sum_{i=0}^{M-1} E_{i} (z^{(i)}) } &\leq \E*_{z^{(1)}, \cdots, z^{(M-1)} \sim (\bar{X}_{k_{1}, \cdots, k_{M-1}} | \bar{X}_{t_{0}}  =x_{t_0}) } { \sum_{i=0}^{M-1} E_{i} (z^{(i)})}  \\
   &+ d_{TV} ((\bar{X}_{k_{1}, \cdots, k_{M-1}} | \bar{X}_{t_{0}}  =x_{t_0}),(X_{k_{1}, \cdots, k_{M-1}} | X_{t_{0}}  =x_{t_0}) ) O(R^2 \Delta)
   \end{split}
 \end{equation}
 The second term on the rhs is bounded by $O(\sqrt{\hat{E}} R^2 \Delta)$ since $ d_{TV} ((\bar{X}_{k_{1}, \cdots, k_{M-1}} | \bar{X}_{t_{0}}  =x_{t_0}),((X_{k_{1}, \cdots, k_{M-1}} | X_{t_{0}}  =x_{t_0})) )\leq \sqrt{\hat{E}}.$ Indeed, by Pinsker inequality, \cite[Theorem 9]{chen2022sampling}, and the definition of $\hat{E}$, we have:
 \begin{align*}
 d_{TV} ((\bar{X}_{k_{1}, \cdots, k_{M-1}} | \bar{X}_{t_{0}}  =x_{t_0}),(X_{k_{1}, \cdots, k_{M-1}} | X_{t_{0}}  =x_{t_0}) ) &\leq \sqrt{\DKL{(\bar{X}_{t} | \bar{X}_{t_{0}}  =x_{t_0})|| (X_{t} | X_{t_{0}}  =x_{t_0})}} \\
 &\leq \sqrt{\E{\sum_{i=0}^{N} \int_{t_i}^{t_{i+1}} \norm{f(t,\bar{X}_{t}) - f(t_i,\bar{X}_{t_i})}^2 dt  | \bar{X}_{t_0} = x_{t_0}  }}  = \sqrt{\hat{E}} . 
 \end{align*}
 To finish the proof, we bound the first term on the RHS of \cref{eq:denoising diffusion change of measure} by $O\parens*{\sqrt{\Delta R^2 \hat{E} +  \Delta \tr{\Sigma_{t_0} (x_{t_0})  } }  }.$ 

Again, by \cref{lem:change of measure for bounded support} and Holder's inequality, we have:
 \begin{align*}
     \sum_{i=0}^{M-1} E_{i} (z^{(i)}) 
     \leq &  O\parens*{\sum_{i=0}^{M-1} \sqrt{ (k_{i+1} -k_{i})  \parens*{R^2 \E{\sum_{j\in H_i}\int_{t_j}^{t_{j+1}} \norm{f(t,\bar{X}_{t}) - f(t_j ,\bar{X}_{t_j})}^2 dt + 
 \int_{k_i}^{k_{i+1}} \tr {\Sigma_{s}^2(\bar{X}_{s})} ds   | \bar{X}_{k_i} = z^{(i)} } }} }
 \\
 \leq & O\parens*{ \sqrt{ (\sum_{i=0}^{M-1} (k_{i+1} -k_{i}) )  \cdot \parens*{ \sum_{i=0}^{M-1}  R^2\E*{\sum_{j\in H_i}\int_{t_j}^{t_{j+1}} \norm{f(t,\bar{X}_{t}) - f(t_j ,\bar{X}_{t_j})}^2 dt + \int_{k_{i}}^{k_{i+1}} \tr {\Sigma_{s}^2(\bar{X}_{s})} ds   | \bar{X}_{k_{i}} = z^{(i)}} } } }
 \end{align*}
 The above together with Jensen's inequality gives
 \begin{align*}
 & \E*_{z^{(1)}, \cdots, z^{(M-1)} \sim (\bar{X}_{k_{1}, \cdots, k_{M-1}} | \bar{X}_{t_{0}}  =x_{t_0}) } { \sum_{i=0}^{M-1} E_{i} (z^{(i)})} \\
 \leq &O\parens*{\sqrt{ \Delta \cdot \E*{ \sum_{i=0}^{k-1}  \parens*{ R^2 \sum_{j\in H_i}\int_{t_j}^{t_{j+1}} \norm{f(t,\bar{X}_{t}) - f(t_j ,\bar{X}_{t_j})}^2 dt + \int_{k_{i}}^{k_{i+1}} \tr{ \Sigma_{t}^2(\bar{X}_{t})} dt }  | \bar{X}_{t_0} = x_{t_0}}     } }  \\
 = & O\parens*{\sqrt{\Delta \cdot \parens*{R^2  \E{\sum_{j=0}^{N} \int_{t_j}^{t_{j+1}} \norm{f(t,\bar{X}_{t}) - f(t_j ,\bar{X}_{t_j})}^2 dt  | \bar{X}_{t_0} = x_{t_0}  } +  \E{ \int_{t_0}^{t_N} \tr{ \Sigma_{t}^2(\bar{X}_{t})} dt | \bar{X}_{t_0} = x_{t_0}  }  } } }\\
 = & O\parens*{\sqrt{\Delta R^2 \hat{E} +  \Delta \tr{\Sigma_{t_0} (x_{t_0})  } }  } 
 \end{align*}
 where the last equality is due to \cref{eq:ito mean difference}.
\end{proof}

\begin{proof}[Proof of \cref{thm:full denoising polylog}]
W.l.o.g.\ assume $ \E_{\pi}{X} = 0;$ otherwise we can query $ \E_{\pi}{X} = f(0,0)$ and shift the distribution $\pi$ by its mean $\E_\pi{X}$. Hence, from this point onward, we assume $\supp(\pi)\subseteq B(0,R).$

Let $(\bar{X}_t)_t$ be the process defined by the stochastic localization SDE \eqref{eq:stochastic localization} with respect to the distribution $\pi.$
Let $ \Sigma_t = \Sigma_t(\bar{X}_t).$
Recall that
\[ d \Sigma_t = - \Sigma_t^2 d t + \text{martingale},\]
which implies $d \E{\Sigma_t} = - \E{\Sigma_t^2} dt  \preceq -\E{\Sigma_t}^2 dt.$ Let $A_t = \E{\Sigma_t}$ and $B_t = A_t^{-1}$ then \[d A_t \leq -A_t^2 dt \quad \text{and} \quad d B_t \succeq -A_t^{-2} d A_t  \geq -A_t^{-2} (-A_t^2 d t) = dt.\]
This combined with $B_0 =\Sigma_0^{-1} \succeq R^{-1} \cdot I$\footnote{In the Loewner order} implies
$ B_t \succeq (R^{-1} + t)\cdot I.$ Equivalently,
\begin{equation}\label{eq:covariance trace bound}
   \E{\Sigma_t} \preceq \frac{1}{R^{-1}+t}\cdot I\Rightarrow \E{\text{tr}(\Sigma_t)} \leq \frac{n}{R^{-1}+t}  
\end{equation}
Let $L = \lceil\log_2 \frac{R}{\delta} \rceil +1.$
Let $\mathbf{T}_0 = 0$ and for $i\in \set{1, \cdots, L-1},$ let $ \mathbf{T}_i = R^{-1} 2^{i-1} $ and $ \mathbf{T}_L = 1/\delta.$ Discretize the interval $[\mathbf{T}_i, \mathbf{T}_{i+1}]$ into $N$ subintervals each of size $ \delta_i = \frac{T_{i+1}-T_i}{N},$ where $N\in \N_{\geq 1}$ is a power of $2$ to be chosen later.
Let the endpoints of this discretization be $ 0=\mathbf{T}_0 = t_0 <t_1 <\cdots < t_{N} = \mathbf{T}_1 = t_{N+1} < \cdots < t_{N L} = \mathbf{T}_L,$ and
let $ (X_t)_t$ be the resulting discretized process \eqref{eq:discretize sde}.
Note that for $i \in \set{0,\cdots,L-1}$
\begin{equation}\label{eq:trace ineq bound for denoising diffusion}
    (\mathbf{T_{i+1}} -\mathbf{T}_i) \E{\tr{\Sigma_{\mathbf{T}_i} (\bar{X}_{\mathbf{T}_i})) } }\leq \frac{n (\mathbf{T_{i+1}} -\mathbf{T}_i)}{(\mathbf{T}_i + R^{-1})} \leq n
\end{equation}

We consider the fallback tree where the root corresponds to the set $\set{1,\cdots, NL}.$ The root has $L$ children $S^{(1)}, \cdots, S^{(L)}$, which corresponds to the sets $ \set{1, \cdots, N}, \cdots , \set{N(L-1)+1, \cdots, NL}.$ For each $\ell\in \set{1,\cdots, L},$ the subtree rooted at $ S^{(\ell)}$ is a full binary tree with height $h =\log N \in \N$ and $N$ leaf nodes.

Our algorithm runs the $\fallback{}$ in \cref{alg:rs2} with respect the Gaussian denoiser scheme at the root of the fallback tree, with parameter $\rho = 1/h,$ then outputs $\frac{X_{\mathbf{T}_L}}{\mathbf{T}_L}$ as an approximate sample from $\pi.$ Let $Q, T$ be the number of queries and rounds. For $i\in \set{1,\cdots, L},$  let $T^{(i)}, Q^{(i)}$ be the number of rounds and queries used by the call to $\RS{}$ at the node $S^{(i)}.$ Clearly, $T =\sum_{i=1}^L T^{(i)}$ and  $Q = \sum_{i=1}^L Q^{(i)}. $

By \cite[Theorem 9]{chen2022sampling} \footnote{This can be viewed as a version of \cref{lem:direct kl divergence bound} when $\delta_i\to 0.$}, we can bound the KL divergence between the discretized and continuous processes in the entire interval $[0,\mathbf{T}_{L}]$ by
\begin{equation} \label{eq:KL divergence between continuous and discrete}
\hat{E}:=\E{\sum_{i=0}^{NL} \norm{f(t,\bar{X}_t) - f(t_i,\bar{X}_{t_i})}^2 dt }
\leq_{(1)} 
\E{\sum_{i=0}^{L-1} \delta_i \int_{\mathbf{T_i}}^{\mathbf{T_{i+1}}} \text{tr}(\Sigma_t^2 (\bar{X}_t )) dt } \leq \sum_{i=0}^{L-1} \delta_i\E{\text{tr}(\Sigma_{\mathbf{T_i}} (\bar{X}_{\mathbf{T_i}})) } \leq_{(2)} \sum_{i=0}^{L-1}  \frac{n}{N}\leq_{(3)} \epsilonTV^2 
\end{equation}
where (1) is due to \cref{prop:mean difference bound via Ito formula}, (2)
is due to \cref{eq:trace ineq bound for denoising diffusion}, and (3) is true when we choose $N \geq \frac{nL}{\epsilonTV^2}.$ 

By Pinsker's inequality, the discretized and continuous processes in the entire interval $[0,\mathbf{T}_{L}]$ are $\epsilonTV$-close in TV distance, so the output distribution is $\epsilonTV$-close to $\text{Law}(\frac{\bar{X}_{\mathbf{T}_L}}{{\mathbf{T}_L}}) = \pi\ast\mathcal{N}(0, \delta \cdot I).$

To bound the expected number of rounds, we need the following analog of \cref{eq:main equation to bound potential for denoising diffusion}.
Let $\Phi: = R^2\delta^{-1} \sqrt{\hat{E}} + \sqrt{\delta^{-1} R^2 \hat{E}}   + L\sqrt{n}$ where $\hat{E}$ is as defined in \cref{eq:kl continuous discrete bounded}. Fix a level $\ell$ where $1\leq\ell \leq h + 1.$ Let the nodes at that level be $ \set{i_0+1, \cdots, i_1}, \set{i_1+1, \cdots, i_2}, \cdots, \set{i_{M-1} +1,\cdots, i_M}$ and let $ k_0 = t_{i_0}, k_1=t_{i_1}, \cdots, k_M=t_{i_M}.$ 

Let $H_i =\set{j\in \set{0,\cdots, NL}: k_i \leq t_j \leq k_{i+1}-1},$ $z^{(0)} \equiv x_{t_0},$ and
\[E_{i}(z) := \sqrt{\E*{\sum_{j\in H_i} \delta_j\norm{f(t_j,X_{t_j}) - f(k_i,X_{{k_i}})}^2 \given X_{k_i} = z } }.\]
then
\begin{equation}\label{eq:main potential bound for full denoising polylog}
    \E*_{z^{(1)}, \cdots, z^{(M-1)} \sim (X_{k_{1}, \cdots, k_{M-1}} | X_{t_{0}}  =x_{t_0}) } { \sum_{i=0}^{M-1} E_{i} (z^{(i)}) } \leq O(\Phi)
\end{equation}
Given \cref{eq:main potential bound for full denoising polylog}, and use the bound $\hat{E} = \frac{nL }{N},$ we have  $\Phi = O(L \sqrt{n})$ when we choose $N \geq \max\set{\frac{R^2}{\delta}, \frac{R^4}{\delta^2}}.$ 
We let 
\[ N = 2^{\lceil \log_2  \max\set{\frac{R^2}{\delta}, \frac{R^4}{\delta^2}, \frac{n}{\epsilonTV^2}} \rceil }  = O(\max\set{\frac{R^2}{\delta}, \frac{R^4}{\delta^2}, \frac{n}{\epsilonTV^2}}).\]
then by \cref{lem:expected number of queries,lem:expected round bound}, the expected number of rounds and queries are bounded by:
\[ \E{ T } \leq O(\Phi \log^3 N) = O(\sqrt{n } \log^4 \parens*{\frac{n R}{\epsilonTV \delta}}) \quad \text{and} \quad \E{Q} = O(L N \log N)=  O( \max\set{\frac{R^2}{\delta}, \frac{R^4}{\delta^2}, \frac{n}{\epsilonTV^2}} \log^2 \parens*{\frac{n R}{\epsilonTV \delta}}  )\]
Finally, we prove \cref{eq:main potential bound for full denoising polylog}. The proof is similar to in \cref{thm:parallelizing diffusion using only score function}. 
By \cref{lem:change of measure for bounded support}, $\forall z^{(0)}, \cdots, z^{(M-1)}: \sum_{i=0}^{M-1} E_{i} (z^{(i)}) \leq O(\frac{R^2}{\delta}),$ thus:
 \begin{align*}
   \E*_{z^{(1)}, \cdots, z^{(M-1)} \sim (X_{k_{1}, \cdots, k_{M-1}} | X_{t_{0}}  =x_{t_0}) } { \sum_{i=0}^{M-1} E_{i} (z^{(i)}) } &\leq \E*_{z^{(1)}, \cdots, z^{(M-1)} \sim (\bar{X}_{k_{1}, \cdots, k_{M-1}} | \bar{X}_{t_{0}}  =x_{t_0}) } { \sum_{i=0}^{M-1} E_{i} (z^{(i)})}  \\
   &+ d_{TV} ((\bar{X}_{k_{1}, \cdots, k_{M-1}} | \bar{X}_{t_{0}}  =x_{t_0}),((X_{k_{1}, \cdots, k_{M-1}} | X_{t_{0}}  =x_{t_0})) ) O(\frac{R^2}{\delta} )\\
   &\leq  \E*_{z^{(1)}, \cdots, z^{(M-1)} \sim (\bar{X}_{k_{1}, \cdots, k_{M-1}} | \bar{X}_{t_{0}}  =x_{t_0}) }  { \sum_{i=0}^{M-1} E_{i} (z^{(i)})} + O(\frac{R^2 \sqrt{\hat{E}}}{\delta})\end{align*}
where the second inequality is due to \cref{eq:KL divergence between continuous and discrete}.
By the construction of the fallback tree,
for each $\ell\in \set{0,\cdots, L},$ there is a unique $ r_\ell\in [M]$ s.t. $k_{r_\ell} = \mathbf{T}_j.$
Next, we bound
\begin{align*}
     &\E*_{z^{(1)}, \cdots, z^{(M-1)} \sim (\bar{X}_{k_{1}, \cdots, k_{M-1}} | \bar{X}_{t_{0}}  =x_{t_0}) }  { \sum_{i=0}^{M-1} E_{i} (z^{(i)})} \\
    \leq  &O\parens*{ \E*{\sum_{i=0}^{M-1} \sqrt{ (k_{i+1} -k_{i})  \parens*{R^2 \E{\sum_{j\in H_i}\int_{t_j}^{t_{j+1}} \norm{f(t,\bar{X}_{t}) - f(t_j ,\bar{X}_{t_j})}^2 dt + 
 \int_{k_i}^{k_{i+1}} \tr {\Sigma_{s}^2(\bar{X}_{s})} ds   \given \bar{X}_{k_i} = z^{(i)} } }} } }\\
 \leq  &O\parens*{ \E*{\sum_{i=0}^{M-1} \sqrt{ R^2 (k_{i+1} -k_{i})   \E{\sum_{j\in H_i}\int_{t_j}^{t_{j+1}} \norm{f(t,\bar{X}_{t}) - f(t_j ,\bar{X}_{t_j})}^2 dt \given \bar{X}_{k_i} = z^{(i)} }}}} \\
 +  &O\parens*{ \E*{\sum_{i=0}^{M-1} \sqrt{(k_{i+1} -k_{i}) \E*{\int_{k_i}^{k_{i+1}} \tr {\Sigma_{s}^2(\bar{X}_{s})} ds\given \bar{X}_{k_i} = z^{(i)} }}}}\\
 \leq  &O\parens*{\sqrt{R^2 \hat{E} \delta^{-1}} } + \sum_{\ell=0}^{L-1} O\parens*{\E*{  \sum_{i \in [r_\ell, r_{\ell+1}-1 ] \cap \N}   \sqrt{(k_{i+1} -k_{i}) \E*{\int_{k_i}^{k_{i+1}} \tr {\Sigma_{s}^2(\bar{X}_{s})} ds\given \bar{X}_{k_i} = z^{(i)} }}}}\\
 \leq_{(1)}  &O\parens*{\sqrt{R^2 \hat{E} \delta^{-1}} } + \sum_{\ell=0}^{L-1} O(\sqrt{(\mathbf{T}_{i+1}-\mathbf{T}_i) \cdot \E{\tr{\Sigma_{\mathbf{T}_i}(\bar{X}_{\mathbf{T}_i})}}})\\
 \leq_{(2)} &O\parens*{\sqrt{R^2 \hat{E} \delta^{-1}} }+ O(\sqrt{n} L)
\end{align*}
where (1) is by Holder's inequality and \cref{eq:ito mean difference},
and (2) is due to \cref{eq:trace ineq bound for denoising diffusion}.
\end{proof}
\begin{remark}\label{remark:approximte gaussian denoiser} In the same setting as \cref{thm:full denoising polylog},
    we can also handle error in the Gaussian denoiser oracle i.e. assume oracle access to the approximate score function $s(t,x)$ where
    \begin{equation}\label{eq:approximate denoiser oracle}
        \forall t,x: \norm{s(t,x) - f(t,x)}^2 \leq \epsilon_{\text{score}}^2
    \end{equation}
    Set $N$ as in the proof of \cref{thm:full denoising polylog}, i.e. $N = O(\max\set{\frac{R^2}{\delta}, \frac{R^4}{\delta^2}, \frac{n}{\epsilonTV^2}}), $ and let \[ \hat{E} := \E{\sum_{i=0}^{NL} \norm{s(t,\bar{X}_t) - f(t_i,\bar{X}_{t_i})}^2 dt } \leq O(\frac{1}{\delta} \cdot \epsilon_{\text{score}}^2) + \frac{nL}{N} = O(\min\set{\epsilon^2_{TV}, \frac{nL}{N}}),\] when we assume
    \[\epsilon_{\text{score}}^2\leq  O(\delta \epsilonTV^2, \delta n/N).\] Under this assumption, we can ensure that the algorithm in \cref{thm:full denoising polylog} with $f(t,x)$ replaced by $s(t,x)$ outputs a sample that is $\epsilonTV$ close to $ \pi \ast \mathcal{N}(0,\delta \cdot I),$ where the expected number of rounds is bounded by 
    \[ O(\sqrt{ \frac{1}{\delta} \cdot \epsilon_{\text{score}}^2 \cdot N  } + \frac{R^2}{\delta}) \sqrt{\hat{E}} +  \sqrt{ \frac{1}{\delta} \cdot (\epsilon_{\text{score}}^2 \cdot N   +R^2) \hat{E}} + L \sqrt{n}) = O(L \sqrt{n})\]
\end{remark}

%% file: missing-proofs.tex
\section{Deferred proofs}\label{sec:misssing proofs}
\begin{proof}[Proof of \cref{prop:expected round helper}]
    Note that $\lim_{x\to 0^+} x \log \frac{1}{x} = 0,$  so $\tilde{\varphi} $ is continuous. In addition, $\tilde{\varphi}'(x) = -1 - \log x$  and $\tilde{\varphi}'(x)$ is strictly decreasing for $x\in (0,+\infty).$  Let $N = |\mathcal{C}_\ell(\textbf{root})|.$ Applying Jensen's inequality for the concave function $\tilde{\varphi}$ gives
    \begin{align*}
         \E*_{Z\sim \mu} {\sum_{S_\ell \in \mathcal{C}_\ell(\textbf{root})}\tilde{\varphi} \parens*{\dTV*{\mu_{S_\ell}\parens*{\cdot\given Z_{<S_\ell}},\nu_{S_\ell}\parens*{\cdot\given Z_{<S_\ell}}}} } &\leq N\E*_{Z\sim \mu}{ \tilde{\varphi} \parens*{ N^{-1} \sum_{S_\ell \in \mathcal{C}_\ell(\textbf{root})}\dTV*{\mu_{S_\ell}\parens*{\cdot\given Z_{<S_\ell}},\nu_{S_\ell}\parens*{\cdot\given Z_{<S_\ell}}} } }  \\
         &\leq  N \tilde{\varphi}\parens*{ N^{-1} \E*_{Z\sim \mu}{ \sum_{S_\ell \in \mathcal{C}_\ell(\textbf{root})} \dTV*{\mu_{S_\ell}\parens*{\cdot\given Z_{<S_\ell}},\nu_{S_\ell}\parens*{\cdot\given Z_{<S_\ell}}} } }\\
         &\leq N \tilde{\varphi}(x_*)
    \end{align*}
    where $ x_* = N^{-1} \E*_{Z\sim \mu}{ \sum_{S_\ell \in \mathcal{C}_\ell(\textbf{root})} \dTV*{\mu_{S_\ell}\parens*{\cdot\given Z_{<S_\ell}},\nu_{S_\ell}\parens*{\cdot\given Z_{<S_\ell}}} }.$

    Note that $\tilde{\varphi}$ is increasing on $[0,1/e],$ and decreasing on $ [1/e, 1],$ thus attain its maximum at $x=1/e.$  
    Let $\Phi:= \E*_{Z\sim \mu}{ \sum_{S_\ell \in \mathcal{C}_\ell(\textbf{root})} \dTV*{\mu_{S_\ell}\parens*{\cdot\given Z_{<S_\ell}},\nu_{S_\ell}\parens*{\cdot\given Z_{<S_\ell}}} }.$
    We have two cases:
    \begin{itemize}
        \item If $\frac{\Phi}{N}\leq \frac{1}{e}:$  Since $ x_* \leq \frac{\Phi}{N},$ $\tilde{\varphi}(x_*) \leq \tilde{\varphi}(\frac{\Phi}{N}) $ thus
        \[  \E*_{Z\sim \mu} {\sum_{S_\ell \in \mathcal{C}_\ell(\textbf{root})}\tilde{\varphi} \parens*{\dTV*{\mu_{S_\ell}\parens*{\cdot\given Z_{<S_\ell}},\nu_{S_\ell}\parens*{\cdot\given Z_{<S_\ell}}}} } \leq \Phi\log (\frac{N}{\Phi})  =\Phi\log N+ \tilde{\varphi}(\Phi) \leq \Phi \log N + \frac{1}{e}\]
        where the last inequality is due to $\tilde{\varphi}(\Phi) \leq \tilde{\varphi}(\frac{1}{e}) =\frac{1}{e}.$
        \item If $\frac{\Phi}{N}> \frac{1}{e}:$
        \[ \tilde{\varphi}(x_*) \leq \tilde{\varphi}(\frac{1}{e}) =\frac{1}{e}\quad \text{thus} \quad \E*_{Z\sim \mu} {\sum_{S_\ell \in \mathcal{C}_\ell(\textbf{root})}\tilde{\varphi} \parens*{\dTV*{\mu_{S_\ell}\parens*{\cdot\given Z_{<S_\ell}},\nu_{S_\ell}\parens*{\cdot\given Z_{<S_\ell}}}} } \leq  N \tilde{\varphi}(x_*) \leq \frac{N}{e} < \Phi  \]
    \end{itemize}
    Thus, in either cases, we have
    \[\E*_{Z\sim \mu} {\sum_{S_\ell \in \mathcal{C}_\ell(\textbf{root})}\tilde{\varphi} \parens*{\dTV*{\mu_{S_\ell}\parens*{\cdot\given Z_{<S_\ell}},\nu_{S_\ell}\parens*{\cdot\given Z_{<S_\ell}}}} }  \leq \Phi(\log N  +1) + \frac{1}{e}. \]
\end{proof}